\documentclass{article}
\usepackage[utf8]{inputenc}
\usepackage[a4paper, portrait, margin=1in]{geometry}

\usepackage{color,amsmath,amsfonts,fullpage,amsthm,mathrsfs,fancyhdr,verbatim,framed,hyperref,xspace,mdframed,tabularx,booktabs,enumitem,comment,float}
\usepackage[dvipsnames]{xcolor}

\usepackage{mathpazo}
\usepackage{wasysym}
\usepackage{endnotes}

\usepackage[linesnumbered,boxed]{algorithm2e}
\SetAlCapSkip{1em}
\SetAlCapNameFnt{\small}
\SetAlCapFnt{\small}
\SetAlgorithmName{Pseudocode}{}{}

\usepackage[capitalize]{cleveref}
\SetKwComment{Comment}{/* }{ */}

\mdfdefinestyle{figstyle}{ %
  linecolor=black!7, %
  backgroundcolor=black!7, %
  innertopmargin=10pt, %
  innerleftmargin=25pt, %
  innerrightmargin=25pt, %
  innerbottommargin=10pt %
}

\renewcommand{\cal}[1]{\mathcal{#1}}
\renewcommand{\sf}[1]{\mathsf{#1}}
\newcommand{\scr}[1]{\mathscr{#1}}
\newcommand{\eps}{\varepsilon}

\newcommand{\ket}[1]{{|#1\rangle}}
\newcommand{\bra}[1]{{\langle#1|}}

\newcommand{\tr}{\mbox{\rm tr}}
\newcommand{\Id}{\mathbf{1}}
\newcommand{\id}{\mathbf{1}}

\newcommand{\mH}{\mathcal{H}}

\newcommand{\ot}{\otimes}
\newcommand{\calD}{\mathcal{D}}

\newcommand{\calF}{\mathcal{F}}
\newcommand{\reg}[1]{{\mathsf{#1}}}

\renewcommand{\cal}[1]{\mathcal{#1}}

\newcommand{\MA}{\mathsf{MA}}
\newcommand{\NP}{\mathsf{NP}}
\newcommand{\AM}{\mathsf{AM}}
\newcommand{\MIP}{\mathsf{MIP}}
\newcommand{\NEXP}{\mathsf{NEXP}}
\newcommand{\RE}{\mathsf{RE}}
\newcommand{\QMA}{\mathsf{QMA}}
\DeclareMathOperator{\poly}{poly}

\newcommand{\val}{\ensuremath{\mathrm{val}}}

\newcommand{\game}{G}
\newcommand{\sampler}{S}
\newcommand{\decider}{D}

\newcommand{\strategy}{\mathscr{S}}

\newcommand{\hp}{{\quarternote}}
\newcommand{\hf}{{\halfnote}}

\newcommand{\pitwolang}{L}
\newcommand{\ptl}{\pitwolang}

\newtheorem{theorem}{Theorem}
\newtheorem{definition}[theorem]{Definition}
\newtheorem{fact}[theorem]{Fact}

\newtheorem{lemma}[theorem]{Lemma}

\newtheorem{remark}[theorem]{Remark}
\newif\ifnotes\notesfalse
\ifnotes
\newcommand{\znote}[1]{\textcolor{blue}{(Tina: #1)}}
\newcommand{\anote}[1]{\textcolor{red}{(Anand: #1)}}
\newcommand{\noteswarning}{{\begin{center} {\Large WARNING: NOTES
        ON}\endnote{Warning: notes on}\end{center}}}
\newcommand{\notesendofpaper}{{\theendnotes}}
\else
\newcommand{\znote}[1]{}
\newcommand{\anote}[1]{}
\newcommand{\noteswarning}{{}}
\newcommand{\notesendofpaper}{}
\fi
\begin{document}
\title{Quantum free games}

\newif\ifnames\namestrue
%\newif\ifnames\namesfalse
\ifnames
\author{Anand Natarajan\thanks{\texttt{anandn@mit.edu}}\\\small{MIT}
  \and Tina Zhang\thanks{\texttt{tinaz@mit.edu}}\\\small{MIT}
}
\else
\fi

%\affil{CSAIL, Massachusetts Institute of Technology}

\maketitle
\noteswarning

\begin{abstract}
The complexity of free games with two or more classical players was essentially settled by Aaronson, Impagliazzo, and Moshkovitz~\cite{AIM}. In the quantum world, there are two complexity classes that can be considered quantum analogues of classical free games: (1) $\AM^*$, the multiprover interactive proof class corresponding to free games with \emph{entangled} players, and, somewhat less obviously, (2) $\sf{BellQMA}(2)$, the class of quantum Merlin-Arthur proof systems with two unentangled Merlins, whose proof states are \emph{separately} measured by Arthur. In this work, we make significant progress towards a tight characterization of both of these classes.
\begin{enumerate}
    \item We show a $\mathsf{BellQMA}(2)$ protocol for 3SAT on $n$
      variables, where the total amount of communication is
      $\tilde{O}(\sqrt{n})$. This answers an open question of Chen and
      Drucker~\cite{CD10} and
      also shows, conditional on ETH, that the algorithm of Brand\~{a}o,
      Christandl and Yard~\cite{BCY} for optimizing over separable
      states is tight up to logarithmic factors. 
    \item We show that $\AM^*[n_{\text{provers}} = 2, q = O(1), a =
      \poly\log(n)] = \mathsf{RE}$, i.e. that free entangled games with
      constant-sized questions are as powerful as general entangled
      games. (In contrast, \cite{AIM} shows that \emph{classical} free
      games are much weaker than general classical games.)
      We show this using a question
      ``hyper-compression" theorem that iteratively applies the
      introspection technique of Ji et al.~\cite{JNVWY20}. Our result
      is a significant improvement over the headline result of Ji et
      al., whose $\MIP^*$ protocol for the halting problem has
      $\poly(n)$-sized questions and answers.
    \item By the same techniques, we obtain a \emph{zero-gap} $\AM^*$
      protocol for a $\Pi_2$ complete language with constant-size
      questions and almost logarithmically ($O(\log n \cdot \log^* n)$) large answers, improving on
      the headline result of Mousavi, Nezhadi and Yuen~\cite{MNY21}. 
    \item Using a connection to the nonuniform complexity of the halting
      problem we show that any $\MIP^*$ protocol
      for $\RE$ \emph{requires} $\Omega(\log n)$ bits of communication. It follows that our results in item 3 are optimal up to an
      $O(\log^* n)$ factor, and that the gapless compression
      theorems of \cite{MNY21} are asymptotically optimal. We
      conjecture that these bounds can be saturated in the gapped case
      as well.
    % \item We make progress towards a direct quantum generalization of the main result of~\cite{AIM}, namely that there is a protocol for 3SAT in $\AM(2)$ with $\tilde{O}(\sqrt{n})$ communication, obtained by combining birthday repetition with a classical PCP with quasilinear proof length. Specifically, we show that conditional on a ``quantum-sound quasilinear games PCP," there is an $\AM^*(2)$ protocol for $\AM$. We note that the strict classical analog of this statement, i.e. an $\AM(2)$ protocol for $\AM$ with $\tilde{O}(\sqrt{n})$ communication, is not known and seems unlikely to exist without a quasilinear-overhead derandomization of $\AM$. \anote{Our conditional result, if unconditionally true, would also be an improvement on Brand\~{a}o-Harrow.} We believe this result adds motivation for showing $\NP$-hardness for the entangled value of two-player games, a problem that frustratingly remains open despite the undecidability result of~\cite{mipstar}. We also note that it is conceivable that $\AM^*(2)$ with small communication could contain up to $RE$. \anote{This whole item is now outdated assuming that infinite question reduction works.}

    \end{enumerate}
% Our last three results reveal that free nonlocal games are essentially as powerful as general nonlocal games, in contrast to the results of \cite{AIM}, which show that classical free games are significantly simpler than general classical games.

\end{abstract}

\newpage

\tableofcontents

\newpage

\section{Introduction}

The 1991 work of Babai, Fortnow and Lund which showed that $\MIP = \NEXP$ \cite{BFL91} remains one of the most important achievements of complexity theory to date. The techniques used in the proof provided a springboard toward several other important results, including the proof that $\sf{PCP} = \NP$ \cite{AS98,ALMSS98} and, more recently, $\MIP^* = \RE$ \cite{JNVWY20}. In short, history shows that the study of \emph{multiplayer games}, in which an honest, polynomial-time verifier referees a game involving two or more potentially dishonest and unbounded provers, has yielded some of the most fruitful avenues of research in the field of complexity theory.

In a multiplayer game, the verifier---being computationally bounded---is inherently at a disadvantage, and must find clever ways to force the unbounded provers with which it interacts to do computations on its behalf, even though it cannot necessarily replicate those computations to check if they were done accurately. One of the most useful powers at its disposal is the ability to ask the provers correlated questions. For example, a common paradigm in multiprover proof design is the \emph{consistency test}, in which the verifier asks one prover (`Alice') to provide an answer to some small subproblem of the overall problem which it is trying to decide, and asks another prover (`Bob') to provide \emph{part of} the answer to the same subproblem. An example of this paradigm is the \emph{clause-variable game}, in which the verifier---who is attempting to decide whether some instance of a constraint satisfaction problem (CSP), such as 3SAT or graph colouring, is satisfiable---asks Alice to provide a satisfying assignment to a single constraint in that CSP, and asks Bob to provide an assignment to just one of the variables participating in that constraint. For example, if the verifier is trying to decide an instance of 3SAT, then it may ask Alice for a satisfying assignment to a single clause in the 3SAT formula, and ask Bob for an assignment to one of the variables in that clause. The verifier then checks that the assignment which Alice provided is indeed a satisfying assignment, and moreover that Bob's answer was consistent with Alice's assignment. Note that a single clause in a 3SAT formula always has a satisfying assignment, even if the entire formula does not. The key element that prevents Alice and Bob from exploiting this to convince the verifier that an unsatisfiable formula is satisfiable is the fact that they cannot communicate. Therefore, the only way Alice can be consistent with Bob is if they have agreed beforehand to answer consistently with a \emph{global} satisfying assignment to the 3SAT formula. Otherwise, if Alice reports assignments that depend on the clause which she was given by the verifier, she will not be consistent with Bob, because Bob does not know which clause she was given.

The crucial advantage which the consistency test paradigm enjoys over single-prover $\sf{NP}$-style verification, in which the verifier simply asks for an assignment which satisfies all the clauses in an instance of 3SAT and checks that this is the case, is that a consistency test allows the verifier to efficiently check the satisfiability of CSPs which have far more clauses than the verifier could efficiently read, if it were to read all of them. Specifically, in only polynomial time, the consistency test paradigm allows a polynomially-bounded verifier to check whether an \emph{exponentially} long CSP is satisfiable or far from satisfiable, since the verifier only needs to send polynomially long questions to Alice and Bob in order to specify which constraint and which variable it wants to know about in an exponentially long CSP. This is the basic but necessary verification framework at the heart of results such as $\MIP = \NEXP$.

We may then ask: what if we take away this power of the verifier, essential to the design of many multiprover interactive proof systems, which allows it to ask the provers correlated questions (and therefore perform tests like the consistency test)? How is the computational power of the multiplayer game model altered if we demand that the verifier's questions to the provers be \emph{independently sampled}? The model in which a computationally efficient verifier referees a game involving two or more potentially dishonest and unbounded provers, and the verifier's questions to the provers must consist of independently sampled and uniformly random bits, is known as the \emph{free game model}. In this paper we will be primarily concerned with the free game model.

The power of classical free games was considered in 2014 by Aaronson, Impagliazzo and Moshkovitz \cite{AIM}. They defined the complexity class $\sf{AM}(k)$, which is exactly the class of problems that can be decided by a computationally efficient verifier interacting with $k$ potentially dishonest and unbounded provers, under the restriction that the verifier's questions to the provers must consist of independently sampled and uniformly random bits. Aaronson, Impagliazzo and Moshkovitz also showed that, for any $k = \poly(n)$, where $n$ is the length of the verifier's input, $\sf{AM}(k) = \sf{AM}$.

This tells us that, relative to the full power of classical multiplayer games, classical free games are very weak. Babai, Fortnow and Lund \cite{BFL91} showed that $\MIP(2)$ contains $\NEXP$, where $\MIP(2)$ corresponds to the `unrestrained' multiplayer two-player game model; on the other hand, Aaronson, Impagliazzo and Moshkovitz showed that, when we force the verifier to ask independently sampled questions, it can only decide problems in $\sf{AM}$---a class which is equal to $\sf{NP}$ under plausible complexity-theoretic assumptions \cite{MV05}. That is, in the classical world, placing the free-game restriction on the verifier seems to result in an exponential decrease in its deciding power!

Intuitively, we can understand this relationship as follows. The best paradigm known for converting $\MIP(2)$ protocols into $\AM(2)$ protocols is what we will call \emph{birthday repetition}. Suppose that there is a one-round $\MIP(2)$ protocol $P$ with constant completeness-soundness gap which allows the verifier to decide some language $L$ of interest, and in which the verifier samples correlated questions $(x,y)$ for the two provers from a set $\cal{X} \times \cal{Y}$. For simplicity, let us suppose that the verifier samples $(x,y)$ uniformly at random from a set $\cal{S} \subseteq \cal{X} \times \cal{Y}$. (This is true, for instance, in the clause-variable example we considered earlier.) We produce a free version of this protocol, $P_{free}$, by simply having the free verifier sample $k$ questions $(x_1, \dots, x_k)$ from $\cal{X}$, and $\ell$ questions $(y_1, \dots, y_\ell)$ from $\cal{Y}$, independently at random. The verifier then checks whether there exists $(x_i, y_j)$ for $i \in [k], j \in [\ell]$ such that $(x_i, y_j) \in \cal{S}$. If there exists such a pair $(x_i, y_j)$, then the free verifier acts as the $\MIP$ verifier would; otherwise, it automatically accepts.

We can represent the set $\cal{X} \times \cal{Y}$ as a bipartite graph $\cal{G}$, with a vertex corresponding to every $x \in \cal{X}$ on the left, and a vertex corresponding to every $y \in \cal{Y}$ on the right. We can also imagine that there is an edge between $x$ and $y$ if and only if $(x,y) \in \cal{S}$. If we assume that every vertex in $\cal{G}$ participates in at least one edge (i.e. that every Alice question $x \in \cal{X}$ has a nonzero chance of being asked, and similarly for every Bob question), the probability that $(x,y) \in \cal{S}$ for $x$ chosen uniformly at random from $\cal{X}$ and $y$ chosen independently and uniformly from $\cal{Y}$ is at least $\Omega\left( \frac{|\cal{X}| + |\cal{Y}|}{|\cal{X}||\cal{Y}|} \right)$. In the case where $|\cal{X}| \approx |\cal{Y}| := N$, this probability is $\Omega(1/N)$. Therefore, using the birthday paradox, we expect to set $k \approx \ell \approx O(\sqrt{N})$ in order to ensure that $P_{free}$, our free version of the $\MIP$ protocol $P$, still has constant soundness. (We call this procedure to turn a non-free game into a free game \emph{birthday repetition} because of the link to the birthday paradox.)

For the $\MIP$ protocols which are sufficiently powerful to capture all of $\NEXP$, as we mentioned earlier, it is typically the case that $|\cal{X}|$ and $|\cal{Y}|$ are both exponentially large in the input length. In other words, in order to convert an $\MIP(2)$-complete protocol into an $\AM(2)$ protocol using birthday repetition, we would need the verifier to send the provers questions of length $2^{cn}$, where $c$ is some constant and $n$ is the length of the input. This is clearly computationally infeasible. Birthday repetition is only computationally feasible when the question sets $\cal{X}$ and $\cal{Y}$ are polynomially large---or, in other words, when the CSP the verifier wants to decide is only polynomially long. This brings us back into the $\sf{NP}$ setting.

However, building on a line of previous work \cite{BT09,ABDFS09},
Aaronson, Impagliazzo and Moshkovitz also identify something which a
free verifier interacting with two noncommunicating provers can do
that no polynomially-bounded verifier interacting with a single prover
can: the former verifier can decide 3SAT \emph{using only $O(\sqrt{N}
  \cdot \poly\log N)$ bits of communication with its provers}, where
$N$ is the number of clauses in the 3SAT instance. Assuming the
Exponential Time Hypothesis (ETH), i.e. that 3SAT cannot be solved in
$2^{o(N)}$ time\footnote{Actually $2^{o(n)}$ whee $n$ is the number of
  variables in the 3SAT instance, but this is linearly related to $N$
  for hard instances.}, a polynomially bounded verifier interacting with a single prover cannot verify 3SAT so efficiently: any such verifier who could would lead to a subexponential-time algorithm for 3SAT. As such, though it may well be the case that $\AM(2) = \NP$, a verifier who referees a two-player free game may still have capabilities that a polynomially bounded verifier interacting with a single prover does not.

\subsection{$\sf{BellQMA}(2)$}
\subsubsection{Background and previous work}
\label{sec:intro-bellqma}

Aaronson, Impagliazzo and Moshkovitz's original motivation \cite[Section 4]{AIM} in studying classical free games was this latter application of deciding 3SAT in sublinear communication, the central ideas in which arose first not from the classical literature but from the study of a quantum class known as $\QMA(2)$. Informally, if $\NP$ is the class of problems which can be efficiently decided by a deterministic classical verifier who is provided with a classical witness, and $\QMA$ is the class of problems which can be efficiently decided by a quantum verifier given a quantum witness, then $\QMA(2)$ is the class of problems which can be decided by a quantum verifier given \emph{two} quantum witnesses. In the $\NP$ world, drawing a distinction between one and two witnesses is clearly absurd: any two classical witnesses can be concatenated into one witness, and any one classical witness can be split arbitrarily into two. In the $\QMA(2)$ world, however, the distinction between one and two witnesses is given meaning by requiring that any `two' witness states are \emph{unentangled} (or, equivalently, that they come from two noncommunicating provers who cannot share entanglement). If unentanglement were a property that an efficient quantum verifier could check for itself, then any $\QMA(2)$ protocol could be converted into a $\QMA$ protocol; however, this is \emph{not} known to be the case, and it remains unknown if $\QMA(2) = \QMA$. We can, of course, also define $\QMA(k)$ for $k \geq 2$, in which the quantum verifier receives $k$ unentangled witnesses from $k$ separate provers.

Our most compelling example of an application for the $\QMA(2)$ setup
(which cannot be instantiated in the $\QMA$ setup, conditioned on the
Exponential Time Hypothesis) is deciding $\NP$ problems in sublinear
communication. That is, we know of a $\QMA(2)$ protocol in which each
of the two provers sends only $O(\sqrt{N} \cdot \log N)$ qubits to the
quantum verifier, where $N$ is the number of clauses in a 3SAT formula
$\phi$, and that verifier can subsequently decide $\phi$ with constant
probability of error. (Note that, in the $\QMA(2)$ model, the quantum
verifier does not send any challenges to the two provers, unlike in
the $\AM(k)$ model---all the communication in a $\QMA(2)$ protocol
happens in a single quantum message from provers to verifier.) A
protocol with sublinear communication to decide 3SAT was firstly
proposed for the $\QMA(\sqrt{N} \poly\log N)$ setup \cite{ABDFS09}, in which a
quantum verifier interacts with $\sqrt{N} \poly\log N$ separate provers, and it was subsequently shown \cite{HM13} that there is a $\QMA(2)$ (two-prover) protocol achieving the same purpose with similar overall communication length.

The pervasive $\sqrt{N}$, which also appeared in the communication complexity of the \cite{AIM} $\AM(2)$ protocol for the same purpose, is not a coincidence---in fact, the \cite{AIM} $\AM(2)$ protocol for deciding 3SAT in sublinear communication draws close inspiration from protocols originally designed for $\QMA(2)$. The $\sqrt{N}$ factor does have some motivation: it originates from a clever application of the birthday paradox \cite{ABDFS09}. So far, nobody has thought of any other technique that might do better. It is natural, then, to wonder whether $\sqrt{N}$ qubits of communication is unavoidable. Is it the case that \emph{any} $\QMA(2)$ protocol for 3SAT must use at least $O(\sqrt{N})$ qubits of communication?

Unfortunately, our provable communication lower bounds in this case fail to match the upper bounds exactly. The best known lower bound on the communication complexity of a $\QMA(2)$ protocol to decide 3SAT originates from \cite{BCY}, which shows that, if there is any $\QMA(2)$ protocol \emph{of a certain restricted form} that can decide 3SAT with constant probability of error, then that protocol must involve at least $O(\sqrt{N})$ qubits of communication. \cite{BCY} shows that any such protocol with smaller communication complexity implies a subexponential-time algorithm for 3SAT, and therefore contradicts the Exponential Time Hypothesis.

The restricted model which \cite{BCY} consider in their lower bound is
one in which the quantum $\QMA(2)$ verifier acts as if it consisted of
two separate parties---call them Arthur and Lancelot---who each
receive one of the two unentangled witnesses provided by the
all-powerful provers. Arthur and Lancelot can then perform separate
measurements on their respective witness states and communicate
classically. After they communicate their measurement outcomes to each
other classically, they are allowed to perform more measurements, and
then communicate classically again, \emph{ad infinitum}; however, they
cannot perform any entangling measurements which straddle the two
witness states. At the end of many rounds of separate measurements and
classical communication, Arthur and Lancelot output a joint
decision. This model is called the \emph{local operations and
  classical communication} (LOCC) model, and the version of $\QMA(2)$
in which the verifier is restricted to behaving in this way is known
as $\textsf{LOCC-QMA}(2)$. \cite{BCY} shows that any
$\textsf{LOCC-QMA}(2)$ protocol for 3SAT must use at least
$O(\sqrt{N})$ qubits of communication between provers and
verifier.

What do we know about upper bounds on the communication necessary to
solve 3SAT in the $\textsf{LOCC-QMA}(2)$
model? Can we at least get a tight characterisation of that class, if
not of general $\QMA(2)$ protocols for $\NP$? The $\textsf{LOCC-QMA}$
model of course encompasses a model in which Arthur and Lancelot
measure their separate witnesses exactly once, and then perform joint
classical computations on the measurement results in order to
determine their decision. This latter model is known as the
$\sf{BellQMA}$ model. In 2010, building on work by Blier and Tapp
\cite{BT09}, Chen and Drucker \cite{CD10} exhibited a remarkably clean
$\sf{BellQMA}$ version of the original \cite{ABDFS09} $\QMA(\sqrt{N}
\poly\log N)$
protocol for 3SAT. In the Chen-Drucker protocol, every separate
quantum witness is measured separately, and the classical measurement
results from these measurements are post-processed classically by the
verifier in order to determine the final decision. The total
communication complexity of the Chen-Drucker protocol is $O(\sqrt{N}
\cdot \log N)$ qubits. The Chen-Drucker protocol would therefore
appear to answer our desire for an $\textsf{LOCC-QMA}$ protocol for
$\NP$ whose communication complexity matches (up to $\log N$ factors)
the lower bound on the communication complexity of $\textsf{LOCC-QMA}$
protocols which was proven by Brand\~{a}o, Christandl and
Yard. Unfortunately, the Chen-Drucker protocol requires $\sqrt{N}$
unentangled provers, not only two, so it cannot show us that the
communication lower bound from the \cite{BCY}
algorithm is optimal.\footnote{We remark that the original \cite{ABDFS09}
  protocol also required $\Omega(\sqrt{N})$ unentangled provers, and
  it was `compiled down' to a two-prover protocol by \cite{HM13};
  however, the `compilation' technique required entangling
  measurements across witness states. We also remark that in the
  multipartite setting, the Chen-Drucker protocol was proven optimal
  by Brand\~{a}o and Harrow~\cite{BH13}.}

Nonetheless, the Chen-Drucker protocol for $\NP$ illuminates the close connection that exists between $\sf{BellQMA}(k)$ and $\AM(k)$, a connection which may not be obvious at first glance. The key ingredient in the Chen-Drucker protocol is a quantum test called the \emph{uniformity test}, which, broadly speaking, involves measuring certain registers of the witness states provided by the provers in the \emph{Fourier basis}, and requiring the measurement outcomes to be zeroes. The zero Fourier state is the uniform superposition in the standard basis. As such, the uniformity test can (morally speaking) act as a substitute for the uniformly random challenges generated by the verifier in the $\AM(k)$ model, because the uniformity test in a sense \emph{forces} the provers to generate their own uniformly sampled challenges. More specifically, if the prover provides us (the verifier) with a state of the form

\begin{equation}
	\ket{\psi} = \sum_{q \in \cal{Q}} \alpha_q \ket{q}_Q \ket{a(q)}_A
	\label{eq:qa-state}
\end{equation}

where $\cal{Q}$ is a set of questions and $a(q)$ represents an answer to a given question $q$, and we can somehow certify that the $\alpha_q$s are all equal to each other (i.e. if we can certify that the question register $Q$ is in a uniform superposition after we---somehow!---`disentangle' it from the answer register $A$), then measuring this state $\ket{\psi}$ in the standard basis is just as good as sampling a uniformly random question $q \in \cal{Q}$, sending it to the prover, and receiving the prover's answer $a(q)$. Therefore, using the uniformity test to replace uniformly generated challenges, we can---sweeping all the inevitable caveats under the rug---simulate $\AM(k)$ protocols in $\sf{BellQMA}(k)$.

Reality, of course, is not quite as clean: this approach to simulating $\AM(k)$ in $\sf{BellQMA}(k)$ is not as general as our vague exposition just now made it out to be. In particular, for the uniformity test technique to work (for an \emph{honest} strategy to exist that passes the uniformity test), it is vital that the answers $a(q)$ are short---constant sized, or at most logarithmically sized. The reason is that, in order to truly get the question register $Q$ into the zero Fourier state, which we define as $\ket{0}_\cal{F} =  \sum_{q \in \cal{Q}} \frac{1}{\sqrt{|\cal{Q}|}} \ket{q}$, we need to `disentangle' it from the answer register first, and this operation involves performing a measurement on the answer register and post-selecting on a measurement outcome which occurs with negligible probability if the answers are long.

In the case where the questions $q$ represent constraints in a CSP, however, and the answers $a(q)$ represent assignments to the variables involved in those constraints, the skies are clear. In essentially all well-studied CSPs, such as 3SAT and graph colouring, any single constraint and any assignment to a single variable in a constraint can be described in constantly many bits! As such, an $\AM(k)$ protocol for 3SAT in the clause-variable paradigm can indeed, at least morally, be `compiled down' into $\sf{BellQMA}(k)$ in this way---which is the starting point for the Chen-Drucker protocol.

\subsubsection{Our results about $\sf{BellQMA}(2)$}

In this work, we resolve the question of whether or not the lower bound proven by \cite{BCY} is tight for $\textsf{LOCC-QMA}(2)$ by exhibiting a $\sf{BellQMA}(2)$ (two-prover) protocol which has communication complexity $O(\sqrt{N} \cdot \log N)$ and decides 3SAT instances with constant probability of error. This question was raised by Chen and Drucker in 2010 \cite{CD10} after they published their protocol, and raised or mentioned several times since then by others \cite[Question 1]{CF11} \cite{BH13} \cite{AIM}, but despite this has remained open for more than 10 years. In resolving this question, we present an (arguably) simpler analysis of the Chen-Drucker uniformity test, as well as a more modular analysis of a $\QMA(2)$ protocol for 3SAT with sublinear communication than any other one we know of, which we hope may be conceptually useful.

As a consequence, we show that the runtime of the \cite{BCY}
algorithm (for approximating the value of a
$\textsf{LOCC-QMA}(2)$ protocol up to constant additive error) is optimal up to logarithmic factors, assuming the
Exponential Time Hypothesis. This is because any improvement to
their algorithm would, in combination with our protocol, result in a subexponential-time algorithm for 3SAT, which would contradict the ETH.

Our protocol is very similar to the Chen-Drucker protocol. The key changes we make are in the analysis, and these changes hinge on the observation that unentanglement is actually \emph{not} necessary to the soundness of the uniformity test. Chen and Drucker assume that the $k = O(\sqrt{N})$ honest provers in their protocol provide the verifier with $k = O(\sqrt{N})$ unentangled copies of a state of the form in \Cref{eq:qa-state}, and they define a $k$-state uniformity test (implementable, of course, using separate measurements on the separate states) with the following properties:
\begin{enumerate}
\item Completeness: 	honest provers providing $k$ copies of a state of the form in \Cref{eq:qa-state}, with $\alpha_q = \alpha^*$ for some constant $\alpha^*$ for all $q$, will pass the $k$-state uniformity test with probability $1 - 2^{-\Omega(k)}$.
\item Soundness: any $k$ unentangled states passing the $k$-state uniformity test with high probability will be such that sufficiently many states among the $k$ states have the form in \Cref{eq:qa-state} with $\alpha_q \approx \alpha^*$ for some constant $\alpha^*$ for all $q$.
\end{enumerate}

Our essential observation is as follows: it is not necessary for the input state to the $k$-state uniformity test to lie in $k$ unentangled registers for a certain form of soundness, which we shall shortly define, to hold. Informally, the soundness guarantee we prove is as follows:

\begin{lemma}[informal version of \Cref{lem:unif-general}]
\label{lem:informal-unif-general}
Let $\ket{\psi}$ be any state passing the Chen-Drucker $k$-state uniformity test with high probability. $\ket{\psi}$ is divided, without loss of generality, into $k$ `question registers' and $k$ `answer registers', \`a la \Cref{eq:qa-state}, which may be entangled. $\ket{\psi}$ is such that measuring all the question registers in the standard basis (approximately) yields a uniformly random string on sufficiently many registers and junk elsewhere.
\end{lemma}

\Cref{lem:informal-unif-general} is the main technical lemma in this part of our work. We will now explain why \Cref{lem:informal-unif-general} yields a two-prover $\sf{BellQMA}$ protocol for 3SAT with $O(\sqrt{N} \cdot \log N)$ communication.

\cite{AIM} exhibits an $\AM(2)$ (classical two-prover free game) protocol with $O(\sqrt{N} \cdot \log N)$ communication complexity which decides 3SAT instances with constant probability of error. This protocol (since it is inspired by the Chen-Drucker protocol) happens to be a clause-variable game which can be `simulated' by a $\sf{BellQMA}$ protocol, in the way that we described at the end of \Cref{sec:intro-bellqma}. In particular, the verifier Arthur's challenge to the first prover Alice consists of $k$ constraints in a CSP, and for each constraint Arthur expects an answer consisting of a constant-sized assignment to the variables involved in that constraint; while Arthur's challenge to the second prover Bob consists of $k$ variables from the same CSP, and for each variable he sends to Bob, Arthur expects to receive a constant-sized assignment to that variable. Leveraging the intuition which we described at the end of \Cref{sec:intro-bellqma}, therefore, we can `compile' the \cite{AIM} $\AM(2)$ protocol into a $\sf{BellQMA}(2)$ protocol: the honest strategy for either prover consists of providing $k$ copies of a state of the form in \Cref{eq:qa-state}, with $\alpha_q = \alpha^*$ for some constant $\alpha^*$ for all $q$. We can then use the Chen-Drucker $k$-state uniformity test to enforce uniformly sampled questions. Completeness holds because the answer to any given question is constantly sized, and the form of soundness which we prove in \Cref{lem:informal-unif-general} is sufficient to induce the soundness guarantees from \cite{AIM}.

The main technical observation which leads to the proof of \Cref{lem:informal-unif-general} is as follows. The Chen-Drucker $k$-state uniformity test has, informally speaking, the following structure:
\begin{enumerate}
\item Given an input state $\ket{\psi}$ on $k$ question and $k$ corresponding answer registers: measure all the answer registers in the Fourier basis. If some `large number' of the resulting measurement outcomes were zeroes, we continue; otherwise, we reject. (We will not be precise about what `large number' means here. For details, see \Cref{fig:unif-test}.)
\item For every $i \in [k]$: if the $i$th answer register measured to zero in step 1, measure the $i$th question register in the Fourier basis as well. If the answer is not zero for any such $i$, reject; otherwise, if the answer is zero for all $i$ such that the $i$th answer register measured to zero in step 1, accept.
\end{enumerate}
Intuitively, this test is trying to leverage the intuition we explained at the end of \Cref{sec:intro-bellqma} to guarantee that as many question registers as possible are in a uniform superposition. Step 1 is necessary because we must `disentangle' the answer registers from the question registers first. We will not explain the completeness property of this test in detail, since it is analysed in \cite[Section 3.1]{CD10}. Instead, we will sketch how we prove \Cref{lem:informal-unif-general}.

Assume that we have some state $\ket{\psi}$ which passes this test with probability 1. Then the measurement in step 1 will yield a `large' set of indices $\cal{S} \subseteq [k]$ such that, for all $i \in \mathcal{S}$, the $i$th answer register measured to zero in the Fourier basis. Denote the post-measurement state after the measurement in step 1 has been performed by $\rho_1$. Because $\ket{\psi}$ passes the $k$-state uniformity test with probability 1, we know that $\rho_1$ must be such that a `large number' of its question registers are in the zero Fourier state (i.e., all the question registers of $\rho_1$ indexed by $i \in \mathcal{S}$ must be in the zero Fourier state---otherwise, step 2 would reject). Therefore, if we hypothetically measured the question registers of $\rho_1$ in the \emph{standard} basis, we would get uniformly random outcomes on a `large number' of the question registers, and junk elsewhere.

The key observation is that this latter hypothetical standard basis measurement on the question registers and the Fourier basis measurement which we performed in step 1 on the answer registers \emph{commute}---because they are performed on different registers. As such, even if we do not firstly measure the answer registers of $\ket{\psi}$ as the test prescribes, and instead directly measure the question registers of $\ket{\psi}$ in the standard basis, we will get uniformly random outcomes on a `large number' of the question registers, and junk elsewhere, just as if we had measured the question registers of $\rho_1$. \Cref{lem:informal-unif-general} follows.

\subsection{$\AM^*(2)$}

\subsubsection{Background and previous work}
\label{sec:intro-am*-background}

The close connection between $\sf{BellQMA}(2)$ and $\AM(2)$ which we explained at the end of Section \ref{sec:intro-bellqma} suggests that $\sf{BellQMA}(2)$ should be considered a `quantum analogue' of $\AM(2)$. However, there is another quantum class which has equally strong claims upon the title. This is the class of problems which can be decided by a \emph{classical} verifier who referees a free game with two unbounded provers \emph{who are allowed to share entanglement}. Following \cite{AIM}, we denote this class by $\AM^*(2)$. As far as we know, Aaronson, Impagliazzo and Moshkovitz were the first ones to define this class \cite[Section 8]{AIM}, and they left characterising its power relative to $\AM(2)$ as an open problem.

Studying the power of `entangled versions' of classical multiprover
classes has a long and fruitful history \cite{CHTW04,IV12,RUV13,FNT14,Ji17}, and has recently led to some surprising and deep
results~\cite{JNVWY20} with connections to pure mathematics. It is not \emph{a priori} clear whether allowing entanglement between the two provers increases or decreases the deciding power of the verifier. On the one hand, the entanglement might allow the provers to help the verifier more effectively, but on the other hand, it might also allow them to cheat more effectively. This is a familiar story: we have seen the same question of whether entanglement helps or hurts arise and be resolved several times already in the history of $\MIP^*(2)$ (the entangled version of $\MIP$) and variants of that class. Ji, Natarajan, Vidick, Wright, and Yuen recently showed that $\MIP^*(2) = \RE$ \cite{JNVWY20}, which clearly indicates that, in the plain multiplayer game model, allowing entanglement increases the deciding power of the verifier (from $\NEXP$ to $\RE$!). On the other hand, it is far from a foregone conclusion that allowing entanglement makes any given multiprover proof system more powerful. For example, it is known that the entangled version of $\oplus \MIP$, a version of $\MIP$ in which the verifier's decision is simply the XOR of two one-bit answers from the two provers, is inside $\sf{EXP}$, even though $\oplus \MIP$ itself is equal to $\NEXP$ \cite{Weh06}. We can conclude from these examples only that it is not clear \emph{a priori} how $\AM^*(2)$ ought to relate to $\AM(2)$.

\subsubsection{Our results about $\AM^*(2)$}
\label{sec:intro-AM*}

In this work, we resolve the open question about the power of $\AM^*(2)$ which was posed by Aaronson, Impagliazzo and Moshkovitz, by showing that, in fact, $\AM^*(2) = \MIP^*(2) = \RE$. In other words, quantum free games are just as powerful as general quantum multiplayer games, even though in the classical world the free-game restriction results in a significant decrease in the verifier's deciding power!

We note that the best lower bound on $\AM^*(k)$ prior to our work, due to Brand\~{a}o and Harrow \cite[Corollary 4]{BH13}, was $\NP \subseteq \AM^*(\sqrt{N})$. In particular, Brand\~{a}o and Harrow showed that there is an $\AM^*(\sqrt{N})$ protocol (analogous to the \cite{AIM} $\AM(2)$ protocol) with $\sqrt{N}$ provers and $O(\sqrt{N} \cdot \log N)$ total communication that decides $N$-clause 3SAT with constant probability of error. Our result subsumes this result: explicitly, we show that there is an $\MIP^*(2)$ protocol with \emph{constant}-sized questions, and answer sizes growing as $\poly\log(n)$, that is capable of deciding all of $\RE$ with constant probability of error, where $n$ is the size of the problem instance being decided. Since free games with constant-sized questions are equivalent to general games with constant-sized questions, we obtain a very communication-efficient $\AM^*(2)$ protocol for $\RE$. We prove this result by using the powerful machinery developed by Ji, Natarajan, Vidick, Wright and Yuen which was not available in 2013 to Brand\~{a}o and Harrow.

The key difference between quantum and classical free games, which allows $\AM^*(2) = \MIP^*(2)$ even though $\AM(2)$ and $\MIP$ are significantly different in power, is that allowing the provers to share entanglement opens up access to the tools provided by the self-testing literature, which allows us to get around the `birthday repetition barrier' we identified in the first section of this introduction. In particular, self-testing allows us to use relatively little communication to force the two provers to \emph{introspect}, namely to generate their own (long) questions, when they play an entangled game, and thus allows us to avoid having to send very large questions in order to achieve a constant probability of free collisions. The machinery of introspection was introduced in \cite{NW19} in order to prove $\MIP^* \supseteq \sf{NEEXP}$, and is at the heart of the \emph{compression theorems} which led to $\MIP^* = \RE$. Compression theorems are transformations that take as input some multiplayer game with long questions and answers and large verifier complexity, and output a new multiplayer game with (usually exponentially) smaller questions and answers and verifier complexity, \emph{that has about the same value as the original game}: in particular, the fact of whether the value of the original game was $= 1$ or $\leq 1/2$ should be preserved. That compression theorems can exist at all for entangled games is testament to the marvellous power of self-testing theorems. In particular, the proof that $\MIP^* = \RE$ follows by (in a sense) recursively applying a compression theorem.

In this work we take the compression theorems that were used to prove $\MIP^* = \RE$ and `bootstrap' them to prove that $\MIP^*$ is equal to $\MIP^*$ with constantly sized questions. Specifically, we prove what we term a \emph{hypercompression theorem} (\Cref{thm:gapped-hyper}), which is also the result of recursively applying compression theorems, but in a slightly different way from the way that appears in the proof of $\MIP^* = \RE$. Our hypercompression theorem starts with any $\MIP^*$ protocol with polynomially long questions and answers, and applies a general compression theorem once in order to turn it into a protocol with polylogarithmically long questions and answers, before recursively applying a \emph{question reduction} theorem to bring the question size down to constant while more or less preserving the answer size. (The efficacy of this recursive application procedure is dependent on the structure of the question reduction theorem---in particular, we cannot reduce the answer size in quite the same way, for reasons related to the fact that the efficacy of answer reduction depends on the running time of the verifier in the original game while question reduction does not.) The question reduction procedure that we use is similar to the one in \cite{JNVWY20}, although we believe that, by incorporating recent improvements to the analysis of question reduction made by de la Salle \cite{dlS22}, one would be able to prove a better question reduction theorem that might be a stepping stone towards an $\MIP^*$ protocol for $\RE$ with constantly sized questions and (truly, or up to a factor of $O(\log^*n)$) logarithmically sized answers. Due to time constraints, we leave this improvement for a future version of the paper. We discuss this possibility in more detail in \Cref{sec:open-problems}.

%(The efficacy of this recursive application procedure is dependent on the structure of the question reduction theorem---in particular, it does not work in the same way for the answer reduction theorem, for reasons related to the fact that the efficacy of answer reduction depends on the running time of the verifier in the original game while question reduction does not.) 

Two remarks about this result are in order for the benefit of the interested reader.
\begin{itemize}
\item We also prove the \emph{gapless} version (\Cref{thm:mip0gap-main}) of this result---namely, that zero-gap $\MIP^*$ is equal to zero-gap $\MIP^*$ with \emph{constantly} sized questions and $O(\log n \cdot \log^*n)$ sized answers. Zero-gap $\MIP^*$ is the same as normal $\MIP^*$ except that, for no-instances, the verifier's acceptance probability is only required to be strictly less than 1 instead of $\leq \frac{1}{2}$. In order to get gapless analogues of the question reduction and answer reduction theorems of \cite{JNVWY20}, we look to \cite{MNY20}, which proves that zero-gap $\MIP^*$ is equal to $\Pi_2$.
\item A natural corollary of our gapless hypercompression theorem is that there is a (non-robust) two-prover test for $n$ EPR pairs that uses only \emph{constantly} sized questions (and $O(n)$ sized answers). This result (\Cref{thm:epr-gapless}) arises from applying the gapless hypercompression theorem to the `question sampling game' of \cite{MNY20}, which self-tests for $n$ EPR pairs, because hypercompression also preserves entanglement bounds. (We believe that it may be possible to improve the $O(n)$-sized answers to $\poly\log(n)$ by applying a round of gapless answer reduction to the question sampling game before we apply hypercompression.) As far as we know, this is the first nonlocal game\footnote{There are \emph{nonlocal correlations} with constant-sized questions~\cite{CGS17}, and indeed constant-sized questions and answers~\cite{Fu19} that self-test maximally entangled states of arbitrarily high dimension. However, this is a different notion of self-testing, where one requires not just the winning probability to be close to optimal, but the entire distribution of answers given questions to be close to a target distribution.} in the literature which achieves a self-test for a growing number of EPR pairs using constantly sized questions (see \cite[Table 1]{BS20}). We leave obtaining an analogous result in the gapped case, which would result in a \emph{robust} two-prover test for $n$ EPR pairs with constant sized questions, as an open problem; see \Cref{sec:open-problems} for more discussion.
\end{itemize}

\subsection{Lower bounds on $\MIP^*$ protocols from Kolmogorov complexity}
\label{sec:intro-lower-bounds}
Our previous result shows that $\MIP^*$ with constant question complexity and polylogarithmic answer complexity is equal to general $\MIP^*$ (with polynomial question, answer and decision complexity). (We also prove that zero-gap $\MIP^*$ with constant question complexity and almost-logarithmic answer complexity is equal to general zero-gap $\MIP^*$.) It is natural to ask how far we can push in this direction. For example, is $\MIP^*$ with constantly sized questions and (truly) $O(\log n)$ sized answers equal to general $\MIP^*$? What about $\MIP^*$ with constantly sized questions and, say, $O(\log \log n)$ sized answers?

Our final set of results shows that the parameters we can achieve by using hypercompression (see the previous section of this introduction) are in fact almost tight. Specifically, we prove (\Cref{thm:lower-bound}) that any $\MIP^*$ protocol deciding all of $\RE$---in fact, any $\MIP^*$ protocol deciding all of $\sf{EEXP}$---must have $q(n) + a(n) \geq \frac{1}{2} \log n$, where $n$ is the instance size and $q(n)$ and $a(n)$ are the question and answer sizes (for a single prover) in the protocol respectively. An identical lower bound holds on question and answer sizes for gapless $\MIP^*$ protocols deciding all of $\sf{EEXP}$. In particular, the latter shows that we have essentially already achieved a tight characterisation of zero-gap $\MIP^*$ as far as question and answer complexity are concerned: \cite{MNY20} exhibited a zero-gap $\MIP^*$ protocol for $\Pi_2$ with $O(\log n)$ question complexity and $O(1)$ answer complexity, and we exhibit a zero-gap $\MIP^*$ protocol for $\Pi_2$ with $O(1)$ question complexity and $O(\log n \cdot \log^*n)$ answer complexity, the former of which matches the lower bound up to constant factors, and the latter of which matches the lower bound up to a factor of $O(\log^*n)$. In the gapped case, some degree of leeway remains between the upper and the lower bound---in particular, the lower bound has $q(n) + a(n) \geq \frac{1}{2}\log n$, but the best upper bound that we believe current techniques could prove only has $q(n) + a(n) = \poly\log(n)$. We think that the upper bound is the one that can be tightened, and we leave closing the gap as an interesting open problem whose resolution may have other significant implications. (See our open problems section, \Cref{sec:open-problems}, for more discussion of this.)

We prove this lower bound by observing a connection between the sizes of questions and answers in an $\MIP^*$ protocol deciding a computational problem and the size of the \emph{advice} that a deterministic Turing machine must take to solve the same problem (or, equivalently, the size of the description of a Turing machine that solves the same problem). More specifically, we show a way to convert any $\MIP^*$ protocol with questions of size $q(n)$, answers of size $a(n)$, and verifier time complexity $t(n)$ deciding a language $L$ into a deterministic Turing machine running in time roughly $2^{t(n)}$ and taking advice of length roughly $2^{2^{2(q(n) + a(n))}}$ which also decides $L$. We then observe that one can use techniques from time-bounded Kolmogorov complexity theory to show that $\sf{EEXP}$ cannot be decided by any Turing machine running in time $2^{\poly(n)}$ and taking $\eps 2^{cn}$ advice for any $\eps + c < 1$.\footnote{One may ask why we had to prove this, i.e. why we did not use known circuit lower bounds for large time classes. The answer is that, because the \emph{advice complexity} of the Turing machine which we obtain from `converting' the $\MIP^*$ protocol is more sensitive to $q(n) + a(n)$ than the \emph{running time} of the same Turing machine, we wanted a lower bound which treated running time and advice separately. In particular, $t(n)$ (the verifier's time complexity in the $\MIP^*$ protocol) could be any arbitrary polynomial in $n$, e.g. $n^{100}$, and may not depend explicitly on $q(n) + a(n)$ (which here could be sub-logarithmic). Because any language is decidable by circuits of size $2^n$, and the running time of the Turing machine $M$ which comes out of our `conversion' process is $2^{t(n)}$, we would not be able to prove any substantial conclusions about $q(n) + a(n)$ by comparing the circuit version of $M$ with known circuit lower bounds if $t(n)$ happened to be $n^{100}$, since then the complexity of the circuit version of $M$ would already be large enough to decide any language even if we only counted the $2^{t(n)} \approx 2^{n^{100}}$ gates that came from encoding the tableau of $M$'s execution. On the other hand, since the \emph{advice complexity} of the Turing machine $M$ depends sharply on $q(n) + a(n)$, a lower bound on $\sf{RE}$ (or $\sf{EEXP} \subseteq \RE$) that has a precise dependence on advice and a looser dependence on time complexity serves our purposes well.

We remark that another lower bound for $\RE$ with a precise dependence on advice and a looser dependence on time complexity is the bound which states that no finite-time Turing machine can solve the halting problem with fewer than $\approx n$ bits of advice. However, this bound is `too much in the other direction', i.e. the lower bound on the advice is very weak because the running time is allowed to be any finite time, and therefore potentially much larger than $2^{\poly(n)}$. We wanted a bound which captured a trade-off between running time and advice that would allow us to derive a logarithmic lower bound on $q(n) + a(n)$, and so we proved the bound stated in the main text.} Combining the two statements shows our claimed lower bound, since an $\MIP^*$ protocol for $\RE$ with very small questions and answers would result in a Turing machine to decide $\RE$ that takes comparatively little advice, which would contradict the lower bound on $\sf{EEXP}$.

% \item One might wonder, having seen our result, why we cannot achieve constant answer size as well as constant question size using similar techniques to those which we used to achieve constant question size alone. (Intuitively, this should be impossible, because a protocol with constantly sized questions and constantly sized answers is clearly not capable of deciding problem instances of growing size. However, one might wish to know more particularly where our techniques fail.) The answer is that the answer reduction theorem, unlike the question reduction theorem, has a \emph{logarithmic} (instead of constant) `bottom limit' beyond which point applying it again no longer achieves any savings.

%\begin{enumerate}
%\item starting with an instance $M$ of the halting problem, where $M$ is some Turing machine,
%\item converting that into an infinite \emph{family} $\cal{M} = \{(M,n)\}_{n \in \mathbb{N}}$ of instances of the (time-bounded) halting problem indexed by $n \in \mathbb{N}$, where the $n$th instance is defined to be a yes-instance if $M$ halts after $n$ steps and a no-instance otherwise;
%\item converting the family of instances $\cal{M}$ into a family of entangled games $\cal{G} = \{G_n\}_{n \in \mathbb{N}}$ indexed by $n \in \mathbb{N}$, where the $n$th game has value $1$ if $(M,n)$ is a yes-instance and value $\leq 1/2$ otherwise;
%\item 
%\end{enumerate}

\subsection{Related work}
We have already addressed much of the literature relevant to our work in sections \ref{sec:intro-bellqma} and \ref{sec:intro-am*-background}; in this section we briefly mention some other related work which we have not yet discussed.
\begin{itemize}
\item Chiesa and Forbes \cite{CF11} address several questions relevant to the question of whether there is a $\sf{BellQMA}(2)$ protocol for $\NP$ with sublinear communication, and in particular also produce a tighter analysis of the Chen-Drucker protocol \cite{CD10}. The original Chen-Drucker analysis simply yielded constant soundness and $1 - \exp(-\sqrt{N})$ completeness in the $O(\sqrt{N})$-prover setting. Chiesa and Forbes obtain a smooth trade-off between the number of provers and the completeness-soundness gap: in particular, they show that the Chen-Drucker protocol executed with $\kappa$ provers (where each prover sends the verifier a witness state that is $O(\log N)$ qubits long) has a completeness-soundness gap of $\Omega(\kappa^2 N^{-1})$, as long as $\kappa = \Omega(\log N)$. However, Chiesa and Forbes do not appear to consider protocols like our protocol, in which each of constantly many (in our case, 2) provers provides the verifier with $O(\sqrt{N})$ qubits, and in which each of the two witness states is (in the case of honest provers) really many copies of the state which each prover would have sent in the original Chen-Drucker protocol, `batched together' under a single prover.
\item Brand\~ao and Harrow \cite{BH13} show, among many other things, that the Chen-Drucker protocol is essentially optimal in the setting of $O(\sqrt{N})$ \emph{symmetric} $\sf{BellQMA}$ provers: that is, they show that any $\sf{BellSymQMA}(\sqrt{N})$ protocol with constant soundness and $O(n^{1/2 - \eps})$ communication (for any $\eps > 0$) would contradict the ETH, where $\sf{BellSymQMA}$ is a further restriction on $\sf{BellQMA}$ in which each prover must send the same state to the verifier. Their result can be understood as an analogue of the \cite{BCY} lower bound in the $\sf{BellSymQMA}(\sqrt{N})$ setting. It is not \emph{a priori} clear that $\sf{BellSymQMA}(\sqrt{N})$ (with $O(\sqrt{N} \cdot \log N)$ communication) can be simulated in $\sf{BellQMA}(2)$ with $O(\sqrt{N} \cdot \log N)$ communication, since having stronger unentanglement guarantees might be useful to the verifier; however, the other direction is not clear either, since in principle a $\sf{BellQMA}(2)$ protocol could require Arthur and Lancelot to do arbitrary measurements on their two separate witness states, which could be entangling across any series of cuts that would attempt to divide those two witness states into $O(\sqrt{N})$ separate witness states. Therefore, the power of $\sf{BellQMA}(2)$ relative to $\sf{BellSymQMA}(\sqrt{N})$ is simply not understood.

It so happens that our $\sf{BellQMA}(2)$ protocol for 3SAT \emph{can} be simulated in $\sf{BellSymQMA}(\sqrt{N})$, because we did not need Arthur and Lancelot to perform any highly entangling measurements for completeness to hold. Meanwhile, in the other direction, our result can also be understood as a proof that the extra unentanglement in the Chen-Drucker protocol does not actually afford it much extra power. We believe this illuminates a potentially interesting connection, because our results suggest either that the difference of unentanglement between $\sf{BellQMA}(2)$ and $\sf{BellQMA}(\sqrt{N})$ is not actually very consequential, or that we have yet to fully exploit its power.
\end{itemize}

% \anote{Points to hit somewhere}
% \begin{itemize}
%   % \item There are three quantum analogues to CSPs: q protocols for c
%   %   CSPs, q games, and q Hamiltonians. We settle the complexity of
%   %   ``free'' versions of the first two, and the last might be settled
%   %   by Brand\~{a}o Harrow (need to check).
%   \item 
% \end{itemize}

\subsection{Open questions}
\label{sec:open-problems}
\begin{enumerate}
\item \textbf{Putting $\MA$ or $\AM$ in $\QMA(k)$ with small communication complexity.} There is now a long line of works about $\QMA(2)$ protocols for $\NP$ with sublinear communication. It is natural then to ask: is there a $\QMA(k)$ protocol (or an $\AM(k)$ protocol) with sublinear communication for $\AM$ (or $\MA$)? The main obstacle here is that we do not have a `PCP theorem' for $\MA$ or $\AM$ (in the sense that we have one for $\NP$), unless $\MA = \NP$ (resp. $\AM = \NP$), but the birthday paradox trick which puts $\NP$ in $\QMA(2)$ with sublinear communication complexity relies centrally on having a very short (short in terms of proof length) PCP for $\NP$. Alternatively, could we prove that, if $\MA$ (or $\AM$) is in $\QMA(k)$ with sublinear communication complexity, then $\MA = \NP$ (or $\AM = \NP)$?
%In Section \ref{sec:intro-bellqma}, we outlined an intuitive connection between $\sf{BellQMA}(k)$ and $\AM(k)$ which suggests that, under certain limited circumstances, $\AM(k)$ can be simulated in $\QMA(k)$. In reality, however, those limited circumstances seem (as far as we currently know) to extend only just far enough to cover deciding $\NP$ languages in $\sf{BellQMA}(k)$ with sublinear communication complexity, even though we know that $\AM(k) = \AM$. 
\item \textbf{A tight gapped $\MIP^*$ protocol for $\RE$.} As we mentioned in \Cref{sec:intro-lower-bounds}, there is some leeway between our lower bound on the communication complexity of any $\MIP^*$ protocol for $\mathsf{EEXP}$ and the upper bound which we believe we can achieve by applying hypercompression to the \cite{JNVWY20} $\MIP^*$ protocol for $\RE$. One way to prove a tight upper bound would be to show that there are \emph{entanglement-sound PCPPs for $\NP$}---that is, to prove the result which \cite{NV18} claimed to prove, but whose proof subsequently turned out to have a bug. Such a result would yield a \emph{gapped answer reduction theorem} that has similar parameters to those of the gapless answer reduction theorem (\Cref{thm:ar-gapless}) which we make use of in \Cref{sec:gapless}. We believe that, in order to get a gapped \emph{question} reduction theorem with similar parameters to those of the gapless question reduction theorem (\Cref{thm:intro-gapless}) that we make use of in \Cref{sec:gapless}, one could apply the recent results of de la Salle \cite{dlS22}. Since our hypercompression theorem is able to use the gapless question and answer reduction theorems of \cite{MNY20} to get a zero-gap $\MIP^*$ protocol for $\Pi_2$ with $O(1)$ questions and $O(\log n \cdot \log^* n)$ answers, we believe we would be able to construct a protocol with similar parameters for the gapped setting if we had equally strong gapped question and answer reduction theorems.
\item \textbf{A \emph{rigid} self-test for $n$ EPR pairs with
    constant-sized questions.} The techniques in this work can only
  yield a bound on the \emph{dimension} of the Hilbert space shared by
  any pair of provers who pass in our so-called `self-test for $n$ EPR
  pairs' with constantly sized questions. Is there a self-test for $n$
  EPR pairs with constantly sized questions (and, say, $\poly(n)$
  sized answers) which guarantees that any two provers who pass in the
  self-test must be using a particular strategy, up to local
  isometries? This is a very powerful and useful property of most
  self-tests for EPR pairs in the literature which is known as
  rigidity.
\item \textbf{A \emph{robust} test for Schmidt rank $n$ with constant-sized
    questions.}
In this work, the only entanglement bounds we can obtain are in the
\emph{gapless} case (i.e. for perfect strategies). Can we show a game
with constant-sized questions where any strategy achieving value $\geq
1/2$ must have \emph{Schmidt rank} at least $n$? (The Schmidt rank is
the number of nonzero Schmidt coefficients, and is a relatively loose
characterization of entanglement.) Such a bound was obtained by
\cite{JNVWY20} for their compression theorems, but we cannot use it
for an interesting technical reason: in our work, in order to perform
parallel repetition for games with large answer sizes, we must use the
analysis of \cite{DSV15}, rather than that of \cite{BVY17}. However,
this analysis does not preserve entanglement bounds, since the
reduction from the parallel repeated strategy to a single-round
strategy requires adding a large amount of entanglement. We refer the
reader to the discussion in \cite[Section 11]{JNVWY20} for more
details on this point.

\item \textbf{A communication lower bound for a self-test for $n$ EPR pairs.} We get a lower bound on the communication complexity of $\MIP^*$ protocols for $\RE$ which almost matches the upper bound we achieve by using hypercompression. As we mention in \Cref{sec:intro-AM*}, another consequence of (gapless) hypercompression is a non-robust self-test for $n$ EPR pairs with $O(1)$ sized questions and (probably, using answer reduction) $\mathrm{\poly\log(n)}$ sized answers. Is there a way to lower bound the communication complexity of a self-test for $n$ EPR pairs using computational arguments, as we did the communication complexity of $\MIP^*$ protocols for $\RE$? It is not at once obvious how to do this, since such a self-test does not directly solve any well-understood computational problem.
\item \textbf{Infinite randomness expansion with two provers.}
  Hypercompression yields a self-test for $n$ EPR pairs with
  \emph{constantly} sized questions. It is tempting then to ask: can
  we do \emph{infinite randomness expansion} \cite{CY13} using only 2
  provers by using this self-test? The na\"ive approach does not work
  because, in the `question sampling game' of \cite{MNY21} and the
  `introspection game' of \cite{JNVWY20}, the probability that the provers are asked the `introspect' questions which cause them to generate randomness (as opposed to being asked questions that test their consistency with each other) decreases by a constant factor every time one applies question reduction, and we need to apply question reduction approximately $\log^*(n)$ times in order to make the questions constant sized. Can this obstacle be gotten around?
\end{enumerate}

\section{Preliminaries}
\label{sec:prelims}
\subsection{Probability basics}
We can represent a probability distribution $\mu : \Omega \rightarrow [0,1]$ over a finite sample space $\Omega$ as a vector $\vec \mu$ of length $|\Omega|$ such that the $i$th entry of the vector $\vec \mu$ is exactly $\mu(i)$. For two probability distributions $\mu, \nu$ over sample spaces $\Omega$ and $\Omega'$, we then denote by $\mu \otimes \nu$ the probability distribution over $\Omega \times \Omega'$ whose vector representation is the vector $\vec{\mu} \otimes \vec{\nu}$.
    \subsection{Quantum information basics}
    \begin{definition}
      For $K \in \mathbb{N}$, the \emph{quantum Fourier transform}
      $\cal{F}_K$ is the unitary map over $\mathbb{C}^{K}$ defined by
      \begin{equation}
        \cal{F}_K \ket{s} = \frac{1}{\sqrt{K}}\sum_{t
          =0}^{K-1} \omega_K^{s \cdot t} \ket{t},
      \end{equation}
      where $\omega_K = \exp(2\pi i / K)$. This map defines the
      \emph{Fourier basis} consisting of the states
      $\ket{\bar{s}} = \cal{F}_K \ket{s}$ for $s \in \{0, \dots,
      K-1\}$. In particular,
      \begin{equation}
        \ket{\bar{0}} = \frac{1}{\sqrt{K}} \sum_t \ket{t} .
      \end{equation}
    \end{definition}
    \begin{definition}
    \label{def:bell-meas}
      A \emph{Bell measurement} is a two-outcome measurement $\{M, \id-
      M\}$ on a bipartite Hilbert space $\mathcal{H}_{AB} = \cal{H}_{A}
      \ot \cal{H}_B$ that can be implemented by separately measuring the
      $A$ and $B$ registers with a POVM measurement, and then applying
      a classical Boolean function to the measurement outcomes. In
      other words, there exist POVMs $\{A_a\}$ acting on $\cal{H}_A$
      and $\{B_b\}$ acting on $\cal{H}_B$, and a Boolean function $f$ such that
      \[ M = \sum_{a,b: f(a,b) = 1} A_a \ot B_b. \]
    \end{definition}
    \subsection{Kolmogorov complexity}
        \begin{definition}
\label{def:time-bounded-kolmogorov}
The following definition is based on Definition 2 from \cite{hm95}.

Fix a choice of a universal simulator $T$ (a Turing machine). Let $K[f,g]$ be the set of strings that can be produced by $T$ running on strings of length $f$ for time $g$. More formally,
\[
K[f,g] = \{ u : \exists w \text{ s.t. } |w| \leq f(|u|), T(w) = u \text{ and this result is obtained in at most $g(|u|)$ steps of $T$}\}.
\]

Lemma 2.1 in \cite{hm95} justifies the use of a fixed universal simulator by showing that the values of $K[f,g]$ do note change very much if we switch $T$ for a different universal simulator.

\end{definition}
\subsection{Nonlocal games and $\MIP^*$ protocols}
    In this section, we define basic notions concerning nonlocal
    games, strategies, and protocols. For a more detailed treatment,
    we refer the reader to \cite{MNY21} and \cite{JNVWY20}.

    \subsubsection{Nonlocal games}
    
    The following definitions are from \cite{JNVWY20}.
    
    \begin{definition}[Two-player one-round games]
  \label{def:game}
  A \emph{two-player one-round game} $\game$ is specified by a tuple
  $(\cal{X}, \cal{Y}, \cal{A}, \cal{B}, \mu, D)$ where
  \begin{enumerate}
  \item $\cal{X}$ and $\cal{Y}$ are finite sets (called the \emph{question
      alphabets}),
  \item $\cal{A}$ and $\cal{B}$ are finite sets (called the \emph{answer
      alphabets}),
  \item $\mu$ is a probability distribution over $\cal{X} \times \cal{Y}$
    (called the \emph{question distribution}), and
  \item $D: \cal{X} \times \cal{Y} \times \cal{A} \times \cal{B} \to \{0,1\}$ is
    a function (called the \emph{decision predicate}).
  \end{enumerate}
\end{definition}

\begin{definition}[Tensor product strategies]
  \label{def:tensor-product-strategy}
  A \emph{tensor product strategy} $\strategy$ for a game $\game = (\cal{X},
  \cal{Y}, \cal{A}, \cal{B}, \mu, D)$ is a tuple $(\ket{\psi}, A, B)$ where
  \begin{itemize}
	\item $\ket{\psi}$ is a pure quantum state, i.e.\ a unit vector in $\cal{H}_A \otimes \cal{H}_B$ for finite
    dimensional complex Hilbert spaces $\cal{H}_A, \cal{H}_B$,
	\item $A$ is a set $\{A^x\}$ such that for every $x \in \cal{X}$, $A^x =
    \{A^x_a \}_{a \in \cal{A}}$ is a POVM over $\cal{H}_A$, and
	\item $B$ is a set $\{B^y\}$ such that for every $y \in \cal{Y}$, $B^y =
    \{B^y_b \}_{b \in \cal{B}}$ is a POVM over $\cal{H}_B$.
\end{itemize}
\end{definition}

\begin{definition}[Tensor product value]
  \label{def:tensor-product-value}
	The \emph{tensor product value} of a tensor product strategy $\strategy =
  (\ket{\psi}, A, B)$  with respect to a game $\game=(\cal{X}, \cal{Y}, \cal{A},
  \cal{B}, \mu, D)$ is defined as
  \begin{equation*}
		\val^*(\game, \strategy) = \sum_{x,\, y,\, a,\, b} \, \mu(x,y)\, D(x,y,a,b)\,
    \bra{\psi} A^x_a \otimes B^y_b\, \ket{\psi}\;.
  \end{equation*}
	For $v\in[0,1]$ we say that the strategy $\strategy$ \emph{passes (or wins)
    $\game$ with probability $v$ if} $\val^*(\game, \strategy) \geq v$.
  The \emph{tensor product value} of $\game$ is defined as
  \begin{equation*}
		\val^*(\game) = \sup_\strategy \val^*(\game, \strategy)\;,
  \end{equation*}
	where the supremum is taken over all tensor product strategies $\strategy$ for
  $\game$.
\end{definition}

\begin{remark}
  Unless specified otherwise, all strategies considered in this paper are tensor
  product strategies, and we simply call them \emph{strategies}.
  Similarly, we refer to $\val^*(\game)$ as the \emph{value} or \emph{quantum value} of the game
  $\game$.
\end{remark}
    
%    \begin{definition}
%      A two-player game $G$ is specified by tuple $(\mathcal{X}, \mathcal{A},
%      \mu, D)$, where $\mathcal{X}$ is a finite set called the
%      \emph{question alphabet}, $\mathcal{A}$ is a finite set called
%      the \emph{answer alphabet}, $\mu$ is a probability distribution
%      over $\cal{X} \times \cal{X}$ called the \emph{question distribution}, and $D: \cal{X} \times \cal{X}
%      \times \cal{A} \times \cal{A} \to \{0, 1\}$  is a function
%      called the \emph{decision predicate}. A \emph{quantum
%        (finite-dimensional) strategy}
%      $\scr{S}$ for the game $G$ consists of a finite-dimensional bipartite entangled
%      state $\ket{\psi}$ together with a collection of POVMs
%      $\{A^x_a\}_{a \in \mathcal{A}}$ and $\{B^x_b\}_{a \in
%        \mathcal{A}}$ indexed by questions $x \in \mathcal{X}$,
%      describing the measurements performed by the two players. 
%    \end{definition}
%

\begin{definition}\label{rem:symmetric-games}
	A game $\game = (\cal{X}, \cal{Y}, \cal{A}, \cal{B}, \mu, D)$ is
  \emph{symmetric} if the question and answer alphabets are the same for both
  players (i.e.\
  $\cal{X} = \cal{Y}$ and $\cal{A} = \cal{B}$), the distribution $\mu$ is
  symmetric (i.e.\
  $\mu(x,y) = \mu(y,x)$), and the decision predicate $D$ treats both players
  symmetrically (i.e.\ for all $x,y,a,b$, $D(x,y,a,b) = D(y,x,b,a)$).
 
We call a strategy $\strategy = (\ket{\psi}, A, B)$
  \emph{symmetric} if $\ket{\psi}$ is a (pure) state in $\mH \otimes \mH$, for some
  Hilbert space $\mH$, that is invariant under permutation of the two factors,
  and the measurement operators of both players are identical.
\end{definition}

We often specify symmetric games $\game$ and symmetric strategies $\strategy$ using
  a compact notation: we write $\game = (\cal{X}, \cal{A}, \mu, D)$ and
  $\strategy = (\ket{\psi}, M)$ where $M$ denotes the set of measurement
  operators for both players.
  
    \begin{definition}
      A \emph{synchronous game} $G = (\cal{X}, \cal{A}, \mu, D)$ is
      one where for all $x \in \cal{X}$ and $a, b \in \cal{A}$ with $a
      \neq b$, it holds that $D(x, x, a, b) = 0$.  A
      \emph{finite-dimensional synchronous strategy} is one where both
      players share a maximally entangled state, and for every
      question $x$ and outcome $a$, Alice and Bob's measurement
      operators $A^x_a, B^x_a$ are projective and satisfy $B^x_a =
      (A^x_a)^T$. 
    \end{definition}
%    \begin{definition}
%      For a nonlocal game $G$ and a strategy $\scr{S}$, the
%      \emph{value} of the strategy $\omega_q(G, \scr{S})$ is the
%      probability of the players winning the game using $\scr{S}$. The 
%      \emph{quantum value} $\omega_q(G)$ is the supremum of the value
%      over all finite-dimensional quantum strategies. The
%      \emph{synchronous value} $\omega_q^s(G)$ is the supremum over
%      all synchronous finite-dimensional strategies.
%    \end{definition}

    \begin{definition}
      For a synchronous nonlocal game $G$ with a value $1$
      finite-dimensional strategy, the entanglement bound
      $\cal{E}(G)$ is the minimum dimension of a synchronous strategy
      that achieves value $1$.
    \end{definition}
    
    \subsubsection{Gapless game families and $\AM^*_0(2)$ protocols}
    \label{sec:prelims-gapless}
    For our gapless results, following \cite{MNY21} we will give
    bounds on the \emph{synchronous} value only. We do not consider
    this to be a major restriction, as for synchronous
    games, the synchronous value
    preserves all of the qualitative fatures of the entangled value,
    while being technically much cleaner to work with. It is also
    known that in the gapped case, for synchronous games, the
    synchronous value is close to the quantum value~\cite{Vid22}.
    
%     We now give a succinct treatment of $\MIP^*$ and related
%    classes. Once again, we refer the reader to \cite{JNVWY20,MNY21}
%    for a fuller treatment. We start with the case of protocols with \emph{zero soundness
%      gap}. 

    Roughly speaking, a language has a zero-gap $\MIP^*$ if
    for every instance $x$ in the language, the corresponding game has
    entangled value $1$, and for every instance not in the language,
    the corresponding game has value less than $1$. For any nonlocal
    game $G$ with question distribution $\mu$, the property of having
    value $1$ is preserved under changing the question distribution to
    any other distribution $\mu'$ as long as it has the same support
    as $\mu$.  Thus, following \cite{MNY21}, we define a zero-gap
    protocol in terms of two Turing machines: a \emph{decider} $D$, that
    given a pair of questions and answers decides whether the game was
    won, and a \emph{checker} $C$, that given a pair of questions decides
    whether they are in the support of the distribution $\mu$. To
    emphasize that the question distribution is not relevant here, and
    to make connections to the gapped case, we will call this class
    $\AM_0^*(2)$. In order to define it, let us first formalize the
    specification of a game in terms of a decider and checker.

    \begin{definition}
      Given a natural number $q$ and a pair of Turing machines $D, C$
      with $D$ taking input $x, y, a, b$ and $C$ taking input $x,y$,
      suppose that $D$ and $C$ always halt when $x, y \in \{0,1\}^q$
    and $D$ reads no more than $m$ bits of $a,b$ when $x, y \in
    \{0,1\}^q$ for some finite $m$. Then the associated nonlocal game, or the \emph{game given by $q$ and $(D,C)$},
    is defined by $(\cal{X} =
      \{0,1\}^q, \cal{A} = \{0,1\}^m, \mu, D')$, where $\mu$ is the
      uniform distribution over $\cal{X} \times \cal{X}$, and $D'$ on
      input $x,y,a,b$ (1) first runs the checker $C(x,y)$ and
      automatically accepts if the checker \emph{rejects} (i.e. declares $x,y$ an invalid question pair), and (2) then runs $D(x,y,a,b)$ and accepts if it
      accepts, and rejects otherwise.
    \end{definition}
    For all pairs $D,C$ that we consider in this paper, the existence
    of $m$ will be clear, and so we do not prove that $m$ exists. 
    
    In order to discuss compression theorems, it is useful to define
    families of nonlocal games indexed by a natural number $n$.
    \begin{definition}
      \label{def:game-family-gapless}
      A \emph{family of zero-gap games} is specified by a triple
      $(D, C, Q)$ of Turing machines where $D$ takes input $(n, x, y,
      a,b)$, $C$ takes input $(n, x,y)$, and $Q$ takes input $n$. The input $n$ is called
      the \emph{index}. For index $n$, the corresponding game is
      the one specified by $Q(n)$ and $(D(n, \cdot, \cdot ,\cdot,
      \cdot), C(n, \cdot, \cdot))$.
            We say that game family has \emph{question
        length} $q(n)$ if the output of $Q$ on input $n$ is at most $q(n)$, \emph{answer length}
      $a(n)$ if $D$ on input $n$ reads no more than $a(n)$ bits of the
      answers $a,b$, and \emph{decider runtime} $t_d(n)$ and
      \emph{checker runtime} $t_c(n)$ if $D$ and $C$ run in the
      respective runtime on input $n$. 
    \end{definition}

    \begin{definition}
    \label{def:am*0-protocol}
      An $\AM^*_0(2)$ protocol is a triple of Turing machines $(D, C, Q)$ such
      that
      \begin{itemize}
      \item $D$ takes input $z, x, y, a,b$ and runs in time polynomial
        in $|z|$.
      \item $C$ takes input $z, x, y$ and runs in time polynomial in
        $|z|$.
      \item $Q$ takes input $z$ and runs in time polynomial in $|z|$.
      \end{itemize}
      We say that the protocol specified by $(D, C,Q)$ decides the
      language $L$ if for any $z$,
      \begin{itemize}
        \item \textbf{Completeness:} if $z \in L$, then the nonlocal
          game $G^z$ given by $Q(z)$ and $(D(z, \cdot, \cdot, \cdot, \cdot), C(z, \cdot,
          \cdot))$ has value $\omega_q(G^z) = 1$.
        \item \textbf{Soundness:} if $z \not\in L$, then the nonlocal
          game $G^z$ given by $Q(z)$ and $(D(z, \cdot, \cdot, \cdot, \cdot), C(z, \cdot,
          \cdot))$ has value $\omega_q(G^z) < 1$.
      \end{itemize}
    \end{definition}
    Definitions \ref{def:game-family-gapless} and \ref{def:am*0-protocol} are not quite compatible, because the Turing machines in the former take the index $n$ as the input, whereas the Turing machines in the latter take as input an arbitrary string $z$. We bridge the gap by defining the notion of an \emph{$n$-indexed $\AM^*_0(2)$ protocol}. It turns out that all the $\AM^*_0(2)$
    protocols we construct will be $n$-indexed protocols.
    \begin{definition}\label{def:nindexed-amstar}
      An $n$-indexed $\AM^*_0(2)[q(\cdot), a(\cdot), t_d(\cdot), t_c(\cdot)]$ protocol is a pair of Turing machines
      $GenG, N$ such
      that
      \begin{itemize}
      \item $GenG$ takes as input $z$ and in polynomial time outputs a description of a nonlocal game
        family specified by a triple $D^z, C^z, Q^z$. These expect input of the form $n,x,y,a,b$
        and $n,x,y$, respectively.
      \item $N$ takes as input $z$ and in polynomial time outputs a
        natural number $n$ with $n = \poly(|z|)$.
      \item $D^z(N(z), x, y, a, b)$ reads at most $q(|z|)$ bits of $x,
        y$ and at most $a(|z|)$ bits of $a,b$, and runs in time at
        most $t_d(|z|)$.
      \item $C^{z}(N(z), x, y)$ reads at most $q(|z|)$ bits of $x, y$
        and runs in time at most $t_c(|z|)$.
      \item $Q^z(N(z))$ returns an output of at most $q(|z|)$ and runs
        in time at most $\poly(|z|)$.
      \end{itemize}
      By default, if we refer to $\AM^*_0(2)$ without specifying $q, a, t_d, t_c$, we assume they are
      polynomial functions.
      We say that the protocol specified by $D, C$ decides the
      language $L$ if for any $z$,
      \begin{itemize}
        \item \textbf{Completeness:} if $z \in L$, then the nonlocal
          game $G$ given by $Q^z(N(z))$ and $(D^z(N(z), \cdot, \cdot, \cdot, \cdot), C^z(N(z), \cdot,
          \cdot))$ has value $\omega^s_q(G^z) = 1$. 
        \item \textbf{Soundness:} if $z \not\in L$, then the nonlocal
          game $G$ given by $Q^z(N(z))$ and $(D^z(N(z), \cdot, \cdot, \cdot, \cdot), C^z(N(z), \cdot,
          \cdot))$ has value $\omega^s_q(G^z) < 1$.
      \end{itemize}

    \end{definition}
    We remark that, if there exists an $n$-indexed $\AM^*_0(2)$ protocol to decide a language $L$, then there also exists an $\AM^*_0(2)$ protocol to decide $L$.
    
\subsubsection{Gapped game families and $\MIP^*(2)$ protocols}
	 \label{sec:prelims-gapped}
    Next, we turn to the gapped case. In this case, the distribution
    on questions \emph{does} matter, so we use the formalism of
    \cite{JNVWY20}, who define an $\MIP^*$ protocol in terms of a pair of
    Turing machines: a \emph{sampler} that, given a random seed as
    input, generates a pair of questions for the two players, and a
    \emph{decider} that given a pair of questions and answers decides
    whether they won the game. In the gapped case, the compression
    theorems that are known impose stringent conditions on the form of
    the sampler, and so we will refer the reader to \cite[Section
    5]{JNVWY20} for definitions.

    For a family of gapped games indexed by $n$, or for an $\MIP^*$
    protocol, we may speak of the question and answer length $q(n),
    a(n)$, and the \emph{sampler} and decider runtimes $t_s(n),
    t_d(n)$. We define these analogously to the gapless case with one
    important difference: we require $q(n)$ to be an upper bound both on the question length and also on the \emph{number of bits of its input} that the sampler may read (informally, the number of random bits that the sampler may use to generate the questions). In the checker formalism, the question length and the number of bits that the checker reads are the same; in the sampler formalism, they may be different. To simplify the notation, we use $q(n)$ as an upper bound on both quantities.
    
 The compression results of \cite{JNVWY20} apply only to games with \emph{normal form samplers}: these are defined there are samplers that generate question pairs using a specific type of map (a ``conditional linear function") applied to the input random seed. We will not need the details of this definition here. However, one aspect of it will be useful: a normal form sampler is defined to include the functionality that, on a special input value, returns the value of $q(n)$. This means that we no longer need a separate Turing machine $Q$ to compute this: it can be folded into the definition of the sampler $S$. Using this, we can now define a family of gapped games in terms of a sampler and decider.
 
    \begin{definition}
    \label{def:gapped-game-family}
      A \emph{family of gapped games} is specified by a pair
      $(S, D)$ of Turing machines where $(S,D)$ constitutes a \emph{normal form verifier} as defined in \cite[Section 5.4]{JNVWY20}. Slightly more specifically, $S$ and $D$ are of the following form:
      \begin{enumerate}
      	\item $S$ takes as input a natural number $n$ along with another argument $\textsf{sampler-args}$.\footnote{We abstract the other arguments to $S$ as \textsf{sampler-args} because they are complicated. For more details, see \cite[Section 4.2]{JNVWY20}.} $S$ reads at most $q(n)$ bits of \textsf{sampler-args} on input $n$, and runs in time at most $t_s(n)$. In addition, if $\textsf{sampler-args} = \textsc{Dimension}$, $S$ on input $(n, \textsc{Dimension})$ outputs $q(n)$.
      \item $D$ takes as input a natural number $n$ and a tuple $(x, y, a, b)$. $D$ reads at most $q(n)$ bits of $x,y$ and at most $a(n)$ bits of $a,b$, and runs in time at
        most $t_d(n)$.
      \end{enumerate}
      The input $n$ is called the \emph{index}. We say that the game family $(S,D)$ has \emph{question
        length} $q(n)$, \emph{answer length}
      $a(n)$, \emph{sampler runtime} $t_s(n)$, and \emph{decider runtime} $t_d(n)$.
    \end{definition}

    \paragraph{The $n$th nonlocal game in a gapped game family} For a precise definition of the \emph{$n$th nonlocal game in the game family $(S,D)$}, we refer to \cite[Section 5]{JNVWY20}. At a
    high level, the verifier in the game $G_n$ corresponding to the $n$th game in the game family specified by $G = (S,D)$ performs the following steps: (1) it sets
    $\textsf{sampler-args}$ to be $\textsc{Dimension}$ and runs $S$ on $(n, \textsc{Dimension})$ to
    compute the question size $q(n)$, (2) it runs $S$ on a random seed of length $q(n)$ to compute a pair of questions to send to the
    provers, and (3) it runs $D$ on the questions and the provers'
    answers to decide whether they have won the game.
    
    We now define a general (two-player, one-round) $\MIP^*(2)$ protocol, following \cite{JNVWY20}. Note that, in the definition below, unlike in Definition \ref{def:gapped-game-family}, we do \emph{not} require that $S,D$ constitute a normal form verifier.
    
%    \anote{This definition is only needed for the Kolmogorov complexity part. For all the protocols we can assume that $S, D$ are normal form}
%    \begin{definition}
%    \label{def:mip*-protocol}
%      An $\MIP^*(2)$ protocol \anote{improve in a later iteration} \znote{in particular, what does `refereed by' mean if $S,D$ are not necessarily normal form?} is a pair of Turing machines $S,D$ such
%      that
%      \begin{itemize}
%      \item $D$ takes input $z, x, y, a,b$ and runs in time polynomial
%        in $|z|$.
%      \item $S$ takes input $z, \textsf{sampler-args}$ and runs in time polynomial in $|z|$.
%      \end{itemize}
%      We say that the protocol specified by $S,D$ decides the
%      language $L$ if for any $z$,
%      \begin{itemize}
%        \item \textbf{Completeness:} if $z \in L$, then the nonlocal
%          game $G^z$ refereed by $S(z, \cdot), D(z, \cdot, \cdot, \cdot, \cdot)$ has value $\omega_q(G^z) = 1$.
%        \item \textbf{Soundness:} if $z \not\in L$, then the nonlocal
%          game $G^z$ refereed by $S(z, \cdot), D(z, \cdot, \cdot, \cdot, \cdot)$ has value $\omega_q(G^z) \leq 1/2$.
%      \end{itemize}
%    \end{definition}

\begin{definition}\label{def:mip*-protocol}
  A language $L$ is in $\MIP^*$ if and only if there exist two probabilistic Turing machines
  $\sampler$ and $\decider$ with the following properties.
  \begin{enumerate}
    \item \textbf{Efficiency:} For every $z\in\{0,1\}^*$ there is
      a game $\game^z = (\cal{X}, \cal{Y}, \cal{A}, \cal{B}, \mu, D)$
      such that:
      \begin{enumerate}
      \item The Turing machine $\sampler$ given input $z$ runs in time $\poly(|z|)$ and returns a
      pair $(x,y) \in \cal{X} \times \cal{Y}$ such that the distribution of $(x,y)$, over the random choices of $\sampler$, is $\mu$.
    \item The Turing machine $\decider$ given as input $z$ and a tuple $(x,y,a,b)
      \in \cal{X} \times \cal{Y} \times \cal{A} \times \cal{B}$ runs in time $\poly(|z|)$ and returns $D(x,y,a,b)$.\footnote{Note that the running time of $\decider$ should be $\poly(|z|)$, even for long inputs $a,b$. This can be ensured by having $D$ return $0$ whenever $x,y,a,b$ are too long with respect to $|z|$.}
    \end{enumerate}
  \item \textbf{Completeness:} If $z \in L$, then $\val^*(\game^z) \geq 2/3$
  \item \textbf{Soundness:} If $z \not\in L$, then $\val^*(\game^z) \leq 1/3$.
  \end{enumerate}
%  We say that the pair $(\sampler, \decider)$ form an $\MIP^*$
%  \emph{verifier} for the language $L$, and the associated family of
%  games $\game^z$ are an $\MIP^*$ protocol for $L$.

We say that the pair $(\sampler, \decider)$ form an $\MIP^*$
\emph{protocol} for the language $L$.
\end{definition}

    Definitions \ref{def:gapped-game-family} and \ref{def:mip*-protocol} are again not quite compatible: the former stipulates that the sampler be \emph{normal form} and takes as input an index $n$; the latter does not require a normal form sampler and takes as input a string $z$. The following
    definition of an \emph{$n$-indexed $\MIP^*(2)$ protocol} bridges the gap. It turns out that all the $\MIP^*(2)$
    protocols we construct will be $n$-indexed protocols.
    \begin{definition}\label{def:nindexed-mipstar}
      An $n$-indexed $\MIP^*(2)[q(\cdot), a(\cdot), t_s(\cdot), t_d(\cdot)]$ protocol is a pair of Turing machines
      $GenG, N$ such
      that
      \begin{itemize}
      \item $GenG$ takes as input $z$ and in polynomial time outputs a description of a nonlocal game
        family specified by a sampler-decider pair $S^z, D^z$ which constitutes a \emph{normal form verifier} \cite[Section 4.2]{JNVWY20}.
      \item $N$ takes as input $z$ and in polynomial time outputs a
        natural number $n$ with $n = \poly(|z|)$.
      \item For any $z$, $S^{z}$ takes as input a natural number $n$ along with another argument $\textsf{sampler-args}$. When $n = N(z)$, $S^z$ reads at most $q(|z|)$ bits of \textsf{sampler-args} and runs in time at most $t_s(|z|)$. In addition, if $\textsf{sampler-args} = \textsc{Dimension}$, $S$ on input $(N(z), \textsc{Dimension})$ outputs $q(|z|)$.
      \item For any $z$, $D^z$ takes as input a natural number $n$ and a tuple $(x, y, a, b)$. When $n = N(z)$, $D^z$ reads at most $q(|z|)$ bits of $x,y$ and at most $a(|z|)$ bits of $a,b$, and runs in time at
        most $t_d(|z|)$.
      \end{itemize}
      By default, when we refer to $\MIP^*$ without specifying $q, a,
      t_s, t_d$, we assume they are
      polynomial functions. 
      We say that the protocol specified by $S,D$ decides the
      language $L$ if for any $z$,
      \begin{itemize}
        \item \textbf{Completeness:} if $z \in L$, then the $N(z)$th nonlocal game in the game family $G^z = (S^z, D^z)$, which we denote by $G^z_{N(z)}$, has value $\omega_q(G^z_{N(z)}) = 1$.
        \item \textbf{Soundness:} if $z \not\in L$, then the $N(z)$th nonlocal game in the family family $G^z = (S^z, D^z)$, which we denote by $G^z_{N(z)}$, has value $\omega_q(G^z_{N(z)}) \leq 1/2$.
      \end{itemize}

    \end{definition}
    We remark that, if an $n$-indexed $\MIP^*(2)$ protocol exists to decide a language $L$, then an $\MIP^*(2)$ protocol exists to decide a language $L$.

\subsubsection{Relationship between gapless and gapped definitions}

The reader will note that in the gapless case, the interactive proof class we defined was called $\AM^*_0(2)$, while in the gapped case, the class we defined was called $\MIP^*$. The reason for this difference (in the use of `$\AM(2)$' vs. `$\MIP$') is that the free game formulation---in which the verifier is defined in terms of a checker and a decider instead of a sampler and a decider---is without loss of generality in the gapless case, but \emph{not} in the gapped case. Indeed, in the gapless case one may
    define a class $\MIP_0^*$ in the same way as $\MIP^*$ (i.e in
    terms of a sampler and decider), except with
    the soundness gap set to $0$.  It is clear that $\AM^*_0(2)
    \subseteq \MIP^*_0$; this is because, given a checker, one can
    define an equivalent sampler that on a random seed generates a
    uniformly random pair of questions $(x,y)$, runs the checker on
    these, and if the checker fails outputs an ``abort'' question pair $(\bot,
    \bot)$ instead. Similarly, $\MIP^*_0 \subseteq \AM^*_0(2)$: in the zero-gap case, we can without loss of deciding power assume that the question distribution is uniform and free, because any nonuniform and correlated question distribution over a finite set $\cal{S} \subseteq \cal{X} \times \cal{Y}$ can be replaced with a free question distribution that is simply uniform over $\cal{X} \times \cal{Y}$, and we will still have a nonzero probability of `hitting' the question pairs that the provers fail on. In our section on nonuniform complexity bounds, we
    will show lower bounds on $\MIP^*_0$, which will thus imply bounds
    against \emph{both} $\MIP^* $ and $\AM^*_0(2)$.
    
    In the gapped case, free games are in general less powerful than general games---except in the special case where the question length is constant. Specifically, define a gapped class $\AM^*(2)$ using the decider-checker formalism that we used to define $\AM^*_0(2)$ in \Cref{sec:prelims-gapless}, except with a constant completeness-soundness gap. In
    the case where the question size is a constant, it is clear that
    $ \MIP^*[q(n) = O(1)] \subseteq \AM^*(2)[q(n)=O(1)]$; this is because the checker can
    run the $\MIP^*$ sampler on all possible seeds (a constant
    number), and generate the list of all valid question pairs efficiently. Ultimately, in \Cref{sec:gapped} we will show a protocol for $\RE$ in $\MIP^*[q(n) = O(1)]$, which will thus imply an $\RE$ protocol in $\AM^*(2)[q(n) = O(1)]$.

\znote{expand eventually; we probably want a proper definition of $\AM^*$}

\subsection{$\QMA(2)$ and related classes}

For a fuller treatment of this class, we refer the reader to \cite{HM13}.
\begin{definition}
  $\QMA(2)$ is the class of quantum Merlin-Arthur proof systems where
  Arthur is a polynomial-time quantum machine and receives a witness
  $\ket{\psi} = \ket{\psi_1} \ot \ket{\psi_2}$
  from Merlin that is guaranteed to be in tensor product across a
  fixed cut. A $\QMA(2)$ protocol decides a language $L$ if for any
  input $x \in L$, there is a witness state $\ket{\psi_1} \ot
  \ket{\psi_2}$ that Arthur accepts with probability at least $c$ (the
  \emph{completeness probability}), and for any input $x \not\in L$,
  no witness state in tensor product form makes Arthur accept with
  probability greater than $s$ (the \emph{soundness
    probability}); when not otherwise specified, we assume $c = 2/3$
  and $s = 1/3$.
\end{definition}

\begin{definition}
  $\sf{BellQMA}(2)$ is the class of $\QMA(2)$ proof systems where the
  POVM element corresponding to the accepting measurement outcome of
  verifier is a Bell measurement (see \Cref{def:bell-meas}).
\end{definition}

    \section{A $\sf{BellQMA}(2)$ protocol for 3SAT}
    
    \subsection{The protocol}
\label{sec:protocol}
\begin{definition}[Generalised $K$-colouring]
	Let $K \in \mathbb{N}$, let $\cal{G} = (V,E)$ be a graph, and let $R: E \times [K] \times [K] \rightarrow \{0,1\}$ be a function. We say that $\cal{G}$ is generalised $K$-colourable with respect to $R$ if there exists an assignment function $c : V \rightarrow [K]$ such that, for all edges $e = (v_1, v_2) \in E$, $R(e, c(v_1), c(v_2)) = 1$.
\end{definition}

\begin{definition}[\textsf{$\delta$-GAP-$K$COL}]
\label{def:delta-gap-kcol}
	\textsf{$\delta$-GAP-$K$COL} is a promise problem. An instance of \textsf{$\delta$-GAP-$K$COL} consists of a graph $\cal{G} = (V,E)$, a number $K \in \mathbb{N}$, and a function $R: E \times [K] \times [K] \rightarrow \{0,1\}$.
	\begin{itemize}
	\item $(\cal{G}, K, R)$ is a YES-instance of \textsf{$\delta$-GAP-$K$COL} if $\cal{G}$ is generalised $K$-colourable with respect to $R$.
	\item $(\cal{G}, K, R)$ is a NO-instance of \textsf{$\delta$-GAP-$K$COL} if, for all possible assignments $c : V \rightarrow [K]$, there exist at least $\delta|E|$ edges $e \in E$ such that $R(e, c(v_1), c(v_2)) = 0$.
	\end{itemize}
\end{definition}

\begin{theorem}
\label{thm:3sat-to-kcol}
There is a reduction $f : \{0,1\}^* \rightarrow \{0,1\}^*$ from 3SAT to \textsf{$\delta$-GAP-$K$COL} with constant $\delta > 0$ and constant $K > 0$ such that, if $x$ is an $N$-clause 3SAT instance, $f(x)$ is an instance $(\cal{G}, K, R)$ of \textsf{$\delta$-GAP-$K$COL} such that $|V| = O(N \cdot \poly\log N)$ and $|E| = O(N \cdot \poly\log N)$.
\end{theorem}

\begin{proof}
See Theorem 2 of \cite{CD10}.
\end{proof}

We now present a $\sf{BellQMA}(2)$ protocol for \textsf{$\delta$-GAP-$K$COL}. The protocol relies on two sub-tests: the uniformity test (\Cref{fig:unif-test}) and the consistency test (\Cref{fig:cons-test}). We call the verifier in this protocol Arthur, and the two provers Alice and Bob.

{
\floatstyle{boxed} 
\restylefloat{figure}
\begin{figure}[H]
Fix an instance $(\cal{G} = (V,E), K, R)$ of \textsf{$\delta$-GAP-$K$COL}.

\textbf{Input:} Let $n = |V|, m = |E|$. All parties in the protocol receive the instance $(\cal{G} = (V,E), K, R)$ as input, along with an integer $k = O(\sqrt{n})$, and a constant $0 < \eta < 1$ to use in the uniformity test (\Cref{fig:unif-test}). Honest provers also receive as input a generalised $K$-colouring of $\cal{G}$, described as a function $c : V \rightarrow [K]$.

The protocol is as follows:
\begin{enumerate}
\item Alice and Bob both send Arthur a state; Alice's state is $k(\log m + 2\log K)$ qubits long, and Bob's state is $k(\log n + \log K)$ qubits long. Let the states that they send be $\ket{\psi_1}$ and $\ket{\psi_2}$. Honest provers send the states
\begin{gather*}
	\ket{\psi_1} = \left( \frac{1}{\sqrt{m}} \sum_{e = (v_1,v_2) \in E} \ket{e} \ket{c(v_1), c(v_2)} \right)^{\otimes k} \\
	\ket{\psi_2} = \left( \frac{1}{\sqrt{n}} \sum_{v \in V} \ket{v} \ket{c(v)} \right)^{\otimes k}.
\end{gather*}
\item Arthur flips a single coin. If it lands heads, he performs the uniformity test (Figure \ref{fig:unif-test}) on both $\ket{\psi_1}$ and $\ket{\psi_2}$, setting $\eta$ to be the choice of $\eta$ that was provided to him as input. The uniformity test also takes two natural number parameters, $K'$ and $Q$. For the uniformity test on $\ket{\psi_1}$, he sets \[K' = K^2, Q = m,\] and for the uniformity test on $\ket{\psi_2}$, he sets \[K' = K, Q = n.\] If it lands tails, Arthur performs the consistency test (Figure \ref{fig:cons-test}) on $\ket{\psi_1} \otimes \ket{\psi_2}$, setting $\cal{G}, K, R$ to be the choices which were provided to him as input.
\end{enumerate}
\caption{The $\sf{BellQMA}(2)$ protocol for \textsf{$\delta$-GAP-$K$COL}. \label{fig:protocol}}
\end{figure}
}

\begin{lemma}[Completeness]
If $\cal{G}$ is generalised $K$-colourable with respect to $R$, then the honest strategy outlined in \Cref{fig:protocol} is accepted with probability $1 - \exp(-\Omega(\sqrt{n}))$.
\end{lemma}

\begin{proof}
The consistency test accepts with probability 1 when $\cal{G}$ is generalised $K$-colourable and the two provers are honest. According to the analysis in \cite[Section 3.1]{CD10}, the uniformity test on $\ket{\psi_1}$ and the uniformity test on $\ket{\psi_2}$ each pass with probability $1 - \exp(-\Omega(\sqrt{n}))$ when the provers are honest. A union bound gives the desired conclusion.
\end{proof}

In the following sections, we analyse the soundness of the protocol.

{
\floatstyle{boxed} 
\restylefloat{figure}
\begin{figure}[htpb]
\textbf{Input:} Two numbers $K', Q \in \mathbb{N}$, another number $0 < \eta < 1$, and a state $\ket{\psi}_{\reg{Q}_1 \reg{A}_1 \dots \reg{Q}_k \reg{A}_k}$ on registers $\reg{Q}_1 \reg{A}_1 \dots \reg{Q}_k \reg{A}_k$. The registers $\reg{Q}_i$ are called the \emph{question} registers and the registers $\reg{A}_i$ are called the \emph{answer} registers.
\begin{enumerate}
    \item Perform a Fourier transform $\calF_{K'}$ on each answer register and then measure it in the standard basis.
    \item Let $Z = \{i : \text{the answer register $\reg{A}_i$ measured to 0}\}$. If $\frac{|Z|}{k} < (1-\eta)\frac{1}{K'}$, reject; otherwise, continue.$^1$
    \item For each answer register $\reg{A}_i$ that measured to 0 in step 1, perform a Fourier transform $\calF_{Q}$ on the $i$th question register $\reg{Q}_i$, and measure it in the standard basis. If any non-zero measurement outcome is obtained at this step, reject. Otherwise, accept.
\end{enumerate}

$^1$ \footnotesize $\eta$ is necessary because even honest Merlins will not always pass in the uniformity test; instead, they will only be able to achieve an \emph{average} of $|Z| = \frac{k}{K'}$, so $\eta$ is necessary to be able to perform a Chernoff bound and achieve $1 - \exp(-k)$ completeness. See \cite[Section 3.1]{CD10} for more details.
\caption{The uniformity test. \label{fig:unif-test}}
\end{figure}
}

{
\floatstyle{boxed} 
\restylefloat{figure}
\begin{figure}[htpb]
\textbf{Input:}
\begin{itemize}
\item A state $\ket{\psi}_{\reg{Q}_1 \reg{A}_1 \dots \reg{Q}_k \reg{A}_k} \otimes \ket{\psi'}_{\reg{Q}'_1 \reg{A}'_1 \dots \reg{Q}'_\ell \reg{A}'_\ell}$ on registers $\reg{Q}_1 \reg{A}_1 \dots \reg{Q}_k \reg{A}_k \reg{Q}'_1 \reg{A}'_1 \dots \reg{Q}'_\ell \reg{A}'_\ell$. The registers $\reg{Q}_i$ and $\reg{Q}'_j$, $i \in [k], j \in [\ell]$, are called the \emph{question} registers, and the registers $\reg{A}_i$ and $\reg{A}'_j$ are called the \emph{answer} registers.
\item A graph $\cal{G} = (V,E)$.
\item A number $K \in \mathbb{N}$.
\item A relation $R: [K] \times [K] \rightarrow \{0,1\}$.
\end{itemize}
\begin{enumerate}
    \item Measure all the registers in the standard basis. Interpret each measurement outcome in a register $\reg{Q}_i$, $i \in [k]$, as a question for Alice, and interpret the measurement outcome coming from the associated answer register $\reg{A}_i$ as her answer to that question. (Therefore, Alice receives $k$ questions and answers each one.) Interpret each measurement outcome in a register $\reg{Q}'_j$, $j \in [\ell]$, as a question for Bob, and interpret the measurement outcome coming from the associated answer register $\reg{A}'_j$ as his answer to that question.
    \item Interpret each Alice question as an edge $e \in
      E$, and interpret the corresponding Alice answer as a pair of
      colours in $[K]$ for the vertices that form the endpoints of
      $e$. Interpret each Bob question as a vertex $v \in 
      V$, and interpret the corresponding Bob answer as a colour in
      $[K]$ for $v$. Let $A \subseteq E$ be the set of all edges obtained as Alice questions
      and $B \subseteq V$ be the set of all vertices obtained as Bob questions. 
    \item For every edge $e \in A$ and vertex $v \in B$ such that $v \in e$, check that Alice's and Bob's colorings
    agree and that the two endpoints of $e$ are assigned colours that satisfy the function $R(e, \cdot, \cdot)$.
\end{enumerate}
\caption{The consistency test. \label{fig:cons-test}}
\end{figure}
}

\subsection{Soundness of uniformity test}

%\anote{Change terminology from "vertex" and "color" to "question" and "answer". This test should be sound whenever the question alphabet is poly-sized and complete when the answer alphabet is constant sized as well.}
%
%\anote{Need to argue completeness. The completeness will determine the size of the sets $T$.}
For illustrative purposes, we begin with a zero-error analysis of the uniformity test. In the proof of \Cref{lem:unif-general}, we will show how the argument presented below generalises to the case of nonzero error.
\begin{lemma}\label{lem:unif-zero-error}
  Suppose $\ket{\psi}$ passes the uniformity test with certainty. Then there exists a collection $\cal{S}$ of subsets of $[k]$ such that, \begin{enumerate}
      \item for all $T \in \cal{S}$, it holds that $|T| \geq \frac{k}{K'}(1-\eta)$, and
      \item the distribution $\mu_Q$ which results from measuring the question registers of $\ket{\psi}$ in the standard basis can be
  decomposed as a mixture 
  \[ \mu_Q = \sum_{T \in \cal{S}} p(T) \mu^{unif}_T \ot \mu^{junk}_{\overline{T}}, \]
  where $p : \cal{S} \rightarrow [0,1]$ is a distribution over
  $\cal{S}$, $\mu^{unif}_T$ is the uniform distribution over $[Q]^{|T|}$ on the
  indices in $T$, and $\mu^{junk}_{\overline{T}}$ is an arbitrary
  distribution on the indices in $[k] - T$. 
  \end{enumerate}
\end{lemma}
\begin{proof}
  Suppose we perform the first step of the uniformity test (\Cref{fig:unif-test}) on $\ket{\psi}$, i.e., we measure all the answer registers of $\ket{\psi}$ in the Fourier
  basis. Let $\rho_{\vec{r}}$ denote the post-measurement state after this measurement conditioned on getting outcome $\vec{r}$. Assuming that $\ket{\psi}$ passes the uniformity test with
  certainty, $\vec{r}
  = r_1,
  \dots, r_k$ must be such that $r_i = 0 \: \forall i \in T$ for some subset $T \subset [k]$ with $|T| \geq
  \frac{k}{K'}(1-\eta)$. Moreover, the probability that $\rho_{\vec{r}}$ now
  passes step 3 of the uniformity test is still $1$. Therefore, 
  \[ \rho_{\vec{r}} = (\ket{\bar{0}}\bra{\bar{0}})^{\ot T} \ot
    \rho_{\overline{T}}, \]
  where the notation $(\ket{\bar{0}}\bra{\bar{0}})^{\ot T}$ means that the registers with indices in
  $T$ are in the all-zero state in the Fourier basis and in tensor product with the other
  registers.

  Thus, measuring $\rho_{\vec{r}}$ in the standard basis will
  yield uniformly random iid outcomes on the registers in $T$ and
  some arbitrary distribution on the other registers.

  Finally, to get the lemma, observe that (letting $\rho$ denote the post-measurement state after the Fourier measurement of step 1 with no conditioning)
  \[ \rho = \sum_{\vec{r}} q_{\vec{r}} \:
    \rho_{\vec{r}}, \]
    for some distribution $q_{\vec{r}}$.
  Thus, the conclusion follows.
\end{proof}

We now proceed to the main technical lemma in this section, which is a
version of \Cref{lem:unif-zero-error} that tolerates constant error.

\begin{lemma}\label{lem:unif-general}
  Suppose $\ket{\psi}$ passes the uniformity test with probability $1
  - \eps > 0$. Then there exists a collection $\cal{S}$ of subsets of $[k]$ such that, \begin{enumerate}
      \item for all $T \in \cal{S}$, it holds that $|T| \geq \frac{k}{K'}(1-\eta)$, and
      \item the distribution $\mu_Q$ which results from measuring the question registers of $\ket{\psi}$ in the standard basis can be
  decomposed as a mixture 
  \[ \mu_Q \simeq_{\delta(\eps)} \sum_{T \in \cal{S}} p(T) \mu^{unif}_T \ot \mu^{junk}_{\overline{T}}, \]
  where
  \begin{enumerate}
      \item $p : \cal{S} \rightarrow [0,1]$ is a distribution over $\cal{S}$,
      \item $\mu^{unif}_T$ is the uniform distribution over $[Q]^{|T|}$ on the indices in $T$,
      \item $\mu^{junk}_{\overline{T}}$ is an arbitrary distribution on the indices in $[k] - T$,
      \item the notation $\simeq_\delta$ indicates that the two sides are
  a distance of $\delta$ apart in total variational distance, and
      \item $\delta(\eps) = O(\eps^{1/4})$.
  \end{enumerate}
  \end{enumerate}
\end{lemma}
\begin{proof}

    Suppose we perform the first step of the uniformity test (\Cref{fig:unif-test}) on $\ket{\psi}$, i.e., we measure all the answer registers of $\ket{\psi}$ in the Fourier
  basis. Let $\rho_{\vec{r}}$ denote the post-measurement state after this measurement conditioned on getting outcome $\vec{r}$, and let $\rho$ denote the overall post-measurement state after this measurement without conditioning on any particular outcome. Let $q_{\vec{r}}$ denote the probability of obtaining any given outcome $\vec{r}$.
  
%   Assuming that $\ket{\psi}$ passes the uniformity test with
%   probability $1 - \eps$, $\vec{r}
%   = r_1,
%   \dots, r_k$ must be such that $r_i = 0 \: \forall i \in T$ for some subset $T \subset [k]$ with $|T| \geq
%   \frac{k}{K}(1-\eta)$, with probability $1 - \eps$.

  Let $p_{success, \vec{r}}$ be a function mapping density matrices to $[0,1]$ such that $p_{success, \vec{r}}(\sigma)$ gives the probability that a given mixed state $\sigma$ passes when it is subjected to step 3 of the uniformity test and $\vec{r}$ was the outcome obtained in step 1 of the uniformity test. Let $\Id_{\vec{r}}$ be an indicator function which indicates whether or not a given vector $\vec{r}$ passes step 2 of the uniformity test (i.e. whether or not $\vec{r}$ is such that there exists $T \subseteq [k]$, with $|T| \geq
  \frac{k}{K'}(1-\eta)$, for which $r_i = 0 \: \forall i \in T$). Using this notation, the probability that $\ket{\psi}$ passes in the uniformity test can then be expressed as
  \[ \sum_{\vec{r}} q_{\vec{r}} \cdot \Id_{\vec{r}} \cdot
    p_{success, \vec{r}}(\rho_{\vec{r}}) \geq 1 - \eps
    .\]

  Rewrite as
  \[ \sum_{\vec{r}} q_{\vec{r}} (1 - \Id_{\vec{r}} \cdot
    p_{success, \vec{r}}(\rho_{\vec{r}} ))\leq \eps
    . \]
  Therefore (using a Markov bound), with probability at least $1 - \frac{1}{\alpha}$, $\vec{r}$ obtained
in step 1 is such that
  \begin{equation} (1 - \Id_{\vec{r}} \cdot
    p_{success}(\rho_{\vec{r}}) ) \leq \alpha \eps. \label{eq:r-good} \end{equation}
  Let us set $\alpha = \frac{1}{\sqrt{\eps}}$, and define any such
  $\vec{r}$ to be \emph{good}. With this definition of $\alpha$, $\vec{r}$ is good
  with probability at least $1 - \sqrt{\eps}$. Note that, for any good
  $\vec{r}$, step 2 of the uniformity test passes with
  certainty (or else $\Id_{\vec{r}} = 0$ and \Cref{eq:r-good} would become $1 \leq \sqrt{\eps}$), and step 3 of the uniformity test applied to $\rho_{\vec{r}}$ passes with
  probability at least $1 - \sqrt{\eps}$.

  Let $\rho_{question | \vec{r}}$ be $\rho_{\vec{r}}$ restricted to its question registers. For any fixed good $\vec{r}$ (for which step 3 of the uniformity test applied to $\rho_{\vec{r}}$ passes with
  probability at least $1 - \sqrt{\eps}$), we have
  \[ \tr[ ((\ket{\bar{0}}\bra{\bar{0}})^{\ot T} \ot
    I_{\overline{T}})\rho_{question | \vec{r}}] \geq 1 - \sqrt{\eps}, \]
  where the notation $(\ket{\bar{0}}\bra{\bar{0}})^{\ot T}$ means that the registers with indices in
  $T$ are in the zero Fourier state and in tensor product with the other
  registers. Thus, by the Gentle Measurement Lemma~\cite[Lemma 9.4.1]{Wilde11}, it holds that
  \begin{equation}
  \label{eq:fixed-r-distance}
  	\| \rho_{question | \vec{r}} - \underbrace{(\ket{\bar{0}}\bra{\bar{0}})^{\ot T}
      \ot \sigma(\vec{r})_{\overline{T}}}_{\sigma(\vec{r})} \|_1 \leq 2 \eps^{1/4}.
  \end{equation}
  By construction, measuring $\sigma(\vec{r})$ in the standard basis
  will yield a distribution $\mu^{\vec{r}}$ that is uniformly random iid outcomes on the registers in $T$ and
  some arbitrary distribution on the other registers.

  Thus, measuring $\rho_{question |\vec{r}}$ in the standard basis (for any good $\vec{r}$) will
  yield a distribution that is $O(\eps^{1/4})$-close to $\mu^{\vec{r}}$ in total
  variational distance, by the relation between variational distance
  and trace distance \cite[Theorem 9.1]{NC02}.

  Let $\rho_{question}$ denote the state $\rho$ (defined in the first paragraph of this proof) restricted to its question registers. To argue about the distribution we obtain by measuring $\rho_{question}$ without the
  conditioning on a fixed good $\vec{r}$, observe that
  \begin{equation} \rho_{question} = \sum_{\vec{r}} q_{\vec{r}} \:
    \rho_{question|\vec{r}} = \sum_{\vec{r} \in BAD} q_{\vec{r}} \:
    \rho_{question|\vec{r}} + \sum_{\vec{r} \in GOOD} q_{\vec{r}} \: \rho_{question|\vec{r}}, \end{equation}
    and recall that $\vec{r}$ is good with probability at least $1 - \sqrt{\eps}$. Given this, there exists a state with no weight on $\rho_{question|\vec{r}}$s with $\vec{r}$ in $BAD$ which is at most $O(\sqrt{\eps})$ from $\rho_{question}$ in trace distance. Formally, if we define a new state
    \[ \rho_{question}' = \sum_{\vec{r} \in GOOD} q_{\vec{r}} \: \rho_{question|\vec{r}} + \left(\sum_{\vec{r} \in BAD} q_{\vec{r}} \right) \rho_{question|\color{BurntOrange} \vec{r}^*}, \]
    where $\vec{r}^*$ is an arbitrary (for concreteness, the lexicographically first) $\vec{r}$ in $GOOD$, we have that
    \begin{equation} \|\rho_{question}' - \rho_{question} \|_1 = O(\eps^{1/2}). \end{equation}
    Meanwhile, note that $\rho_{question}'$ can be expressed as a sum
    \begin{equation} \rho_{question}' = \sum_{\vec{r} \in GOOD} q_{\vec{r}}' \: \rho_{question|\vec{r}}, \end{equation}
    where $q_{\vec{r}}'$ is some distribution over $\vec{r}$.
    For any good $\vec{r}$, let $T(\vec{r})$ denote a set such that $T \subseteq [k]$, $|T| \geq
  \frac{k}{K'}(1-\eta)$, $r_i = 0 \: \forall i \in T$. By the strong convexity of the trace distance \cite[Theorem 9.3]{NC02} and \Cref{eq:fixed-r-distance}, we have that
    \begin{equation} \Big\| \rho_{question}' - \sum_{\vec{r} \in GOOD} q_{\vec{r}}' \left( (\ket{\bar{0}}\bra{\bar{0}})^{\ot T(\vec{r})}
      \ot \sigma(\vec{r})_{\overline{T(\vec{r})}} \right) \Big\|_1 \leq \sum_{\vec{r} \in GOOD} q_{\vec{r}}' \cdot 2\eps^{1/4}. \end{equation}
       Therefore,
      \begin{equation} \Big\| \rho_{question}' - \sum_{\vec{r} \in GOOD} q_{\vec{r}}' \left( (\ket{\bar{0}}\bra{\bar{0}})^{\ot T(\vec{r})}
      \ot \sigma(\vec{r})_{\overline{T(\vec{r})}} \right) \Big\|_1 \leq 2\eps^{1/4}. \end{equation}
      By the triangle inequality, then,
      \begin{equation} \Big\| \rho_{question} - \sum_{\vec{r} \in GOOD} q_{\vec{r}}' \left( (\ket{\bar{0}}\bra{\bar{0}})^{\ot T(\vec{r})}
      \ot \sigma(\vec{r})_{\overline{T(\vec{r})}} \right) \Big\|_1 = O(\eps^{1/4}) + O(\eps^{1/2}) = O(\eps^{1/4}). \end{equation}
      Finally, by the contractivity of the trace distance under completely positive trace-preserving maps \cite[Theorem 9.2]{NC02}, measuring both $\rho_{question}$ and $\sum_{\vec{r} \in GOOD} q_{\vec{r}}' \left( (\ket{\bar{0}}\bra{\bar{0}})^{\ot T(\vec{r})}
      \ot \sigma(\vec{r})_{\overline{T(\vec{r})}} \right)$ in the standard basis will not increase the trace distance between them. Measuring the latter in the standard basis manifestly results in a distribution of the form
      \[ \sum_{T \in \cal{S}} p(T) \mu^{unif}_T \ot \mu^{junk}_{\overline{T}} \]
      for $\cal{S}$ the set $\{ T : \exists \vec{r} \in GOOD \text{ s.t. } T = T(\vec{r}) \}$.
  Thus, the conclusion follows.

\end{proof}

\subsection{Soundness of consistency test and soundness of main protocol}

We begin by making a few definitions.

\begin{definition}[Consistency game]
\label{def:consis-game}
For any given graph $\cal{G} = (V,E)$, natural number $K$, and relation $R : [K] \times [K] \rightarrow \{0,1\}$, we define the \emph{$(k, \ell)$ consistency game}, denoted $G^{k, \ell}(\cal{G},K,R)$ or simply $G^{k, \ell}$ when the parameters are clear from context, to be the following \emph{classical} two player free game. (Note that this game is identical to the $k, \ell$ birthday repetition game from \cite{AIM}.)

\begin{itemize}
  \item Alice receives a uniformly random size-$k$ subset $A$ of the set of edges $E$, and Bob receives a uniformly random size-$\ell$ subset $B$ of the set of vertices $V$.
  \item Alice responds with a colouring of all the vertices that are at the endpoints of edges in $A$ (i.e. Alice gives a number in $[K]$ for every vertex that is at the end of some edge in $A$), and Bob responds with a colouring of all the vertices in $B$.
  \item For every edge $e \in A$ and vertex $v \in B$ such that $v \in e$, Arthur checks that Alice and Bob's colorings
    agree and that the colours assigned to the two endpoints of $e$ satisfy the relation $R(e, \cdot, \cdot)$.
  \end{itemize}
\end{definition}

    \begin{definition}[Free game with special question distribution]
    Let $G$ be a two-player free game where question pairs are uniformly sampled from a question set $X \times Y$, and let $\calD$ be a distribution over $X \times Y$. Then $G|_{\calD}$ denotes the game $G$ where the question pairs are sampled according to $\calD$.
    \end{definition}

  We would like to prove the soundness of the protocol from section \ref{sec:protocol} by reducing its soundness
  to that of the consistency game from \Cref{def:consis-game}, which was already analysed as the `birthday game' in \cite{AIM}. This means that, given a strategy for our
  QMA(2) protocol (i.e. a pair of witness states) from section \ref{sec:protocol}, we would like to
  construct a strategy for the consistency game. The statement we want to prove is formalised in the following lemma.

%   \anote{Old: We will do so as
%   follows:
%   \begin{itemize}
%     \item For a given question set $I$, Alice postselects on the
%       question register of her witness state $\ket{\psi_A}$ containing
%       the set $I$ (in any location).
%     \item Next, Alice measures her whole state in the standard
%       basis. Thanks to the postselection, there is guaranteed to be a
%       set of indices $T$ in the question registers that yield the
%       outcomes $I$.
%     \item Alice answers by looking at the color outcomes in the
%       corresponding color registers to $T$.
%     \end{itemize}
%     We would like to argue that this strategy succeeds with high
%     probability by relating its value to the value of $\ket{\psi_A},
%     \ket{\psi_B}$. How?????? TBD}
    
    \begin{lemma}\label{lem:cons-soundness}
    Let $(\cal{G}, K, R)$ be an instance of \textsf{$\delta$-GAP-$K$COL}. Suppose the two states $\ket{\psi_1}, \ket{\psi_2}$ are accepted in the protocol of \Cref{fig:protocol} with probability at least $1 - \eps$. Then there exists a strategy for $G^{k', k'}(\cal{G}, K, R)$, $k' = \frac{k}{K'}(1-\eta)$, with value $1 - O(\eps^{1/4})$.
    \end{lemma}

%\begin{definition}[Consistency game with special question distribution]
%We define the \emph{$(k, \ell)$ consistency game with question distribution $\cal{D}$}, denoted $G^{k, \ell} |_{\cal{D}}$, to be the consistency game from Definition \ref{def:consis-game}, except that the questions for Alice and Bob are chosen from a distribution $\cal{D} = \cal{D}_A \times \cal{D}_B$ instead of a uniformly random distribution.
%\end{definition}

We delay the proof of Lemma \ref{lem:cons-soundness} until we have proven Lemmas \ref{lem:cons-classical-strat}, \ref{lem:cons-decompose-D}, and \ref{lem:cons-discard-questions}. It is clear, however, that Lemma \ref{lem:cons-soundness} taken together with the following lemma, Lemma \ref{lem:cons-birthday-soundness}, yields the desired constant soundness for the protocol of section \ref{sec:protocol}.

  \begin{lemma}\label{lem:cons-birthday-soundness}
    Suppose $\cal{G} = (V,E)$ is a graph with $n$ vertices and $\Theta(n)$ edges, $K \in \mathbb{N}$ is a constant, and $R : E \times [K] \times [K] \rightarrow \{0,1\}$ is a function. Suppose that any generalised $K$-coloring of $\cal{G}$ with respect to $R$ violates at
    least $\delta$-fraction of the edges for some constant $\delta > 0$. Then the
    classical value of the consistency game $G^{k, \ell}(\cal{G})$ for $k = \ell = \Omega(\sqrt{n})$ is at most $1 - c$ for a constant $c > 0$.
  \end{lemma}
  \begin{proof}
  This follows from Theorem 26 of~\cite{AIM}.
  \end{proof}
  
  From now on in this section, we will fix an instance $(\cal{G}, K, R)$ of \textsf{$\delta$-GAP-$K$COL}, and omit the parameters in our notation for the consistency game $G^{k, \ell}(\cal{G}, K, R)$.
    
    \begin{lemma}\label{lem:cons-classical-strat}
    For any pair of states $\ket{\psi_1}, \ket{\psi_2}$ that are accepted by the $\sf{QMA}(2)$ verifier Arthur in the consistency test from Figure \ref{fig:cons-test} with probability $1 -\nu$, there exists a product distribution $\calD(\psi_1, \psi_2) = \calD^A \ot \calD^B$ over question pairs in $G^{k, \ell}$ and a randomized classical strategy achieving value $1 - \nu$ on $G^{k,\ell}|_{\calD(\psi_1, \psi_2)}$.
    \end{lemma}
    \begin{proof}
    To obtain $\calD$ and the classical strategy, perform a standard basis measurement of $\ket{\psi_1}$ and $\ket{\psi_2}$.
    \end{proof}

%    Now, suppose that we have passed the
%    uniformity test with very high probability so the question
%    distribution is a mixture of partially iid distributions as shown
%    above. 
The following is a restatement of Lemma \ref{lem:unif-general}.
    \begin{lemma}\label{lem:cons-decompose-D}
    For $\ket{\psi_1}, \ket{\psi_2}$ each passing the uniformity test (Figure \ref{fig:unif-test}) with probability $1 - \nu$, the distribution $\calD(\psi_1, \psi_2)$ obtained by measuring the question registers of $\ket{\psi_1}$ and $\ket{\psi_2}$ in the standard basis is of the form
    \[ \calD(\psi_1, \psi_2) = \calD^A\ot \calD^B, \]
    where, for $W \in \{A,B\}$, $\calD^W$ has a decomposition of the form
    \[ \calD^{W} = \sum_{T: T \subseteq [k], |T| \geq \frac{k}{K'}(1-\eta)} p^{W}({T}) \calD^{unif}_T \ot \calD^{junk, W}_{\overline{T}} + \calD^{error, W}, \]
    where $p^W$ is a distribution mapping subsets $T \subseteq [k]$ to $[0,1]$, and $\| \calD^{error, W}\|_1 \leq \delta(\nu) = O(\nu^{1/4})$.
    \end{lemma}

    Intuitively, Lemma \ref{lem:cons-decompose-D} says that, if Alice and Bob provide states $\ket{\psi_1}, \ket{\psi_2}$ which pass the uniformity test (Figure \ref{fig:unif-test}) with high probability, then the distribution over questions for $G^{k, \ell}$ which is obtained by measuring the question registers of $\ket{\psi_1}$ and $\ket{\psi_2}$ in the standard basis can be expressed (on each of Alice's and Bob's sides) as a convex mixture of distributions, such that most of the distributions making up this convex mixture are uniform on some significant fraction of their indices, and the rest of the distributions (the $\cal{D}^{error}$ ones) are arbitrary.
    
    We now prove a lemma which will allow us to reduce the soundness of the consistency game with questions sampled in such a way to the soundness of a smaller instance of the consistency game with uniformly random questions.
%    
%    Observe that if we reveal to Alice and Bob which term in
%    the mixture they are playing, their value cannot
%    decrease.
%    \begin{definition}[Explicit version of a game]
%    	 Let $G$ be a game with question distribution $\calD = \sum_i p_i \calD_i$ and decision predicate $V$, and let $\mathcal{S}$ be a randomized classical strategy achieving value $\omega$. Let $G'$ be the following game, which we will refer to as the \emph{explicit version} of $G$:
%    \begin{itemize}
%        \item First, the verifier samples $i \sim p_i$.
%        \item Next, the verifier samples a question pair $(x,y) \sim \calD_i$.
%        \item The verifier sends $(x,i)$ to Alice and $(y,i)$ to Bob, and receives answers $a, b$ respectively.
%        \item Verifier accepts if $V(x,y,a,b)$ accepts.
%    \end{itemize}
%    \end{definition}
%    \begin{lemma}\label{lem:cons-reveal-mixture}
%    $G'$, the explicit version of $G$, has randomized classical value at least $\omega(G)$.
%    \end{lemma}
%    \begin{proof}
%    Alice and Bob can ignore $i$ and play according to their original strategy. The marginal question distribution on $(x,y)$ in $G'$ is just $\calD$. So the probability that this strategy wins is exactly $\omega$.
%    \end{proof}
    
    \begin{lemma}\label{lem:cons-discard-questions}
    Let $\calD = \calD^A \ot \calD^B$, where $\calD^A$ is of the form
    \[ \sum_{T: T \subseteq [k], |T| = k'} p^A(T) \calD_T^{unif} \ot \calD_{\overline{T}}^{junk, A} \]
    and $\calD^B$ is of the form
    \[ \sum_{T: T \subseteq [\ell], |T| = \ell'} p^B(T) \calD_T^{unif} \ot \calD_{\overline{T}}^{junk, B} \]
    with $k' \leq k, \ell' \leq \ell$. Then
    \[ \omega(G^{k, \ell}|_{\calD}) \leq \omega(G^{k', \ell'}). \]
    \end{lemma}
    \begin{proof}
    Fix $S$, a strategy for the game $G^{k, \ell}|_{\cal D}$. Any strategy $S$ for $G^{k, \ell}|_{\cal D}$ automatically induces a strategy $S'$ for $G^{k', \ell'}$. Concretely, this induced strategy works as follows: given an Alice question $x'$ from $G^{k', \ell'}$, Alice samples a set $T$ according to the distribution $p(T)$, samples an $x''$ from $\calD_{\overline{T}}^{junk}$, and sets $x = x' \| x''$. She then samples an answer $a$ for the question $x$ using her strategy for $G^{k, \ell}|_{\calD}$. $a$ will necessarily assign colours to all the endpoints of edges in the set of edges represented by $x'$; Alice responds with these colours only, which we will denote by $a'$. Bob does likewise, receiving question $y'$, sampling question $y''$ to form $y = y' \| y''$, obtaining answer $b$ to question $y$, and returning $b'$, the restriction of $b$ to that part which is relevant to $y'$. 
    
    We can analyse the success probability of $S'$ relative to that of $S$ through a series of hybrids.
    \begin{enumerate}
    \item In the first hybrid, the strategy $S$ is played in $G^{k,\ell}|_{\calD}$. Suppose that $S$ has a success probability of $p$ in $G^{k,\ell}|_{\calD}$.
    \item In the second hybrid, we define a new strategy $R$ for $G^{k,\ell}|_{\cal{D}}$. The new strategy works as follows: Alice samples a question $x'$ from the question distribution of $G^{k', \ell'}$, samples a set $T$ according to the distribution $p(T)$, samples an $x''$ from $\calD_{\overline{T}}^{junk}$, and sets $x = x' \| x''$. She then plays strategy $S$ on question $x$. Bob does likewise, sampling question $y'$, sampling question $y''$ to form $y = y' \| y''$, and playing strategy $S$ on $y$. The form of $\calD_W$ for $W \in \{A,B\}$ means that the questions $x$ and $y$ in this hybrid are distributed exactly as they would be in $G^{k,\ell}|_{\calD}$. Therefore, the success probability of $R$ is still $p$.
    \item In the third hybrid, Alice and Bob play the `induced strategy' $S'$ outlined in the first paragraph of this proof in the game $G^{k',\ell'}$. Note that, for all possible questions $(x', y')$ in $G^{k', \ell'}$ and all possible embeddings of $(x', y')$ into questions $(x,y)$ in $G^{k,\ell}|_\calD$, the checks that the rules of $G^{k', \ell'}$ require Arthur to perform on $a' \subseteq a, b' \subseteq b$ form a subset of the checks that the rules of $G^{k, \ell}|_\calD$ require Arthur to perform on $a,b$. The latter means that, for any valid question $(x',y')$ in $G^{k', \ell'}$ and any $(x,y)$ induced by $(x',y')$ according to the procedure in the first paragraph of this proof, a winning answer to question pair $(x,y)$ induces a winning answer to $(x',y')$. Therefore, the success probability can only increase between the last hybrid and this one.
    \end{enumerate}

We conclude that, if there is a strategy $S$ for $G^{k,\ell}|_{\calD}$ which wins with probability $p$, then the strategy $S'$ for $G^{k',\ell'}$ induced by $S$ also wins with probability at least $p$.
    \end{proof}
    
Now we are ready to prove Lemma \ref{lem:cons-soundness}. For convenience, we restate Lemma \ref{lem:cons-soundness} below as Lemma \ref{lem:cons-soundness-2}.

   \begin{lemma}\label{lem:cons-soundness-2}
    Let $(\cal{G}, K, R)$ be an instance of \textsf{$\delta$-GAP-$K$COL}. Suppose the two states $\ket{\psi_1}, \ket{\psi_2}$ are accepted in the protocol of \Cref{fig:protocol} with probability at least $1 - \eps$. Then there exists a strategy for $G^{k', k'}(\cal{G}, K, R)$, $k' = \frac{k}{K'}(1-\eta)$, with value $1 - O(\eps^{1/4})$.
    \end{lemma}

    \begin{proof}
    Let $\calD = \calD^A \ot \calD^B$ and $\cal{S}$ be, respectively, the product distribution over questions in $G^{k, k}$ and the randomised classical strategy succeeding with probability $1-2\eps$ in $G^{k,k}|_{\calD}$ which arise from applying \Cref{lem:cons-classical-strat} to $\ket{\psi_1}, \ket{\psi_2}$.
    
    Now, applying \Cref{lem:cons-decompose-D} to $\calD^W$ for $W \in \{A,B\}$, we obtain decompositions
      \[ \calD^{W} = \sum_{T: T \subseteq k, |T| \geq \frac{k}{K'}(1-\eta)} p^{W}(T) \calD^{unif}_T \ot \calD^{junk, W}_{\overline{T}} + \calD^{error, W}, \]
    where the error terms have bounded 1-norm at most $\delta(\eps) = O(\eps^{1/4})$. 
    
    We can, without loss of generality, rewrite such a decomposition as
    \[ \calD^{W} = \sum_{T: T \subseteq k, \color{BurntOrange}|T| = \frac{k}{K'}(1-\eta)} p^{W}(T) \calD^{unif}_T \ot \calD^{junk, W}_{\overline{T}} + \calD^{error, W}. \]
    
    Since the error terms in $\calD^W$ have probability mass at most $\delta(\eps)$, the probability mass of the subgames in $G^{k, k}|_\calD$ where either the Alice or the Bob questions are drawn from $\calD^{error, W}$ is at most $2 \cdot \delta(\eps)$. We may thus discard these terms from $\calD^W$ at the cost of changing the value of the game $G^{k, k}|_{\calD}$ by $O(\delta(\eps))$, obtaining a new game $G^{k, k}|_{\cal{D}'}$ where the question distribution $\calD' = (\calD^A)' \otimes (\calD^B)'$ is such that $(\calD^A)'$ and $(\calD^B)'$ are mixtures over only the ``good" terms, and which has value at least $1 - 2\eps - O(\delta(\eps))$.
    
    To proceed, let us apply \Cref{lem:cons-discard-questions} to $G^{k, k}|_{\cal{D}'}$. This tells us that $\omega(G^{k, k}|_{\cal{D}'}) \leq \omega(G^{k', k'})$, with $k' = \frac{k}{K'}(1-\eta)$. Therefore,
    \begin{align*}
     1 - 2\eps - O(\delta(\eps)) &\leq  \omega(G^{k', k'}).
    \end{align*}
    This concludes the proof of the lemma.
    \end{proof}
    
    Moreover, by \Cref{lem:cons-birthday-soundness}, $\omega(G^{k', k'})$ is at most $1 - c$ when $k' = \Omega(\sqrt{n})$. Putting this together with Lemma \ref{lem:cons-soundness}, we obtain that the success probability of a cheating Merlin in the protocol of section \ref{sec:protocol} is at most $1 - c'$ for some constant $c' > 0$: the soundness of the consistency game $\omega(G^{k', k'})$, as expressed in \Cref{lem:cons-birthday-soundness}, requires that
    \[ 1 - O(\eps^{1/4}) \leq 1 - c, \]
    and therefore
    \[ \eps \geq \Omega(c^{4}).\]

    % Moreover, for each term in the mixture, the soundness of
    % the Birthday Game upper bounds the value. Putting this together
    % gives us an upper bound on the value of the QMA(2) game.

    \subsection{Lower bounds for $h_{\mathrm{Sep}}$ conditional on
      ETH}
    \begin{definition}[$h_{\mathrm{Sep}}$]
      Let $\cal{H}_{AB} := \cal{H}_A \otimes \cal{H}_B$, where $\cal{H}_A$ and $\cal{H}_B$ are finite dimensional Hilbert spaces. Let $\cal{S}_{AB}$ be the set of separable states on $\cal{H}_{AB}$, i.e. the set of density matrices $\rho_{AB}$ such that $\rho_{AB}$ can be written as
      \[ \rho_{AB} = \sum_k p_k \: \rho_{A,k} \otimes \rho_{B,k} \]
      where, for all $k$, $\rho_{A,k}$ is a density matrix on $\cal{H}_A$, $\rho_{B,k}$ is a density matrix on $\cal{H}_B$, and $p_k$ is a probability.
      Given a Hermitian matrix $M$ on $\cal{H}_{AB}$, $h_{\mathrm{Sep}}(\eps, M, \cal{H}_{AB})$ is the problem of estimating
      \[ \max_{\rho \in \cal{S}_{AB}} \tr(M\rho) \]
      up to additive error $\eps$.
    \end{definition}
    \begin{theorem}
    \label{thm:hsep-lower-bound}
     Let $|\cal{H}|$ denote the dimension of a Hilbert space $\cal{H}$. Suppose $\mathscr{A}$ is an algorithm to solve $h_{\mathrm{Sep}}(\eps, M, \cal{H}_{AB})$ for constant $\eps$ and $M$ such that $\{M, \Id - M\}$ is a Bell measurement (see \Cref{def:bell-meas}). If $\mathscr{A}$ has time complexity at most
      \[ \exp(O(\log^{1 - \nu}|\cal{H}_A| \log^{1-\mu} |\cal{H}_B|)) \]
      for $\nu + \mu > 0$, then 3SAT with $N$ clauses has an algorithm taking time $2^{N^{1-(\nu+\mu)/2} \cdot \poly\log N}$.
    \end{theorem}
    \begin{proof}
      Suppose $\cal{G} = (V,E)$ is a graph with $n$ vertices and $m$ edges, and let $(\cal{G}, K, R)$ be an instance of \textsf{$\delta$-GAP-$K$COL} (\Cref{def:delta-gap-kcol}). We show a $\sf{BellQMA}(2)$ protocol in Section \ref{sec:protocol} to decide any such instance of \textsf{$\delta$-GAP-$K$COL}, which has a constant completeness-soundness gap if $\delta$ and $K$ are both constants, and where the witness is on two unentangled registers of size $O(\sqrt{n} \log m)$ and $O(\sqrt{n} \log n)$. By Theorem \ref{thm:3sat-to-kcol}, we can reduce any $N$-clause instance of 3SAT to an instance of \textsf{$\delta$-GAP-$K$COL} where $\delta$ and $K$ are constants and $m, n = N \poly \log N$. As such, we can set $|\cal{H}_A| = |\cal{H}_B| = 2^{\sqrt{N} \cdot \poly\log N}$. Applying the hypothetical algorithm $\mathscr{A}$ to the measurement $\{M, \Id - M\}$ induced by our protocol, we get that $\mathscr{A}$ can solve 3SAT in time
      \[ \exp\big(O(\log^{1 - \nu}|\cal{H}_A| \log^{1-\mu} |\cal{H}_B|)) \leq \exp(N^{1 - (\nu + \mu)/2} \cdot \poly\log N). \]
    \end{proof}

Assuming the Exponential Time Hypothesis (that $N$-clause 3SAT does not have any algorithm taking time $2^{o(N)}$), \Cref{thm:hsep-lower-bound} shows that there does not exist an algorithm for $h_{\mathrm{Sep}}$ on Bell measurements taking time at most
\[ \exp(O(\log^{1 - \nu}|\cal{H}_A| \log^{1-\mu} |\cal{H}_B|)) \]
for any constant $\nu + \mu > 0$. Therefore, \Cref{thm:hsep-lower-bound} shows that the algorithm given by \cite{BCY} for $h_{\mathrm{Sep}}$ on LOCC measurements, a superclass of Bell measurements, is optimal (possibly up to factors doubly logarithmic in $|\cal{H}_A|$ and $|\cal{H}_B|$).

    \section{Free entangled games: the gapless case}
    \label{sec:gapless}
    In this section, we show that any zero-gap $\MIP^*$ protocol can
    be compressed into a protocol with constant-length questions,
    without greatly increasing its answer length and verifier runtime. As
    a consequence, we obtain that $\MIP^*$ and $\AM^*$ coincide in the
    gapless case.

    It is important to interpret this result carefully: in the gapless
    case, deciding whether a free game has value $1$ is
    trivially as hard as deciding whether a general game has value $1$,
    since any game can be made free by having the verifier sample questions for the two provers independently and
    automatically accept on ``invalid'' question pairs. This
    transformation preserves value 1 strategies, and if the value of
    the original game was bounded away from $1$, so will be the value
    of the transformed game. Thus, the conclusion $\MIP^*_0 = \AM^*_0(2)$
    on its own is \emph{not} interesting, and in fact (using our definitions of $\AM^*_0(2)$ and $\MIP^*_0$ from \Cref{sec:prelims})
    it was \emph{already} shown by \cite{MNY20,MNY21} that $\AM^*_0(2)
    = \MIP^*_0 = \Pi_2$. Rather, what \emph{is} interesting about the
    results in this section is that we can obtain protocols with
    constant question size. In the subsequent section, we will use similar techniques to obtain protocols with constant question size in the
    gapped case, enabling us to show the nontrivial and novel fact
    that $\MIP^* = \AM^*(2) = \RE$.

    The main complexity result of the section is \Cref{thm:mip0gap-main-0}.
    \begin{definition}\label{def:ptl}
      The language $\pitwolang$ consists of all strings $x$ that are descriptions
      of Turing machines such that $\forall y \in \{0,1\}^*$, there exists
      a time $t \in \mathbb{N}$ such that Turing machine $M_x$ described
      by $x$ halts on input $y$ in $t$ steps. 
    \end{definition}
    \begin{fact}
      The language $\pitwolang$ is complete for $\Pi_2^0$. 
    \end{fact}

    \begin{theorem}[\Cref{thm:mip0gap-main} below]
        There is an $\AM^*_0(2)[O(1), O(\log^*(\log(n)) \cdot \log(n)), O(n
  \log^*(n))]$ protocol for the $\Pi_2$-complete language $\ptl$ (\Cref{def:ptl}). 
  \label{thm:mip0gap-main-0}
    \end{theorem}

        \Cref{thm:mip0gap-main-0} is shown by proving a gapless ``hypercompression''
    theorem, which compresses the question length of any game to a
    constant. Another consequence of this theorem is that we
    obtain a (nonrobust) self-test for $n$ EPR pairs with a constant number of
    questions. This is the first such test known; previous results
    have all required at least logarithmic-sized questions, even in
    the zero-gap (i.e. nonrobust) case \cite[Table 1]{BS20}.

    \begin{theorem}[\Cref{thm:epr-gapless} below]
        There is a family of games $G^{EPR}$ with $q(n) = Q_0$, $a(n) =
  O(n)$, $t_c(n) = \poly(n)$ and $t_d(n = \poly(n)$, and
  $\cal{E}(G^{EPR}_n) = n$. 
    \end{theorem}

    \subsection{Granular compression theorems for games}
    \label{sec:gapless-compression}
Some previous works that prove compression theorems state them for
families of games (see \Cref{def:game-family-gapless} for our definition of a family of games) while adhering to a convention
where the question and answer length and the verifier runtime stay roughly the same in $n$ (the index associated with the game family) after compression, and the effect of the compression is expressed by relating the value of $G^\hp_n$ to the value of $G^\hf_{2^n}$, where $G^\hp$ is the `compressed' game family and $G^\hf$ is the uncompressed game family. For instance, the \emph{question reduction theorem} of
\cite{JNVWY20} informally says the following: given a verifier $V^\hf$ who generates a game family $G^\hf$ with questions of length $q^\hf(n)$ and answers of length $a^\hf(n)$, there is `compressed verifier' $V^\hp$ who generates a game family $G^\hp$ with questions of length $\poly(q^\hf(n))$ and answers of length $2^{a^\hf(n)} + 2^{q^\hf(n)} + \poly(n)$, such that the value of the $n$th game in $G^\hp$ is similar to the value of the $2^n$th game in $G^\hf$. We informally refer to this convention as the `scaled-up setting'.

In this paper, we would prefer to state the
\cite{JNVWY20} compression theorem as the following: given a verifier
$V^\hf$ that generates a game $G^\hf$ with question length $q^\hf(n)$ and
answer length $a^\hf(n)$, there is a `compressed
verifier' $V^\hp$ that generates a game $G^\hp$ of question length $\poly\log
q^\hf(n)$ and answer length $a^\hf(n) + q^\hf(n) + \poly\log(n)$, such that the value of the $n$th game in $G^\hp$ is similar to the value of the $n$th game in $G^\hf$. We refer to this as the `scaled-down setting'.
%
%In addition, we will aim to give bounds on
%the question and answer length that are as tight as possible in terms
%of $n$.

The reason we prefer the `scaled-down' convention is that,
ultimately, we would like to repeatedly apply compression to a single protocol
(deciding a particular language $L$) while holding the index $n$
constant, to obtain a protocol with shorter questions that decides the
same language. The scaled-down setting makes it easier notationally for us to do so.

\subsubsection{Gapless question reduction}
We will start by restating the gapless compression theorem
of~\cite{MNY21} in the scaled-down setting, with more precise bounds
on the question and answer length (which we will need in order to prove our hypercompression theorem).

\begin{definition}
The Question Sampling game $QS_n$ is a game with $|\mathcal{X}_{MS}| \cdot n^2$ possible questions and $|\mathcal{A}_{MS}|^2 +2^n$ possible answers. Thus, it has question and answer bit length $q_{QS}(n) = \lceil 2 \log n + \log |\mathcal{X}_{MS}| \rceil$ and $a_{QS}(n) = \lceil 2 \log|\mathcal{A}_{MS}| \rceil + n$. It has checker runtime $t_c^{qs}(n) = O(\log n)$ and decider runtime $t_d^{qs}(n) = O(n)$. 
\end{definition}
The precise definition of the game is given in \cite[Section 3.3]{MNY21}.

\begin{theorem}\label{thm:intro-gapless}
There is a Turing machine $\mathrm{GaplessIntro}$ with the following
properties. Let $G = (D,C,Q)$ be a family of games with question length $q(n)$,
answer length $a(n)$, decider runtime $t_d(n)$, and checker runtime
$t_c(n)$, where $t_c(n), t_d(n) = \poly(n)$, and suppose 
$Q$ on input $n$ computes $q(n)$ in time $O(\log n)$. Then
$\mathrm{GaplessIntro}$ given a description of $D, C, Q$ outputs a
family of games $D^{intro}, C^{intro}, Q^{intro}$ such that the
following hold. 
\begin{enumerate}
    \item (\textbf{Question length:}) The question alphabet size of
      $G^{intro}_n$ is $7 + |\mathcal{X}_{QS}(q(n))| =
      |\mathcal{X}_{MS}| \cdot (q(n))^2 + 7$. The question length of
      $G^{intro}_n$ is thus upper-bounded by $q_{intro}(n) =\lceil 2
      \log q(n) +  7 + \log |\mathcal{X}_{MS}| \rceil $.
    \item (\textbf{Answer length:}) The answer length of $G^{intro}_n$
      is $a(n) + q(n)+ O(1)$. (Here, $O(1)$ means a universal
      constant independent of $G$.) \label{item:intro-answer}
    \item (\textbf{Checker:}) The checker $C^{intro}$ on input
      $(n,x,y)$ computes $q_{intro}(n)$ in time $O(\log n)$, and then runs a universal
      checker $C^{intro, universal}$ on input $(q, x,y)$, which takes time $t_{c}^{qs}(q) + O(1)$. The total checker runtime is thus $t_{c}^{qs}(q_{intro}(n)) + O(\log n)$, where $t_{c}^{qs}$ is
    the checker time bound for the Question Sampling game (and is a function only of
    $q_{intro}(n)$).
    \item (\textbf{Decider:}) The runtime of the decider $D^{intro}$
      is bounded by $t_{d}^{intro}(n) = t_d(n) + t_{d}^{qs}(q_{intro}(n)) +
      O(\log n)$, where $t_{d}^{qs}$ is
    the decider time bound for the Question Sampling game (and is a function only of
    $q_{intro}(n)$). \label{item:intro-runtime}
    \item (\textbf{Completeness:}) For all oracularizable
      finite-dimensional synchronous
      strategies for $G_n$, there exists an oracularizable
      finite-dimensional synchronous
      strategy for $G^{intro}_n$ achieving at least as high a
      value. \label{item:intro-completeness}
    \item (\textbf{Soundness:}) If $\omega_q(G_n) < 1$, then
      $\omega_q(G^{intro}_n) < 1$. \label{item:intro-soundness}
    \item (\textbf{Entanglement bound:}) $\cal{E}(G^{intro}_n) \geq
      \max\{\cal{E}(G_n), 2^{2n}\}$. \label{item:intro-entanglement}
    \item (\textbf{Efficient computability:}) 
      $\mathrm{GaplessIntro}$ should run in time
      $O(|D| + |C|  + |Q|)$. Moreover, $|D^{intro}| = |D| + O(1)$, and
      $|C^{intro}| = |C| + O(1)$, and $|Q^{intro}| = |Q| + O(1)$. \label{item:intro-efficient-computability}
  \end{enumerate}
  \end{theorem}
  \begin{proof}
    The Turing machine $\mathrm{GaplessIntro}$ constructs the
    introspected game of \cite[Theorem 4.1]{MNY21}.
    The claimed bounds on the question and answer length follow by
    inspecting the game in~\cite{MNY21}. The completeness, soundness, and
    entanglement bound are shown directly in~\cite{MNY21}.

    For the checker runtime, by the hypothesis, it takes time $O(\log
    n)$ to compute $q(n)$, and thus time $O(\log n)$ to compute
    $q_{intro}(n) = \lceil 2 \log q(n) +  7 + \log |\mathcal{X}_{MS}| \rceil$. Once
    $q_{intro}(n)$ has been computed, the ``universal''
    \anote{explain} checker in~\cite{MNY21}
    either calls the checker for the question sampling game, or
    performs a constant-time check on the special question
    types. Thus, the total checker runtime is
    $t_{c}^{qs}(q_{intro}(n)) + O(\log n)$ as claimed.

    For the decider runtime, the decider of \cite{MNY21} runs the
    original decider of $D$, the decider for the question sampling
    game, and possibly some constant-time additional operations. In
    order to run the question sampling game decider, it must compute
    $q_{intro}(n)$ which takes time $O(\log n)$. This
    yields a total runtime of $t_d(n) + t_{d}^{qs}(q_{intro}(n)) +
    O(\log n)$.

    For
    efficient computability, this follows by inspecting the
    description of the game.

    We remark that in \cite{MNY21}, in the corresponding theorem, the
    conclusions are stated to hold only for all $n \geq n_0^{intro}$,
    where $n_0^{intro}$ can depend on properties of the input
    game. However, by inspecting the proof, it can be verified that
    this restriction on $n$ is only needed due to the non-asymptotic
    form of the bounds stated by \cite{MNY21}. In fact, the
    completeness and soundness hold for all $n$, and bounds on
    question and answer length and on checker and decider runtime as
    we have stated them here also hold for all $n$.
  \end{proof}

  \paragraph{Remark} The checker $C^{intro}$ is in fact
  \emph{universal}, and does not depend at all on $G$.

\subsubsection{Gapless answer reduction}
The following theorem is essentially \cite[Theorem 5.1]{MNY21}, but
with tighter bounds on the question length, answer length, and decider runtime. 
\begin{theorem} \label{thm:ar-gapless}
There is a Turing machine $\mathrm{GaplessAnsReduce}$ with the following
properties. Let $G = (D,C,Q)$ be a game family with complexity bounds $q(n), a(n),
t_c(n), t_d(n)$ respectively, where $a(n) \leq t_d(n) = \poly(n)$ and
$q(n) = O(\log n)$. Then $\mathrm{GaplessAnsReduce}$ given as input a description of $(D,C,Q)$ returns a description of a game family $G^{ans} = (D^{ans}, C^{ans},Q^{ans})$ such that for the game family $G^{ans}$ and for all $n$, we have
\begin{enumerate}
    \item (\textbf{Question length:}) The question length is
      $q_{ans}(n) = 2q(n)
      + O(\log |D| + \log t_d(n))$, where $|D|$ is the description length of the
      Turing machine $D$. 
    \item (\textbf{Answer length:}) The answer length is $O(1)$. 
    \item (\textbf{Checker runtime:}) The checker runs in time $t_c(n) + O(\log |D| + \log t_d(n))$.
    \item (\textbf{Decider runtime:}) The decider runtime is
      $\poly\log(n)$.
    \item (\textbf{Completeness:}) For any oracularizable
      finite-dimensional synchronous
      strategy $\scr{S}$ to $G_n$, there is an oracularizable
      finite-dimensional synchronous strategy $\scr{S}^{ans}$ to $G^{ans}_n$ such that
      \[ \omega_q(G^{ans}_n, \scr{S}^{ans}) \geq \frac{1}{2} +
        \frac{1}{2} \omega_q(G_n, \scr{S}). \]
    \item (\textbf{Soundness:}) If $\omega^s_q(G_n) < 1$, then
      $\omega^s_q(G^{ans}_n) < 1$.
    \item (\textbf{Entanglement bound:}) $\cal{E}(G^{ans}_n) \geq \cal{E}(G_n)$.
    \item (\textbf{Efficient computability:}) 
      $\mathrm{GaplessAnsReduce}$ should run in time
      $O(|D| + |C|  + |Q|)$. Moreover, $|D^{ans}| = |D| + O(1)$, and
      $|C^{ans}| = |C| + O(1)$, and $|Q^{ans}| = |Q| + O(1)$.
\end{enumerate}
\end{theorem}
\begin{proof}
  The properties claimed follow from the proof of \cite[Theorem
  5.1]{MNY21}. As in the case of question reduction, we remark that in \cite{MNY21}, in the corresponding theorem, the
    conclusions are stated to hold only for all $n \geq n_0^{intro}$,
    where $n_0^{intro}$ can depend on properties of the input
    game. However, this is only because of the non-asymptotic form of the bounds stated there; completeness and soundness hold for all $n$, and the asymptotic bounds we have stated also hold for all $n$.
  
  To see how the items in the conclusion of the theorem follow from \cite{MNY21}, recall how the answer-reduced game is
  obtained:
  \begin{enumerate}
    \item First the game $G$ is \emph{oracularized}: this increases
      the question and answer length by a factor of $2$, and increases
      the decider runtime by $O(a_n)$. 
    \item Next, the game $G^{ans}$ is obtained by instructing the
      prover in the oracularized game to construct a \emph{tableau} of the verifier's
      computation on the answers, as in the Cook-Levin theorem. The verifier's
      computation accepts the prover's answers if and only if the
      tableau satisfies a set of local constraints, each of which acts
      on a constant number of locations in the tableau. The verifier's questions in $G^{ans}$ now consist
      of a question from the oracularization of $G$, together with a
      constant number of indices into the tableau. The honest prover's answers
      consists of values written in the tableau at these indices. The
      verifier checks that these answers satisfy the relevant local constraint.
  \end{enumerate}

  The length of the tableau is denoted $L$ in their proof, and it is
  shown that $L = \poly(|D|, q(n), t_n) = \poly(|D|, t_d(n))$.
  \begin{enumerate}
    \item (\textbf{Question length:})  The
      question set $G^{ans}$ consists of the Cartesian product of the question set of
      $G$ with all tuples of at most 3 indices into the tableau. Thus, the
      question length of $G^{ans}$ is at most $2q(n) + O(\log L) = q(n)
      + O(\log \poly(|D|, t_d(n))) = 2q(n) + O(\log t_d(n)  + \log |D|)$.
    \item (\textbf{Answer length:}) This is stated in \cite{MNY21}; in
      fact, the answers are at most $3$ bits long.
    \item (\textbf{Checker runtime:}) A description of the valid question pairs is given in \cite[Table 5]{MNY21}. From this we see that the checker needs to run the checker $C$ from the original game, as well as an additional check that runs in time linear in the number of bits in the indices into the provers' tableau. This gives the claimed runtime.
    \item (\textbf{Decider runtime:}) In \cite{MNY21} it is stated
      that the decider runtime is
      $\log^\gamma n$ for a constant $\gamma$.
    \item (\textbf{Completness, soundness, and entanglement bound:}) these are identical to
      the statements in \cite{MNY21}.
    \item (\textbf{Efficient computability:}) The description of the answer-reduced Turing machine $D$ consists of the code for $D$ together with the code for the answer-reduction transformation (essentially for the Cook-Levin reduction), which is of size $O(1)$. For the checker $C^{ans}$, from the description of valid question pairs on \cite[Table 5]{MNY21}, we see that $C^{ans}$ invokes $C$ on part of its input and then performs a fixed additional check, so the description length of the Turing machine is $|C| + O(1)$. For $Q^{ans}$, this follows because the question length of the answer-reduced verifier is efficiently computable from the question-length of the original verifier.
  \end{enumerate}
\end{proof}
\subsubsection{Gapless compression}
By composing question and answer reduction, one obtains a gapless
compression theorem~\cite[Theorem 6.1]{MNY21}. We state the bounds
that we obtain below.
\begin{theorem}\label{thm:gapless-compress}
There exists a Turing machine $\mathrm{GaplessCompress}$ with the following properties.  Let $G = (D,C,Q)$ be a game family with complexity bounds $q(n), a(n),
t_c(n), t_d(n)$ respectively, such that $t_c(n), t_d(n) = \poly(n)$. Then on input $(G,D,Q)$, $\mathrm{GaplessCompress}$ returns a description of a game family $G' = (D', C',Q')$
with parameters $q'(n), a'(n), t_c'(n), t_d'(n)$ such that
\begin{enumerate}
    \item (\textbf{Question length:}) The question length of $G'_n$ is
      $q'(n) = O(\log q(n) + \log t_d(n) + \log |D|)$.
    \item (\textbf{Answer length:}) The answer length of $G'_n$ is
      $a'(n) = O(1)$.
    \item (\textbf{Checker runtime:}) The checker runtime is $t_c'(n) = O(\log n + \log t_d(n) + \log |D|)$.
    \item (\textbf{Decider runtime:}) The decider runtime is $t_d'(n)
      = \poly \log(n)$.
    \item (\textbf{Completeness:}) For any oracularizable synchronous
      finite-dimensional strategy $\scr{S}$ for $G_n$, there exists an
      oracularizable synchronous finite-dimensional strategy
      $\scr{S}'$ for $G'_n$ such that $\omega_q(G'_n, \scr{S}') \geq
      \frac{1}{2} + \frac{1}{2} \omega_q(G_n, \scr{S})$.
    \item (\textbf{Soundness:}) If $\omega^s_q(G_n) < 1$, then
      $\omega^s_q(G'_n) < 1$.
    \item (\textbf{Entanglement bound:}) $\cal{E}(G'_n) \geq
      \max\{\cal{E}(G_n), 2^{2n}\}$
    \item (\textbf{Efficient computability:}) 
      $\mathrm{GaplessCompress}$ should run in time
      $O(|D| + |C|  + |Q|)$. Moreover, $|D'| = |D| + O(1)$, and
      $|C'| = |C| + O(1)$, and $|Q'| = |Q| + O(1)$.
\end{enumerate}

\end{theorem}
\begin{proof}
  Compose \Cref{thm:intro-gapless} and \Cref{thm:ar-gapless}.
\end{proof}

\subsection{Gapless hypercompression}

\begin{theorem}
  \label{thm:gapless-hyper}
  Let $G = (GenG^{\halfnote}, N^{\halfnote})$ be an $n$-indexed $\AM^*_0(2)[ q^\hf(\cdot), a^\hf(\cdot),
  t_c^\hf(\cdot), t_d^\hf(\cdot)]$
  protocol.
  There exists a Turing machine $\mathrm{GaplessHypercompress}$ independent of
  $G$ that, given as input
  the description of a pair of Turing machines
  $(GenG^{\halfnote},N^{\halfnote})$, outputs a
  description of a pair of Turing machines $G^{\hp} = (GenG^\hp, N^\hp)$ with the
  following properties: $G^{\hp}$ is an $n$-indexed $\AM^*_0(2)[Q_0, a^{\hp}(\cdot),
  t_c^{\hp}(\cdot), t_d^{\hp}(\cdot)]$ protocol with $N^\hp = N^{\halfnote}$ and
  \begin{enumerate}
  \item \textbf{Question length:} $Q_0$ a universal constant such that, for all $\ell > Q_0$, $\lceil 2\log \ell + 7 + + \log |\mathcal{X}_{MS}| \rceil < \ell$.
  \item \textbf{Answer length:} For any $n$, $a^{\hp}(n) = a^\hf(n) + \log^*(q^\hf(n))
    (q^\hf(n) + O(1))$.
  \item \textbf{Checker runtime:} For any $n$, $t_c^{\hp} = O(1)$.
  \item \textbf{Decider runtime:} For any $n$, $t^{\hp}(n) = t_d^\hf(n) + (|D^\hf| + |C^\hf| + |Q^\hf|)\poly q^\hf(n)$.
  \item \textbf{Completeness:} Let $z$ be a problem instance, let
    $G^{\halfnote} = (D^{\halfnote},C^{\halfnote}, Q^{\hf})$ be the
    game family output by $GenG^{\halfnote}(z)$, and let $G^\hp =
    (D^\hp, C^\hp, Q^{\hp})$ be the game family output by $GenG^\hp(z)$. Let $G^{\halfnote}_n$ be the $n$th family in $G^{\halfnote}$, and let $G^\hp$ be the $n$th game in $G^\hp$. For all $n$, if $G^{\halfnote}_n$ has a value-$1$
    oracularizable finite-dimensional synchronous strategy $\scr{S}^{\halfnote}$, then
    $G^{\hp}_n$ has a value-$1$ oracularizable finite-dimensional synchronous strategy
    $\scr{S}^{\hp}$. 
  \item \textbf{ Soundness:} For all $n$, if $\omega^s_q(G^{\halfnote}_n)
    < 1$, then $\omega^s_q(G^{\hp}_n) < 1$ as well. 
  \item \textbf{Entanglement bound:} For all $n$, $\cal{E}(G^{\hp}_n)
    \geq \cal{E}(G^{\halfnote}_n)$.
  \end{enumerate}
\end{theorem}
\begin{proof}
  On any input $z$, $GenG^\hp$ runs $GenG^{\halfnote}$ on input $z$ in order to generate a decider $D^{\halfnote}$ and a checker $C^{\halfnote}$. $GenG^\hp$ then outputs the following descriptions of a decider $D^\hp$ and a checker $C^\hp$.
  \begin{enumerate}
    \item \textbf{Checker:}
    The checker $C^\hp$, on input $(n,x,y)$, ignores $n$ and runs the checker $C^{intro}$ from \Cref{thm:intro-gapless} on input $(n^*,x,y)$ for the value of $n^*$ such that $q_{intro}(n^*) = Q_0$.
  	 \item \textbf{Decider:}
    The decider executes the algorithm in \Cref{alg:hypercomp-pseudo}.
  \end{enumerate}

  \vspace{10pt}
  \IncMargin{1em}
  \begin{algorithm}[H]
    \DontPrintSemicolon

    \textbf{Input}: $n, x, y, a, b$
    
    Initialize $i := 0$, $\ell_0 := q^\hf(n)$, $D_{comp} = D^{\halfnote}, C_{comp} = C^{\halfnote}, Q_{comp} = Q^{\halfnote}$.

    \While{$\ell_i > Q_0$}{
      $(D_{comp}, C_{comp}, Q_{comp}) := \mathrm{GaplessIntro}(D_{comp}, C_{comp}, Q_{comp})$
      
      \Comment{GaplessIntro is defined in \Cref{thm:intro-gapless}}

      $\ell_{i+1} := \lceil 2 \log \ell_i + 7 + \log |\mathcal{X}_{MS}|\rceil$
      
      $i := i+1$
    }
    Run $D_{comp}(n,x,y,a,b)$. Accept if it accepts, otherwise reject.
    \caption{The hypercompressed game.}\label{alg:hypercomp-pseudo}
  \end{algorithm}\DecMargin{1em}
  \vspace{10pt}

  \paragraph{Checker and decider runtime}
  It is clear that the checker $C^{\hp}$ runs in constant time.
  
  To bound the decider's runtime, we need to compute two things: (1)
  the runtime of the loop calling $\mathrm{GaplessIntro}$, (2) and the
  runtime of the final decider $D_{comp}$ that is evaluated on
  $n,x,y,a,b$.

  For the first, notice that each call of $\mathrm{GaplessIntro}$ shrinks
  $\ell_{i+1}$ relative to $\ell_i$ by an exponential factor: we may take $Q_0$ to be
  sufficiently large so that for all $\ell_i > Q_0$, $\ell_{i+1} = \lceil 2 \log \ell +  7 + \log |\mathcal{X}_{MS}|
  \rceil \leq 3 \log \ell_i$. Let us denote the number of iterations of
  the loop by $k$. The preceding argument shows that 
  $k = O(\log^*(q^\hf(n)))$, where $\log^*$ is the iterated
  logarithm.

  Now, during each iteration, the call to $\mathrm{GaplessIntro}$ takes
  time linear in the description length of $D_{comp}$,
  $C_{comp}$ and $Q_{comp}$. By \Cref{item:intro-efficient-computability} of \Cref{thm:intro-gapless}, each
  application of $\mathrm{GaplessIntro}$ increases the description lengths of
  $D^\hf$, $C^\hf$, and $Q^\hf$ by a constant additive factor, and hence at
  the $i$th iteration, $|D_{comp}^i| = |D^\hf| + ci$, $|C_{comp}^i|
  = |C^\hf| + ci$, $|Q_{comp}^i|
  = |Q^\hf| + ci$ for some constant $c$. Hence, the total amount of time spent by the calls to
  $\mathrm{GaplessIntro}$ in the loop scales as
  \[ \sum_{i=1}^{k} (|D^\hf| + |C^\hf| + |Q^\hf| + 3ci) =
    O\left( (\log^*(q(n)))^2  + \log^*(q(n)) \cdot (|D^\hf| +
      |C^\hf| + |Q^\hf|)\right). \]

  Now we calculate the runtime of the final decider. By \Cref{item:intro-runtime} of \Cref{thm:intro-gapless}, the runtime of the decider increases by $t_d^{qs}(\ell_i) + O(\log n)$ when we apply GaplessIntro for the $i$th time. (The bound in \Cref{item:intro-runtime} of \Cref{thm:intro-gapless} says that applying compression once increases the runtime of the decider by $t_d^{qs}(q_{intro}) + O(\log n)$; in \Cref{alg:hypercomp-pseudo-gapped} we have set each $\ell_i$ to be an upper bound on the question length of the game output by GaplessIntro applied for the $i$th time.) Therefore, the runtime of the final decider is bounded by
  \begin{align*}
    t_d^\hp(n) = &t_d^\hf(n) + \sum_{i=1}^{k} (t_{d}^{qs}(\ell_i(n)) + O(\log n)) \\
    \leq &t_d^\hf(n) + \sum_{i=1}^{k} (t_{d}^{qs}(\ell_0(n)) + O(\log n)) \\
      \leq &t_d^\hf(n) + \log^*(q^\hf(n)) \cdot \poly(q^\hf(n)) \\
      = &t_d^\hf(n) + \poly(q^\hf(n)).
  \end{align*}

  \paragraph{Question length}
  When the loop terminates, we are guaranteed that $\ell \leq Q_0$ and
  moreover that $\ell$ is the question length of the game associated to
  $D_{comp}$.  Thus, the question length is at most $Q_0$.
  \paragraph{Answer length}
  The loop runs for $k= O(\log^*(q^\hf(n))$ iterations. Thus, by
  \Cref{item:intro-answer} of \Cref{thm:intro-gapless}, the final decider has answer length
  \[ a^{\hp}(n) \leq a^\hf(n) + \sum_{i=1}^{k} (\ell_i(n) + O(1)) \leq a^\hf(n)
    + \log^*(q^\hf(n))(q^\hf(n) + O(1)). \]
  \paragraph{Completeness}
    This follows by repeatedly applying
    \Cref{item:intro-completeness} of \Cref{thm:intro-gapless}.
  \paragraph{Soundness} This follows by repeatedly applying
  \Cref{item:intro-soundness} of \Cref{thm:intro-gapless}.
  \paragraph{Entanglement bound}
  This follows by repeatedly applying
  \Cref{item:intro-entanglement} of \Cref{thm:intro-gapless} and
  (loosely) lower-bounding the max by $\cal{E}(G^\hf_n)$.
\end{proof}

\begin{theorem}\label{thm:mny21-main}
    There is an $n$-indexed $\AM^*_0(2)[q = O(\log(n)), a = O(1), t_c = \poly\log(n), t_d = \poly\log(n)]$
    protocol for $\pitwolang$.
  \end{theorem}
  \begin{proof}
This theorem is obtained by applying one round of gapless compression
to the main result of \cite{MNY21}. That
work showed explicitly that the problem of computing the exact
$q$-value of a nonlocal game is $\Pi_2$-complete. We trace through the
reduction to obtain an $\AM^*_0(2)$ protocol with the claimed bounds (many of which are not written in the theorem statements in \cite{MNY21}).

By the definition of an $n$-indexed $\AM^*_0(2)$ protocol (\Cref{def:nindexed-amstar}), this means that we must find a pair of Turing
machines $GenG^{\pitwolang}, N^{\pitwolang}$ that satisfy the
following properties (stated somewhat informally):
\begin{itemize}
\item $GenG^{\ptl}$ takes input $z$ and in time $\poly(|z|)$ outputs
  descriptions of Turing machines $D^z$, $C^z$, and $Q^z$ specifying a
  nonlocal game family.
\item The game $G = (D^z_{N(|z|)} ,C^S_{N(|z|)})$ has value
  $\omega_q(G) = 1$ iff $z \in \ptl$, and has question, answer, and
  runtime complexity scaling as $Q^z(N(|z|)) = O(\log(|z|)), O(1)$, and $\poly(|z|)$ respectively.
\end{itemize}

In the remainder of the proof, we will explain how to construct
$GenG^{\ptl}, N^{\ptl}$ and why the properties hold.

We will start by obtaining an $n$-indexed $\AM^*_0(2)$ protocol for $\ptl$ with
\emph{polynomial} question and answer length, by tracing through
\cite[Lemma 6.7]{MNY21}.
Given an instance $z$, define the computable predicate $\phi(y, t)$ to be
true iff the Turing machine $M_x$ on input $y$ halts in $t$ timesteps. Then $z \in
\pitwolang$ iff the following $\Pi_2^0$ sentence $S$ is true:
\[ S = \forall y, \exists t \phi(y,t). \]
Following the proof in \cite{MNY21}, for every $m$ define the
$\Sigma_1$ sentence (where we identify the natural numbers $1, \dots,
m$ with bit strings in binary to interpret them as inputs to $\phi$)
\[ S_m = \exists t_1, \dots, t_m \bigwedge_{i=1}^{m} \phi(i, t_i), \]
For any given $m$, the game $\mathrm{HaltingGame}(S_m) = (q_m = (m
|S|)^{\lambda_{Halt}}, D_m, C_m)$ is the game from
\cite{JNVWY20} to decide the $\Sigma_1 = \RE$ sentence $S_n$, and has runtime, question length, and
answer length bounded by
$(m|S|)^{\lambda_{Halt}}$, where $\lambda_{Halt}$ is a universal
constant. Moreover, from \cite[Theorem 12.6]{JNVWY20}, the family of games
$\{\mathrm{HaltingGame}(S_m)\}_m$ is actually generated (in the sense
of \Cref{def:game-family-gapless}) by a triple of
Turing machines $(D, C, Q)$ with $m$ as the index, and that run in
polynomial time and have description length bounded by
$\poly(|z|)$. (The Turing machine $Q$ is simply the Turing machine
that on input $m$ outputs $(m|S|)^{\lambda_{Halt}}$.

Now, applying SuperCompress \cite[Theorem
6.2]{MNY21} to the family $(D,C, Q)$, there exists a family of games
$G^{super} = (D^{super}, C^{super}, Q^{super})$, a $\lambda =
O(\lambda_{Halt})$ and a $\kappa = \poly(|D|, |C|, \lambda_{Halt} + 1, m_0,
\lambda^{\poly(\lambda)}) = \poly(|z|, \lambda_{Halt})$ such that
the game $G^{super}_{\kappa}$ has value $1$ iff 
$\mathrm{HaltingGame}(S_m)$ has value $1$ for all $m \geq
\kappa$. In
other words, the sentence $S$ is true iff the game
$G^{super}_{\kappa}$ has value $1$. 
Descriptions of $D^{super}, C^{super}, Q^{super}$ are computable in polynomial
time from descriptions of $D, C, Q$ (and thus in time $\poly(|z|)$) by \cite[Theorem
6.2]{MNY21}. Moreover, $\kappa$ is computable in time $\poly(|z|)$, by the explicit
expression for $\kappa$ given in the proof of \cite[Claim
6.3]{MNY21}. Therefore, given $z$, we can compute descriptions of
$D^{super}_{\kappa}, C^{super}_{\kappa}, Q^{super}_{\kappa}$ in time
$\poly(|S|)$. Furthermore by Item~1 of \cite[Theorem 6.2]{MNY21}, it
holds that $D^{super}_{\kappa}, C^{super}_{\kappa}, Q^{super}_{\kappa}$ run in time
at most $\kappa^{\lambda}  = \poly(|z|)$. 

At this point, we have constructed an $n$-indexed $\AM^*_0(2)[\poly(n),
\poly(n), \poly(n), \poly(n)]$ protocol for $\ptl$.
To see this, we define Turing machines $GenG^{\fullnote}, N^{\fullnote}$ as follows:
\begin{itemize}
\item $GenG^{\fullnote}$ first parses the input $z$, then computes descriptions of
  $D^{super}, C^{super}$, and $Q^{super}$ in time $\poly(|z|)$.
\item $N^{\fullnote}$ computes $\kappa(|z|)$ using the explicit expression
  mentioned above.
\end{itemize}

However, we want to get an $n$-indexed $\AM^*_0(2)$ protocol for $L$ which has \emph{logarithmic}
question complexity and \emph{constant} answer complexity. To obtain
this, we need to apply a round of gapless compression. The gapless
compression theorem (\Cref{thm:gapless-compress}) requires as input a
sequence of games indexed by $n$; for any input $z$, we take the
output of $GenG^{\fullnote}(z)$ to be this sequence. Applying \Cref{thm:gapless-compress}, we
obtain a new triple of Turing machines $(D^{z, \halfnote},
C^{z,\halfnote}, Q^{z, \halfnote})$
such that
\begin{itemize}
  \item Let $G^{z,\halfnote}_{n}$ be the $n$th game in the game family defined by $G^{z,\halfnote} = (D^{z, \halfnote}, C^{z,\halfnote})$. $G^{z,\halfnote}_{n}$ has question
    complexity $O(\log n + \log |D^{super}|) = O(\log n + \log |z|)$,
    and answer complexity $O(1)$.
  \item The runtime of $D^{z,\halfnote}_n, C^{z,\halfnote}_n,$ and
    $Q^{z, \halfnote}$ is bounded
    by $\poly\log n$.
  \item For all $n \geq n_0$, $G^{z,\halfnote}_n$ simulates
    $G^{super}_n$.
  \item Descriptions of $D^{z,\halfnote}, C^{z,\halfnote},$ and $Q^{z,\halfnote}$ can be
    computed in polynomial time from descriptions of $D^{super}$ and $C^{super}$.
\end{itemize}
Without loss of generality assume that $\kappa$ has been picked so
that $\kappa(|z|) > n_0$. We can now finally define the Turing machines
$GenG^{\halfnote}, N^{\halfnote}$ that constitute an $n$-indexed $\AM^*$ protocol for $\ptl$
with the desired complexity bounds.
\begin{itemize}
  \item $GenG^{\halfnote}$ first parses the input $z$, and computes a
    description of $D^{super}, C^{super}, Q^{super}$ in time
    $\poly(|z|)$. Next, it computes a description of $D^{z,
      \halfnote}, C^{z, \halfnote}, Q^{z, \halfnote}$
    in time $\poly(|D^{super}|) = \poly(|z|)$, and returns this as output.
  \item  $N^{\halfnote}$ computes $\kappa(|z|)$ (it is identical to $N$
    defined above).
  \end{itemize}
  On input $z$, the game executed by this protocol is exactly the game
  $G^{\halfnote}_\kappa$. We know that $\omega_q(G^{\halfnote}_\kappa)
  =1$ iff $\omega_q(G^{super}_\kappa) = 1$, which is true iff $z \in
  \ptl$. Thus, this protocol decides $\ptl$. The claimed question,
  answer, and runtime bounds follow by the discussion above\anote{expand}.
  \end{proof}
\begin{theorem}\label{thm:mip0gap-main}
  There is an $\AM^*_0(2)[q(n) = Q_0, a(n) = O(\log^*(\log(n)) \cdot \log(n)),t_c(n) = O(1), t_d(n) = O(n
  \log^*(n))]$ protocol for the $\Pi_2$-complete language $\ptl$,
  where $Q_0$ is the universal constant from \Cref{thm:gapless-hyper}.

\end{theorem}
\begin{proof}
  By \Cref{thm:mny21-main}, we know that $\ptl$ has an $n$-indexed $\AM^*_0(2)[\log(n),
  O(1), \poly(n)]$ protocol. Applying \Cref{thm:gapless-hyper} yields the
  conclusion.
\end{proof}

\begin{theorem}\label{thm:epr-gapless}
  There is a family of games $G^{EPR}$ with $q(n) = Q_0$, $a(n) =
  O(n)$, $t_c(n) = O(1)$ and $t_d(n) = \poly(n)$, and
  $\cal{E}(G^{EPR}_n) = n$.
\end{theorem}
\begin{proof}
  Apply \Cref{thm:gapless-hyper} to the Question Sampling from
  \cite{MNY21}. 
\end{proof}

\section{Free entangled games: the gapped case}
\label{sec:gapped}

Our treatment of the gapped case will largely follow along the
lines of the gapless case.
%For simplicity of presentation, we only
%state our version of the gapped \emph{question} reduction theorem which we need in order to prove our hypercompression theorem; we do not state or use a gapped answer reduction
%theorem. As a result, our final protocol for $\RE$, which is obtained by applying hypercompression directly to the $\MIP^*$ protocol for $\RE$ from \cite{JNVWY20}, will have
%polynomially large answers. We expect that the answer size of our protocol can be reduced (to $\poly\log n$, where $n$ is the instance size) by using the same strategy which we used in the gapless case, namely, that of applying a layer of question and answer reduction to the starting protocol before
%proceeding to hypercompression.
Throughout this section, we reuse analogous notation from the gapless case for
convenience. For instance, the universal constant question size we
obtain for our gapped hypercompression theorem will also be denoted
$Q_0$ here, even though it is not necessarily the same constant as in
the gapless case.
\subsection{Granular gapped compression theorems}
\subsubsection{Gapped question reduction}

The following definition states the properties of the Pauli Basis test from \cite[7.3]{JNVWY20}.
\begin{definition}
The Pauli Basis test $PB_n$ is a family of games with question length
$q_{PB}(n) = \poly\log(n)$,  answer length $a_{PB}(n) = O(n)$, sampler runtime $t_s^{PB}(n) = \poly\log(n)$, and decider runtime $t_d^{PB}(n) = O(n)$.
\end{definition}

The following theorem combines the introspection and gap amplification
theorems of \cite{JNVWY20}.

\begin{theorem}\label{thm:intro-gap}
  There is a Turing machine $\mathrm{GappedIntro}$ with the following properties.
Let $G = (S,D)$ be a family of gapped games with question length $q(n)$,
answer length $a(n)$, decider runtime $t_d(n)$, and sampler runtime
$t_s(n)$, where $t_s(n), t_d(n) = \poly(n)$, and where the sampler is
an $\ell$-level conditional linear sampler with $\ell \leq 10$. Then
$\mathrm{GappedIntro}$ given a description of $S, D$ outputs a
tuple $(S^{intro}, D^{intro})$ defining a family of games
$G^{intro}$ such that the
following hold.
\begin{enumerate}
    \item (\textbf{Question length:}) The question length is
      upper-bounded by $q^{intro}(n) = \max\left( C \log^{\beta}(q(n)), C \log^{\beta}(q^*) \right),$ where $C,
      \beta, q^*$ are universal constants.
    \item (\textbf{Answer length:}) The answer length of $G^{intro}_n$
      is $O((a(n) + q(n) + O(1)) \cdot \poly\log q(n))$. \label{item:intro-answer-gap}
    \item (\textbf{Sampler:}) The sampler $S^{intro}$ on input
      $n$ computes $q_{intro}(n)$ and then runs a universal
      sampler $S^{intro, universal}$ on input $q$. The latter
      runs in time $\poly\log q(n) \cdot O(t_{s}^{PB}(q))$. Thus the total sampler
      runtime is $t_{s}^{intro}(n) = O(\poly\log q(n) \cdot t_{s}^{PB}(q(n)))$.
      The sampler $S^{intro}$ is an $\ell'$-level conditional linear
      sampler with $\ell' = 5$.
    \item (\textbf{Decider runtime:}) The runtime of the decider $D^{intro}$
      is bounded by $t_{d}^{intro}(n) = (t_d(n) + t_{d}^{PB}(q(n)) +
      O(\log n)) \cdot \poly\log q(n)$. \label{item:intro-runtime-gap}
    \item (\textbf{Completeness:}) For all oracularizable
      finite-dimensional synchronous
      strategies for $G_n$, there exists an oracularizable
      finite-dimensional synchronous
      strategy for $G^{intro}_n$ achieving at least as high a
      value. \label{item:intro-completeness-gap}
    \item (\textbf{Soundness:}) If $\omega_q(G_n) \leq 1/2$, then
      $\omega_q(G^{intro}_n) \leq  1/2$. \label{item:intro-soundness-gap}
    % \item (\textbf{Entanglement bound:}) \anote{Wrong: $\cal{E}(G^{intro}_n, 1) \geq
    %   \max\{\cal{E}(G_n), 2^{2n}\}$.} \label{item:intro-entanglement-gap}
    \item (\textbf{Efficient computability:}) 
      $\mathrm{GappedIntro}$ runs in time
      $O(|D| + |S|)$. Moreover, $|D^{intro}| = |D| + O(1)$ and
      $|S^{intro}| = |S| + O(1)$. \znote{check this} \label{item:intro-efficient-computability-gap}
  \end{enumerate}
\end{theorem}
\begin{proof}
  There are several major differences between this theorem and the
  introspection theorem of \cite{JNVWY20}, which we detail below. For notational clarity, we refer to the introspected game constructed in \cite{JNVWY20} as $G^{BareIntro}$, to distinguish it from the one we construct in this theorem.
  \begin{itemize}
  \item \textbf{Parameters} In \cite{JNVWY20}, the introspection
    theorem was stated in the ``scaled-up'' setting (as described in \Cref{sec:gapless-compression}): the value of the
    game $G^{BareIntro}_n$ was related to the value of $G_N$ for $N =
    2^n$. In contrast, here we will work in the ``scaled-down''
    setting: we would like to relate $G^{intro}_n$ to $G_n$. 
  \item \textbf{$\lambda$-boundedness}
    In \cite{JNVWY20}, the game $G$ to be compressed was assumed to be
     ``$\lambda$-bounded'': the question length, answer length, and
     runtimes for $G_n$ were all assumed to be bounded by $n^{\lambda}$,
     and the complexity of the introspected game was specified in
     terms of $n$ and $\lambda$. The compression procedure required
     $\lambda$ as an input, in order to set the size of the instance
     of the Pauli basis test to be used in order to sample questions. This formulation was convenient there
     to discuss \emph{fixed points} of game compression maps, but is
     not sufficient for us because we would like the question and
     answer length to scale separately. Thus, instead of specifying
     an exponent $\lambda$, we ask the sampler $S$ to efficiently compute $q(n)$ (in time $t_s(n)$). The
     compression procedure will run $S$ to determine the size of the
     instance of the Pauli basis test to run.
    \item \textbf{Parallel repetition:} The introspection game in
      \cite{JNVWY20} is not ``gap-preserving.''  Rather, it has the
      following soundness guarantee: if $\omega_q(G^{BareIntro}_n) \geq 1 -\eps$,
      then $\omega_q(G_n) \geq 1 - \delta(\eps,n)$. We must set $\eps$ to be a vanishing function of $n$ in order to obtain $\delta \leq 1/2$.
      
      We will remedy this by applying parallel repetition on top of the game $G^{BareIntro}$. Parallel repetition is also used in \cite{JNVWY20} to amplify soundness gaps. Unfortunately, directly applying the strong anchored parallel
      repetition theorem of \cite{BVY17} that was used in \cite{JNVWY20} will
      not suffice for us. This is because the bound given by that
      theorem has a dependence on the answer length of the game, and
      the bound becomes trivial when $k$ is less than the answer
      length. In our setting it is essential that $k$ should depend
      only on the question length. Thus, we turn instead to the older
      parallel repetition bound of \cite{DSV15},  which applies to
      projection games. As noted in \cite{DSV15}, we can convert any
      game to a projection game using oracularization, at only a
      constant-factor cost to the soundness gap. (Oracularization will not always preserve \emph{completeness}, but in our case it does because we assume that the input game had a perfect oracularizable strategy in the completeness case.) One the game has been oracularized, the result of \cite{DSV15} tells us that by taking sufficiently many repetitions $k$, if $\omega_q(G^{BareIntro}_n) \leq 1 - \eps(n)$, then $\omega_q((G^{BareIntro, orac}_n)^{\otimes k}) \leq 1/2$ as desired.

    \end{itemize}

    We now describe, in more detail, the construction of the introspected game $G^{intro} = (S^{intro},
    D^{intro})$ and prove the properties claimed in the theorem.

    \paragraph{Bare Introspection:}  We start by performing a version
    of introspection as done in \cite[Section 8]{JNVWY20}, but with
    the following modifications
    \begin{itemize}
    \item \textbf{Sampler:} The sampler sets $R = q(n)$ (rather than
      $R = N^{\lambda}$, which it computes by invoking $S(n, \textsc{Dimension})$.
    \item \textbf{Decider:} The decider sets $N = n$ (rather than
      $2^n$) and does not abort/time out. 
    \end{itemize}
    \paragraph{Oracularization:} Next, we perform a
    modified version of oracularization based on the version in
    \cite[Section 9]{JNVWY20}. The only difference is the type graph:
    we eliminate the edge between the types $\textsc{A}$ and
    $\textsc{B}$. This means the resulting game is a projection game:
    either both players receive the same type, in which case they must
    return equal answers, or one player receives the $\textsc{Oracle}$ type and
    the other receives the $\textsc{A}$ or $\textsc{B}$ type, in which
    case the answer of the former uniquely determines a correct answer
    for the latter.
    \paragraph{Parallel repetition:}  Finally, we perform parallel
    repetition as in \cite[Section 11]{JNVWY20}, \emph{without} the
    anchoring transformation. To compute the number of repetitions needed, we must recall the soundness of the bare introspection procedure. This is given by the following statement: if $\omega_q(G^{BareIntro}_n) \geq 1 - \eps$, then
      \begin{equation}
        \omega_q(G_n) \geq 1 - \delta(\eps, n), \qquad 
      \delta(\eps, n) = a((\log q(n))^a \eps^b + (\log
      q(n))^{-b}),\label{eq:gapped-intro-soundness}\end{equation}
      where $a > 0$ and $b, 0 < b < 1$ are constants depending only on $\ell$.  (For our
      purposes, we have fixed the bound $\ell \leq 10$, so we may take
      $a,b$ to be universal constants.)
      Let $\eps(n)$ be a function such that $\delta(\eps(n),n) < \frac{1}{2}$ for all $n \in \mathbb{N}$ as long as $q(n) > q^* \: \forall n$, where $q^*$ is a universal constant depending only on $b$; note that $\eps(n)$ can be chosen such that $1/\eps(n) \leq \poly\log(q(n))$. (The assumption $q(n) > q^* \: \forall n$ is without loss of generality, since we can always pad the questions to be longer than some constant length before we start.) We would like to find a $k$
      such that, if $\omega_q(G_n^{BareIntro}) < 1 - \eps(n)$, then $\omega_q\left((G_n^{BareIntro,orac})^{\otimes k}\right) < \frac{1}{2}$; this would imply that, if $\omega_q\left((G_n^{BareIntro,orac})^{\otimes k}\right) \geq \frac{1}{2}$, then $\omega_q(G_n^{BareIntro}) \geq 1 - \eps(n)$, which in turn implies that $\omega_q(G_n) \geq 1 - \delta(\eps(n), n) \geq \frac{1}{2}$.
    
    Now, the result of \cite{DSV15}
      states the following: there are universal constants $C, c$, such
      that
      \begin{align*}
        \omega_q((G_n^{BareIntro, orac})^{\ot k})
        &\leq (1 - C(1 -
          \omega_q(G^{BareIntro, orac}_n))^c)^{k/2}\\
        &\leq ( 1 - C(\eps)^c)^{k/2}        \\
        &\leq \exp\left(-C(\eps)^c \cdot \frac{k}{2} \right),
      \end{align*}
      where in the last step we used that $(1 - \eps)^n \leq
      \exp(-\eps n)$ for all $n > 0$. If we choose $k$ such that
      \begin{equation} k \geq \frac{2 \log 2}{C} \eps^{-c} =
        \poly(\log q(n)), \label{eq:num-repetitions} \end{equation}
      then $\omega_q((G_n^{BareIntro, orac})^{\ot k}) \leq 1/2$ as desired.
    
For this choice of $k$, the resulting game $(G^{BareIntro, orac})^{\otimes k}$ is the introspected
    game $G^{intro}$.

    We now indicate how to show the properties claimed in the theorem
    statement.
    \begin{itemize}
    \item \textbf{Question length:} The question length of the bare
      introspected game is $O(\log q(n))$. Oracularization increases
      this by a factor of $2$, and parallel repetition by a factor of
      $k(n) = \poly\log q(n)$, resulting in a total question length of
      $\poly\log q(n)$.
      \item \textbf{Answer length:} The answer length of the bare
      introspected game is $q(n) + a(n) + O(1)$. Oracularization
      and parallel repetition increase this by a factor of $\poly\log
      q(n)$ resulting in a total answer length of $O((q(n)  + a(n) +
      O(1)) \cdot \poly\log q(n))$.
    \item \textbf{Sampler:} The sampler essentially has to run $k(n)$
      copies of the oracularization of the bare introspected sampler
      in parallel. Each of these takes time $O(t^{PB}_s(q(n)) +
      O(1))$, so the total runtime is $O(\poly\log q(n) \cdot t^{PB}_s(q(n)))$.
    \item \textbf{Decider runtime:} The runtime for the bare
      introspected decider is $t_d(n)  + t_d^{PB}(n) + O(\log n)$, where
      $t_d^{PB}(n)$ is the decider runtime for the Pauli basis
      test and the $O(\log n)$ comes from the runtime of $Q$. The total runtime scales is the bare runtime multiplied by
      the number of repetitions, which is $\poly \log q(n)$. Thus, the
      total runtime is $\poly \log q(n) \cdot (t_d(n) + t_d^{PB}(n) +
      O(1))$. 
    \item \textbf{Completeness:} This is a straightforward consequence
      of the completeness of bare introspection, oracularization, and
      parallel repetition.
    \item \textbf{Soundness:} The soundness analysis of bare
      introspection in \cite{JNVWY20} gives a bound in terms of $n$,
      independent of the question length. However, by examining the
      soundness argument in \cite[Section 8.4]{JNVWY20}, we see that
      the soundness function $\delta(\eps, n)$ depends only on $\eps$
      and $R(n)$, the bit length of the random seed of the questions to
      be compressed. In fact, we have
      \[ \delta(\eps, n) = a((\log R(n))^a \eps^{b} + (\log
        R(n))^{-b}). \]
      In our setting $R(n) = q(n)$, yielding the soundness claimed
      above in \Cref{eq:gapped-intro-soundness}.
      
      Thus, by the soundness of bare introspection,
      if $\omega_q(G) \leq 1/2$, then $\omega_q(G^{BareIntro})
      \leq 1 - \eps(n)$, and thus by the soundness of oracularization,
      $\omega_q(G^{BareIntro, orac}) \leq 1 - 2\eps(n)$. By the parallel repetition theorem of
      \cite{DSV15} together with the choice of $k(n)$, it holds that
      if $\omega_q(G^{intro}) \leq (1 - C(2\eps(n))^{c})^{k/2} \leq
      1/2$ as desired.
    \item \textbf{Efficient computability:} This follows from
      inspecting the sampler and decider for $G^{intro}$.
    \end{itemize}
\end{proof}

\subsubsection{Gapped answer reduction}

\begin{theorem} \label{thm:ar-gapped}
There is a Turing machine $\mathrm{GappedAnsReduce}$ with the following properties.
Let $G = (S,D)$ be a game family with complexity bounds $q(n), a(n),
t_s(n), t_d(n)$ respectively, where $a(n) \leq t_d(n)$ and
$q(n) \leq t_s(n)$, and such that $(S,D)$ is an $\ell$-level normal form verifier. Then
$\mathrm{GappedAnsReduce}$ given a description of $S,D$ outputs a description of a game family $G^{ans} = (S^{ans}, D^{ans})$, such that for the game family $G^{ans}$ and for all $n \geq 2$, we have
\begin{enumerate}
	 \item (\textbf{Conditional sampler:}) $(S^{ans}, D^{ans})$ is a $\max(\ell + 2,5)$ normal form verifier.
    \item (\textbf{Question length:}) The question length is
      $q^{ans}(n) = \poly(t_s(n), \log(t_d(n)), \log|D|)$, where $|D|$ is the description length of the
      Turing machine $D$. 
    \item (\textbf{Answer length:}) The answer length is $a^{ans}(n) = \poly(\log(t_d(n)),\log(|D|))$.
    \item (\textbf{Sampler runtime:}) The sampler runtime is $\poly(t_s(n),t_d(n),|D|)$. If $S$ is an $\ell$-level conditional linear sampler, then the sampler $S^{ans}$ is an $\ell'$-level conditional linear samper for $\ell' = \max(\ell+2,5)$. 
    \item (\textbf{Decider runtime:}) The decider runtime is $\poly(t_s(n),t_d(n),|D|)$.
    \item (\textbf{Completeness:}) If $G_n$ has a finite-dimensional oracularizable synchronous strategy of value 1, then $G_n^{ans}$ has a finite-dimensional oracularizable synchronous strategy with value 1. 
    \item (\textbf{Soundness:}) If $\omega^s_q(G^{ans}) \geq 1 - \eps$, then
      $\omega^s_q(G_n) \geq 1 - \delta(\eps,n)$, where $\delta(\eps, n) = a((\log t_d(n) + \log |D|)^a \eps^b + (
      \log t_d(n) + \log |D|)^{-100 b})$, $a$ is a universal constant such that $a > 0$, and $b$ is a universal constant such that $0 < b < 1$.
        \item (\textbf{Efficient computability:}) 
      $\mathrm{GappedAnsReduce}$ runs in time
      $O(|D| + |S|)$. Moreover, $|D^{ans}| = |D| + O(1)$, and
      $|S^{ans}| = |S| + O(1)$. 
%    \item (\textbf{Entanglement bound:}) $\cal{E}(G^{ans}_n, 1 - \eps) \geq \cal{E}(G_n, 1 - \delta(\eps,n))$. \znote{Define this version of the entanglement bound that has a soundness parameter in it?}
\end{enumerate}
\end{theorem}

\begin{proof}
This theorem is essentially \cite[Theorem~10.27]{JNVWY20}, but with a better dependence on $|D|$ obtained by using a slightly different PCP construction, that is more suited to our ``scaled down" setting, where we are imagining reducing from polynomial to polylogarithmic answer size, without shrinking the decider runtime. This is in contrast to \cite{JNVWY20}, where the answer reduction was from exponential to polynomial size in the index $n$, and reduced the decider runtime from exponential to polynomial as well. The PCP construction we will use is slightly different, so we will start by reviewing the scaled-up case in brief.

\paragraph{The scaled up case: review} In \cite{JNVWY20}, to perform answer reduction, one starts by applying the Cook-Levin theorem to the (exponential-runtime ) decider to create an instance of Succinct Circuit SAT, specified by another circuit $\mathcal{C}$ consisting of $s = \poly(|D|)$ gates, and taking polynomially many inputs labeled $x,o$. This circuit is then converted into a 5SAT formula $\mathcal{F}(x,o,w)$ where $w \in \mathbb{F}_q^s$ using the Tseitin transformation together with some additional gadgets as described in \cite[Section 10.3]{JNVWY20}. In this transformation, we need to add the $s$ new variables labeled $w$---essentially one per wire in the circuit---to keep the degree per variable down to a constant. The original Succinct Circuit SAT instance is satisfiable if and only if the following holds:
    \[ \forall x,o, \exists w \mathcal{F}(x,o,w) = 1. \]
    
    To check this, the answer-reduced decider in \cite{JNVWY20} computed an \emph{arithmetization} of $\mathcal{F}$, converting it into a polynomial $\mathcal{F}_{arith}: \mathbb{F}_q^{m'}$, where $m'$ is the total number of variables $x,o,w$. Identifying a subset $H \subset \mathbb{F}_q$ with $\{0,1\}$, the values of $\mathcal{F}_{arith}$ on $H^{m'}$ are exactly equal to values of the Boolean formula $\mathcal{F}$ on Boolean assignments in $\{0,1\}^m$.
    
    Now, to check the desired predicate, the decider asks the provers to compute low-degree polynomials $g_1, \dots, g_5$ encoding assignments to the $w$ variables. It makes queries to these polynomials and to some additional low-degree polynomials to certify that  
    \[ \mathcal{F}_{arith}(x,o,w) (g_1(x) - o_1) \dots (g_5(x) - o_5). \]
    is low-degree and identically $0$ on $H^{m'}$. \anote{Explain more}
\paragraph{The scaled-down setting}

In our setting, we are instead starting from a polynomial-time decider, which we will convert into a Circuit SAT instance using the Cook-Levin theorem. Given a circuit $\mathcal{C}$, we create a 5SAT formula $\mathcal{F}(a, b, w)$ as before. However, the predicate we wish to check now is just
    \[ \exists a, b, w \mathcal{F}(a,b,w) = 1. \]
    The formula $\mathcal{F}$ is a conjunction of clauses, each of which depends on at most 1 variable each from the blocks $a,b,w_1, w_2, w_3$. 
    
    We will now proceed to convert this formula into a polynomial. Let us fix a field $\mathcal{F}_q$ (with size $q$ to be chosen later), and define $H$ to be two points in $\mathbb{F}_q$ associated with $\{0,1\}$. Define the polynomial $\phi$ by
    \[ \phi(x_1, x_2, \dots, x_5, o_1, o_2, \dots, o_5) = \begin{cases} 1 & \text{if $a_{x_1} ^{o_1} \vee b_{x_2}^{o_2} \vee w_{1,x_3}^{o_3} \vee w_{2,x_4}^{o_4} \vee w_{3,x_5}^{o_5}$ is a clause in $\mathcal{F}$} \\ 0 & \text{otherwise} \end{cases}, \]
    where $x_1, \dots, x_5$ are indices in $[n]$ identified with points on the subcube $H^m$. 
    To explicitly construct $\phi$, recall the indicator polynomial
    \begin{align}
        \mathrm{ind}_{m,y}(x) &= \begin{cases} 1 & \text{if } x = y \\ 0 & \text{otherwise} \end{cases} \\
        &= \prod_{i: y_i = 1} x_i \prod_{i: y_i = 0} (1 - x_i).
    \end{align}
    This polynomial has individual degree $1$ in each variable (and thus total degree $m$). Using this, we may write $\phi$ as a sum of terms, each of which is a product of 10 indicators over the $a,b,w_1, w_2, w_3, o_1, \dots, o_5$ variables respectively. Overall $\phi$ has individual degree $1$ in each variable and total degree $5m + 5$.

    The formula $\mathcal{F}$ is satisfiable iff for all $x_1, \dots, x_5 \in H^{m}$ and $o_1, \dots, o_5 \in H$,
    \[ c_0(x_1, \dots, x_5, o_1, \dots, o_5) = \phi(x_1, \dots, x_5, o_1, \dots, o_5)  \cdot (g_1(x_1) - o_1)(g_2(x_2) - o_2) \dots (g_5(x_5) - o_5) \]
    is equal to 0. This can also be checked by the zero-on-subcube test: we check that $g$ is a member of the ideal of polynomials that zero on $H^{5m+5} $ by asking for a decomposition in terms of the generators of that ideal. Specifically, if we can write
    \[ c_0(z) = \sum_{i=1}^{5m+5} c_i(z) \cdot \mathrm{zero}(z_i), \]
    where $z \in \mathbb{F}_q^{5m+5}$ is shorthand for the $5m+5$ variables $x_1, \dots, x_5, o_1, \dots, o_5$, each $c_i(z)$ is an arbitrary polynomial of the appropriate degree $d$ and $\mathrm{zero}(\cdot)$ is a fixed univariate polynomial that is zero on $H$, then we have certified that $c_0$ is zero on $H^{m'}$. 
    This is done on the zero-on-subcube test by querying $c_0$ and $c_i$ on random points and subspaces in $\mathbb{F}_q^{5m+5}$.

    \paragraph{Setting of parameters}
    Let $G = (S,D,Q)$ be the game family for which we wish to perform answer reduction, and let $q(n), a(n), t_s(n), t_d(n)$ be its complexity bounds. We will now choose parameters for the PCP to perform answer reduction.
    
    First, by applying the Cook-Levin theorem to the decider $D$, we obtain a 5SAT instance with $s = \poly(t_d(n), |D|)$ variables and clauses. (See \cite{JNVWY20} for a detailed description of the Cook-Levin reduction; note that we do not need a succinct representation of the resulting formula here.)
    
    Next, we must choose the PCP parameters $q, m, m', d$ as follows. We will let $\gamma$ be a natural number-valued parameter that we will set later, which we will use the adjust the soundness as required, and let $a,b$ be universal constants that will be below as part of the definition of $\delta_{LD}$.
    \begin{enumerate}
            \item Let $m = \lceil \log s \rceil = O(\log t_d(n) + \log |D|)$.
        \item Let $m' = 5m + 5$.
        \item Let $q = 2^k$ where $k$ is the smallest odd integer satisfying:
        \begin{align}
        k &\leq \gamma m \label{eq:kleqgammam} \\
        2^{-kb } &\leq m^{-(\gamma b)} \label{eq:twokb} \\
        km'/2^k &\leq m^{-b \gamma} \label{eq:mdoverq},
        \end{align}
        and such that $2^k$ is divisible by $m'$
        \item   Let $d = k$.
    \end{enumerate}
    We note that asymptotically, this gives us $d = k = O(m) = O(\log t_d(n) + \log |D|)$
   To compute the soundness of the PCP, we will need the soundness function of the low degree test, defined as
   \begin{equation}
       \delta_{LD}(\eps, q, m,d,r) = a(dmr)^a (\eps^b + q^{-b} + 2^{-bmd}).
   \end{equation}
   The interpretation of this is that for any strategy for the $r$-simultaneous low-degree test with field size $q$, $m$ variables, individual degree $d$, and success probability $1 - \eps$ can be rounded to one that samples a tuple of $r$ polynomials over $\mathbb{F}_q^m$ with individual degree $d$, and the error incurred in the rounding is $\delta_{LD}$.
   
    With the parameter settings above, we have that 
       \begin{align}
       \delta_{LD}(\eps, q, m', d, m' + 6)  &= a(dm'(m'+6))^{a} (\eps^b + q^{-b} + 2^{-bm'd}) \\
       &\leq a(dm'(m'+6))^{a}( \eps^b + q^{-b}) \\
       &\leq a (c \gamma m^3)^a (\eps^{b} + q^{-b}) \\
       &\leq a (c\gamma)^a (m^{3a} \eps^{b} + m^{3a} m^{-(3a + \gamma b)}) \\
       &= a(c\gamma)^a (m^{3a} \eps^{b} + m^{-\gamma b}),
   \end{align}
   where $c$ is some universal constant.
    Here, in going from the first to the second line, we used that $2^{m'd} = 2^{m'k} \geq 2^k = q$. In going from the second to the third, we used \Cref{eq:kleqgammam}. In going from the third to the fourth, we used that \Cref{eq:twokb}. 
    
    \paragraph{Question, answer, and runtime complexity of the test}
    The answer-reduced verifier computes a pair of questions $(x,y)$ from the question distribution of the original game $G$, and sends them to the provers. For a particular $(x,y)$, the answer-reduced verifier will then make queries to the PCP for the decision predicate induced by $x,y$. These will consist of queries to the polynomials $c_0$ and $c_1, \dots, c_{m'}$ defined above, as well as queries to the query polynomials $g_1, \dots, g_5$ that are supposed to be low-degree extensions of the strings $a, b, w_1, w_2, w_3$. This makes a total of $m' + 1 + 5 = m' + 6$ polynomials over $\mathbb{F}_q^{m'}$, with individual degree at most $d$ in each variable. The question and answer complexity is dominated by the simultaneous low-degree test to these $m' + 6$ polynomials. 
    
    In this test, the question size in bits is 
    \[ q^{ans}(n) = q(n) + O(m' \log q) =  \poly( t_s(n), \log t_d(n), \log |D|),\]
    using the bound that the question length of the original game $q(n)$ is upper-bounded by the sampler runtime $t_s(n)$. The answer size is $O((m' + 6) \cdot (d + 1) \cdot \log q)$---the number of bits needed to specify the $\mathbb{F}_q$-valued coefficients of $m'+6$ univariate polynomials with degree at most $d$. This scales as 
    \[ a^{ans}(n) = O(m^3) = \poly(\log t_d(n), \log |D|). \] 
    
    For the decider runtime, the answer-reduced decider must first compute the SAT formula produced by the Cook-Levin reduction for the original decider $D$. This takes time $\poly(|D|, t_s(n), t_d(n))$. Next, the answer-reduced decider executes the PCP verifier for this SAT formula together with the low-degree test. This takes time that is $\poly(|D|, t_d(n), m', d, \log q)$. Thus the total runtime of the decider scales as $\poly(t_s(n), t_d(N), |D|)$. A similar bound may be obtained for the sampler runtime.
    
   \paragraph{Soundness analysis}
   
   In our modified answer reduction, we perform the same test as in \cite[Section 7]{JNVWY20}, but with a modified definition for the polynomial $c_0$, and without the $w$ variables. The majority of the soundness analysis can be reused with the appropriate change in parameters. In particular, it is shown in the proof of the soundness part of \cite[Theorem~10.27]{JNVWY20} that any strategy for the answer reduced game winning with probability $1 - \eps$ implies a strategy for the original game winning with probability $1 - \delta$ with $\delta= O((\delta_{LD}(\eps, q, m', d, m'+6)^{1/2} + (m'd/q))^{1/2})$.
    With our setting of parameters, this yields
    \begin{align*}
        \delta &\leq \left( C a^{1/2} (c \gamma)^{1/2} (m^{3a} \eps^{b} + m^{-\gamma b})^{1/2} + \left(\frac{m'd}{q}\right)^{1/2}\right)^{1/2} \\
        &\leq \left( C a^{1/2} (c \gamma)^{1/2} (m^{3a} \eps^{b} + m^{-\gamma b})^{1/2} + m^{-\gamma b/2}\right)^{1/2} \\
        &\leq \left( 2C a^{1/2} (c \gamma)^{1/2} m^{1.5 a} \eps^{0.5 b} + (2 Ca^{1/2} (c\gamma)^{1/2} + 1) m^{-\gamma b/2}\right)^{1/2} \\
        &\leq 4 C^{1/2} (ca\gamma)^{1/4} m^{0.75 a} \eps^{0.25 b} + 2(2C(ca\gamma)^{1/2} + 1)^{1/2} m^{-0.25\gamma b}.
    \end{align*}
    where in going from the first to th second line, we used \Cref{eq:mdoverq}, and in the subsequent steps we used that $\sqrt{x + y} \leq 2(\sqrt{x} + \sqrt{y})$. This is of the form claimed in the theorem if we set $\gamma = 100$ and recall that $m = O(\log t_d(n) + \log |D|)$.
    
    \paragraph{Efficient computability}
    The new sampler $S^{ans}$ generates question pairs from $S$ together with question pairs used in the low-degree tests performed by the answer-reduced verifier. The description length of the Turing machine that generates these additional questions is $O(1)$, so the total description length of $S^{ans}$ is $|S| + O(1)$. Likewise, for the decider, the answer reduced decider computes the Boolean circuit evaluated by the original decider $D$, applies the Cook-Levin reduction to convert it to a SAT formula, and then executes the appropriate PCP verifier. The description length of the Turing machine that performs these steps is $|D| + O(1)$. Both $S^{ans}$ and $D^{ans}$ can be computed by concatenating descriptions of $S,D$ with constant-sized strings, so the claimed runtime bound on $\mathrm{GappedAnsReduce}$ follows.

% \znote{Below is old}
% \begin{enumerate}
% \item \textbf{Answer length:}
% This comes from inspecting the answers in Figure~14 of \cite{JNVWY20}. The longest answer consists of the specification of $m'$ degree-$d$ polynomials over $\mathbb{F}_q$, where $m' = \poly(\log t_s(n), \log t_d(n), \log(|D|))$, $d = O(\log m')$, and $q = 2^d$. Specifying a single such polynomial requires $d+1$ coefficients in $\mathbb{F}_q$, or a total of $(d+1)\log q$ bits. So the total answer length is at most $m'\cdot (d+1) \log q = \poly(\log t_s(n), \log t_d(n), \log |D|)$.
% \end{enumerate}
\end{proof}

\begin{theorem}\label{thm:ar-gapped-pr}
  There is a Turing machine $\mathrm{GappedAnsPR}$ with the following properties.
Let $G = (S,D)$ be a family of gapped games with question length $q(n)$,
answer length $a(n)$, decider runtime $t_d(n)$, and sampler runtime
$t_s(n)$, where $t_s(n), t_d(n) = \poly(n)$, and where the sampler is
an $\ell$-level conditional linear sampler with $\ell \leq 10$. Then
$\mathrm{GappedAnsPR}$ given a description of $S, D$ outputs a
triple $S^{ansPR}, D^{ansPR}$ defining a family of games
$G^{ansPR}$ such that the
following hold.

Let $q^{ans}(n)$ be defined in terms of $t_s(n), t_d(n), |D|$ as it is in \Cref{thm:ar-gapped}. Let $k(n)$ be a parameter such that 
$k(n) = \poly(\log t_d(n) + \log |D|)$. 

\begin{enumerate}
    \item (\textbf{Question length:}) The question length is
      upper-bounded by $q^{ansPR}(n) = k(n) \cdot \poly(t_s(n), \log(t_d(n)), \log |D|)$.
    \item (\textbf{Answer length:}) The answer length of $G^{ansPR}_n$
      is $k(n) \cdot \poly(\log (t_d(n)), \log(|D|))$. %\label{item:intro-answer-gap}
    \item (\textbf{Sampler:}) The total sampler
      runtime is $t_{s}^{ansPR}(n) = k(n) \cdot \poly(t_s(n), t_d(n), |D|)$.
      If $S$ is an $\ell$-level conditional linear sampler, then the sampler $S^{ansPR}$ is an $\ell'$-level conditional linear sampler for $\ell' = \max(\ell+4,7)$. 
    \item (\textbf{Decider runtime:}) The runtime of the decider $D^{ansPR}$
      is bounded by $t_{d}^{ansPR}(n) = k(n) \cdot \poly(t_s(n), t_d(n), |D|)$. 
    \item (\textbf{Completeness:}) If there is an oracularizable
      finite-dimensional synchronous
      strategy for $G_n$ achieving value $1$, there exists an oracularizable
      finite-dimensional synchronous
      strategy for $G^{ansPR}_n$ achieving value $1$. 
    \item (\textbf{Soundness:}) If $\omega_q(G_n) \leq \frac{1}{2}$, then
      $\omega_q(G^{ansPR}_n) \leq  
      \frac{1}{2}$. \label{item:ar-soundness-gap}
    % \item (\textbf{Entanglement bound:}) \anote{Wrong: $\cal{E}(G^{intro}_n, 1) \geq
    %   \max\{\cal{E}(G_n), 2^{2n}\}$.} \label{item:intro-entanglement-gap}
    \item (\textbf{Efficient computability:}) 
      $\mathrm{GappedAnsPR}$ runs in time
      $O(|D| + |S|)$. Moreover, $|D^{ansPR}| = |D| + O(1)$, and
      $|S^{ansPR}| = |S| + O(1)$.
      
    %   \label{item:intro-efficient-computability-gap}
  \end{enumerate}
\end{theorem}

\begin{proof}
This theorem is obtained by applying the parallel repetition of \cite{DSV15} to \Cref{thm:ar-gapped}. Specifically, given the triple $(S^{ans}, D^{ans})$ which is output by \Cref{thm:ar-gapped}, let $\eps \geq 0$ be such that $\omega_q(G^{ans}_n) \leq 1-\eps$. We apply the parallel repetition of \cite{DSV15} to this game with $k$, the number of repetitions, set as it is in the theorem statement. This has the straightforward effect on question and answer length and sampler and decider runtime of multiplying all these quantities by $k$. In order to obtain the claimed soundness bound, note that \cite{DSV15} gives
\begin{align}
        \omega_q((G_n^{ans})^{\ot k})
        &\leq (1 - C \cdot (1 -
          \omega_q(G^{ans}_n))^c)^{k/2}\\
        &\leq ( 1 - C\cdot (\eps)^c)^{k/2}        \\
        &\leq \exp\left(-C(\eps)^c \cdot \frac{k}{2} \right), \label{eq:anspr-general}
\end{align}
where $C,c$ are universal constants.

As in the introspection case, we need to deal with the fact that answer reduction is not gap-preserving. Rather, it has the
      following soundness guarantee: if $\omega_q(G_n^{ans}) \geq 1 -\eps$,
      then there exist constants $a > 0, 0 < b < 1$ such that
      \begin{equation}
        \omega_q(G_n) \geq 1 - \delta(\eps, n), \qquad 
      \delta(\eps, n) = a((\log t_d(n) + \log |D|)^a \eps^b + (
      \log t_d(n) + \log |D|)^{-100b}).\end{equation}
      Let $\eps(n)$ be a function such that $\delta(\eps(n),n) < \frac{1}{2}$ for all $n \geq 2$. We may assume WLOG that such an $\eps(n)$ always exists by padding the the runtime of $D$ so that $\delta(0,2) < 1/2$. 
      
      Note that $\eps(n)$ can be chosen such that $1/\eps(n) \leq \poly(\log t_d(N) + \log |D|)$. We would like to find a $k$
      such that, if $\omega_q(G_n^{ans}) < 1 - \eps(n)$, then $\omega_q\left((G_n^{ans})^{\otimes k}\right) < \frac{1}{2}$; this would imply that, if $\omega_q\left((G_n^{ans})^{\otimes k}\right) \geq \frac{1}{2}$, then $\omega_q(G_n^{ans}) \geq 1 - \eps(n)$, which in turn implies that $\omega_q(G_n) \geq 1 - \delta(\eps(n), n) \geq \frac{1}{2}$. Using \Cref{eq:anspr-general}, we see that if we choose $k(n)$ such that
      \begin{equation} k(n) \geq \frac{2 \log 2}{C} \eps(n)^{-c} =
        \poly(\log t_d(n) + \log |D|), \end{equation}
      then $\omega_q((G_n^{ansPR})^{\ot k}) \leq 1/2$ as desired. This is precisely the setting of $k(n)$ from the theorem statement.
      
%      Then, if $G_n$ satisfies
%      $\omega_q(G_n) \leq 1/2$, it holds that there exists a function
%      $\eps(n)$ with $1/\eps(n) = \poly( q^{ans}(n))$, such that
%      $\omega_q(G_n^{ans}) \leq 1 - \eps(n)$. 

      It remains to show two things: the number of levels of the sampler, and the  efficient computability. For the former, it is shown in \cite[Theorem~11.4]{JNVWY20} that if the answer-reduced sampler has $\ell'$ levels, then the parallel-repeated sampler has $\ell' + 2$ levels, so the conclusion follows from the bound on the number of levels in \Cref{thm:ar-gapped}. For the efficient computability, this follows because the description of $G^{ansPR}$ consists of a description of $G^{ans}$ together with the code to generate the parallel repetitions, which has size $O(1)$.

\end{proof}
\subsubsection{Gapped compression}
By composing question reduction (\Cref{thm:intro-gap}) and parallel-repeated answer reduction (\Cref{thm:ar-gapped-pr}), one obtains a gapped
compression theorem. We state the bounds
that we obtain below.
\begin{theorem}\label{thm:gapped-compress}
  There exists a Turing machine $\mathrm{GappedCompress}$ with the following properties. Fix a universal constant $q^*$. Let $G = (S,D)$ be a game family with complexity bounds $q(n), a(n),
t_s(n), t_d(n)$ respectively, such that $q(n) > q^* \: \forall n \in \mathbb{N}$. Then on input $(S,D)$, $\mathrm{GappedCompress}$ returns a game
family $G' = (S',D')$
with parameters $q'(n), a'(n), t_s'(n), t_d'(n)$ such that 
\begin{enumerate}
    \item (\textbf{Question length:}) The question length of $G'_n$ is
      \[q'(n) = \poly(\log q(n) \cdot t_s^{PB}(q(n)),\:\: \log t_d(n) ,\:\: \log t_d^{PB}(q(n)),\:\: \log |D|,\:\: \log\log n).\]
    \item (\textbf{Answer length:}) The answer length of $G'_n$ is
      \[a'(n) = \poly(\log\log q(n) ,\:\: \log t_d(n) , \:\: \log t_d^{PB}(q(n)),\:\: \log |D|,\:\: \log\log n).\]
    \item (\textbf{Sampler:}) The sampler runtime is
    \[t_s'(n) = \poly(\log q(n) \cdot t_s^{PB}(q(n)),\:\:  t_d(n), \:\: t_d^{PB}(q(n)),\:\: |D|,\:\: \log n).\]
    If $S$ is $\ell$-level conditionally linear with $\ell \leq 10$, then $S'$ is conditionally linear with $\ell \leq 9$.
    \item (\textbf{Decider runtime:}) The decider runtime is
    \[t_d'(n) = \poly(\log q(n) \cdot t_s^{PB}(q(n)),\:\:  t_d(n) ,\:\: t_d^{PB}(q(n)),\:\: |D|,\:\: \log n).\]
    \item (\textbf{Completeness:}) If there is an oracularizable
      finite-dimensional synchronous
      strategy for $G_n$ achieving value $1$, there exists an oracularizable
      finite-dimensional synchronous
      strategy for $G^{ansPR}_n$ achieving value $1$.
    \item (\textbf{Soundness:}) If $\omega_q(G_n) < \frac{1}{2}$, then
      $\omega_q(G'_n) < \frac{1}{2}$.
      \item (\textbf{Efficient computability:}) $\mathrm{GappedCompress}$ runs in time
      $O(|D| + |S|)$. Moreover, $|D'| = |D| + O(1)$, and
      $|S'| = |S| + O(1)$. 
%    \item (\textbf{Entanglement bound:}) $\cal{E}(G'_n) \geq
%      \max\{\cal{E}(G_n), 2^{2n}\}$
\end{enumerate}

\end{theorem}
\begin{proof}
  Compose \Cref{thm:intro-gap} and \Cref{thm:ar-gapped-pr}. The parameters behave as follows.
  \begin{itemize}
      \item \textbf{(Question length:)}
      Recall that $k(n) = \poly(\log t_d^{intro}(n) + \log |D^{intro}|) = \poly(\log t_d(n) + \log t_d^{PB}(q(n)) + \log\log (n) +  \log\log q(n) + \log |D|)$. The question length is 
      \begin{align*}
      q'(n) &= k(n) \cdot \poly(t_s^{intro}(n), \log(t_d^{intro}(n)), \log |D^{intro}|) \\
      &= \poly(\log q(n) \cdot t_s^{PB}(q(n)), \log t_d(n), \log t_d^{PB}(q(n)), \log\log(n), \log |D|).
      \end{align*}
      \item \textbf{(Answer length:)} The answer length is
      \begin{align*}
          a'(n) &= k(n) \cdot \poly(\log(t_d^{intro}(n)), \log(|D^{intro}|) )\\
          &= \poly(\log t_d(n), \log t_d^{PB}(q(n)), \log \log(n), \log \log q(n), \log |D|).
      \end{align*}
      
      \item \textbf{(Sampler runtime:)} The sampler runtime is
      \begin{align*}
          t_s'(n) &= k(n) \cdot \poly(t_s^{intro}(n),  t_d^{intro}(n), |D|) \\
          &= \poly(\log q(n) \cdot t_s^{PB}(q(n)), t_d(n), t_d^{PB}(q(n)), \log q(n), \log(n), |D|).
      \end{align*}
      
      \item \textbf{(Decider runtime:)} The decider runtime is
      \begin{align*}
          t_d'(n) &= k(n) \cdot \poly(t_s^{intro}(n), t_d^{intro}(n), |D^{intro}|) \\
          &= \poly(\log q(n) \cdot t_s^{PB}(q(n)), t_d(n), t_d^{PB}(q(n)), \log q(n), \log(n), |D|).
      \end{align*}
  \end{itemize}
  The remaining properties follow directly from composing \Cref{thm:intro-gap} and \Cref{thm:ar-gapped-pr}. For reference, the parameters for $G^{intro}$ are given in the table below.

  \begin{center}
    \begin{tabular}{ c | c }
    $G$ & $G^{intro}$  \\ \hline
    $q(n)$ & $\poly\log q(n)$ \\
     $a(n)$  & $(a(n) + q(n) + O(1)) \cdot \poly\log q(n)$  \\
     $t_s(n)$ & $\poly\log q(n) \cdot t_s^{PB}(q(n))$  \\
     $t_d(n)$ & $(t_d(n) + t_d^{PB}(q(n)) + O(\log(n)) \cdot \poly\log (q(n))$  \\
     $|D|$ & $|D| +O(1)$ \\
     $|S|$ & $|S| + O(1)$ 
\end{tabular}
\end{center}
\end{proof}

\subsection{Gapped hypercompression}

\begin{theorem}
  \label{thm:gapped-hyper}
  There exists a Turing machine $\mathrm{GappedHypercompress}$ that, given as input
  the description of a tuple of Turing machines $G^\hf = (S^\hf,D^\hf)$, outputs a
  description of a tuple $G^{\hp} = (S^{\hp},D^{\hp})$ with the
  following properties. Let $G^\hf = (S^{\hf},D^{\hf})$ be an $\MIP^*[ q^{\hf}(n), a^{\hf}(n),
  t_s^{\hf}(n), t^{\hf}_d(n)]$
  protocol. Then $G^{\hp}$ is an $\MIP^*[Q_0, a^{\hp}(n),
  t_s^{\hp}(n), t_d^{\hp}(n)]$ protocol with
  \begin{enumerate}
  \item \textbf{Question length:} $Q_0$ a universal constant such that, for all $\ell > Q_0$, $\max(C \log^\beta(\ell), C \log^\beta(q^*)) < \ell$, where $C, \beta, q^*$ are the universal constants that appear in the question length quoted in \Cref{thm:intro-gap}.
  \item \textbf{Answer length:} $a^{\hp}(n) = a^{\hf}(n) \cdot \poly\log(n)$. 
  \item \textbf{Sampler runtime:} $t_s^{\hp} = O(1)$.
  \item \textbf{Decider runtime:} $t^{\hp}(n) = \poly(n) + O((\log \log
\log q^{\hf}(n))^2 \cdot (|D^{\hf}| + |S^{\hf}|))$
  \item \textbf{Completeness:} For all $n$, if $G^{\hf}_n$ has a value-$1$
    oracularizable finite-dimensional synchronous strategy $\scr{S}^{\hf}$, then
    $G^{\hp}_n$ has a value-$1$ oracularizable finite-dimensional synchronous strategy
    $\scr{S}^{\hp}$. 
  \item \textbf{ Soundness:} For all $n$, if $\omega^s_q(G^{\hf}_n)
    \leq 1/2$, then $\omega^s_q(G^{\hp}_n) \leq 1/2$ as well.
  \end{enumerate}
\end{theorem}
\begin{proof}
  The protocol $G^{\hp} = (S^{\hp}, D^{\hp})$ consists of 
  a sampler and a decider as follows.
  \begin{enumerate}
    \item \textbf{Sampler:}
    The sampler is the same as the sampler $S^{intro}$ from
    \Cref{thm:intro-gap} for question
    length $Q_0$.
  \item \textbf{Decider:}
    The decider executes the algorithm in \Cref{alg:hypercomp-pseudo-gapped}
  \end{enumerate}

  \vspace{10pt}
  \IncMargin{1em}
  \begin{algorithm}[H]
    \DontPrintSemicolon

    \textbf{Input}: $n, x, y, a, b$
    
    Initialize $i := 0$, $\ell_0 := q^{\hf}(n)$, $D_{comp} = D^{\hf}$, $S_{comp} =
    S^{\hf}$.

    \While{$\ell_i > Q_0$}{
      $D_{comp}, S_{comp} := \mathrm{GappedIntro}(D_{comp},
      S_{comp})$

      $\ell_{i+1} := \lceil C \log^{\beta} \ell \rceil$
      $i := i+1$
    }
    Run $D_{comp}(n,x,y,a,b)$. Accept if it accepts, otherwise reject.
    \caption{The hypercompressed game.}\label{alg:hypercomp-pseudo-gapped}
  \end{algorithm}\DecMargin{1em}
  \vspace{10pt}
\begin{enumerate}
\item   \textbf{Question length:} For any $G^{\hf}$, the hypercompressed game $G^{\hp}_n$ has question length at most $Q_0$.

 This follows from the definition of $G^{\hp}_n$ (if the question length were larger than $Q_0$, we could apply $Intro$ another time and strictly reduce the question length).

\item \textbf{Runtime: } First, we claim that for any game $G^{\hf}$ with question length $q^{\hf}(n)$,
  the hypercompressed game $G^{\hp}_n$ is obtained by applying
  $\mathrm{GappedIntro}$ to $G^{\hf}$ at most $m = O(\log \log \log(q^{\hf}(n)))$
  times.  To show this, let $f(\ell) = C \log^{\beta} \ell$. We know
  from the previous item that the question length of $G^{\hp}_n$ is
  at most $Q_0$.   
% \begin{equation}
%     f(f(\ell)) = C \log^{100} (C \log^{100} \ell) = C ( \log C + 100 \log \log \ell)^{100}
% \end{equation}
Let $g(\ell)$ be the number of iterations of $f$ required to reach $Q_0$. We can define $g$ by the following recurrence:
\begin{equation}
    g(\ell) = \begin{cases} 
    0 & \ell \leq Q_0 \\
    1 + g(f(\ell)) & n > Q_0.
    \end{cases}
\end{equation}
We claim that $g(\ell) \leq G \log \log \log(\ell)$ for some constant
$G$, as long as $Q_0$ is chosen to be sufficiently large. (In fact this bound is
not tight, and $g$ grows much more slowly, but this will suffice for
our purposes.) We prove this by induction. The base case is $\ell \leq
Q_0$, in which case the bound holds. For the inductive step, write
\begin{align} 
g(\ell) &= 1 + g(f(\ell)) \\
&= 1 + g(C \log^{\beta}(\ell)) \\
&\leq 1 + G \log\log\log(C \log^{\beta}(\ell)) \\
&= 1 + G \log \log (\log C + \beta \log \log \ell) \\
&\leq 1 + G \log \log (2\beta \log \log \ell) \\
&= 1 + G \log\log 2\beta + G \log \log \log \log \ell \\
        &\leq 3 G \log \log \log \log \ell \\
  &\leq G \log \log \log \ell.
\end{align}
Here, in the first inequality, we have used the inductive hypothesis
plus the assumption that the choice of $Q_0$ is such that for all
$\ell > Q_0$, $C \log^\beta \ell < \ell$, and in the second inequality
we assume that $Q_0$ is large enough that $\beta \log \log Q_0 \geq
\log C$, and in the third inequality, we assume that $Q_0$ is
sufficiently large that both $1$ and $G \log \log 2\beta$ are less
than or equal to $G \log \log \log \log Q_0$, and in the fourth
inequality we assume that $\log \log \log \log Q_0 \leq \frac{1}{3}
\log \log \log Q_0$.  

This bounds the number of iterations of the loop. Each iteration calls
the procedure $\mathrm{GappedIntro}$ which takes time $O(|D_i| + |S_i|) = O(|D^{\hf}| + (i+1)|S^{\hf}|)$. The total runtime of the loop is
thus bounded by $O(m |D^{\hf}| + m^2|S^{\hf}|) = O( (\log\log\log q^{\hf}(n))^2
\cdot (|D^{\hf}| + |S^{\hf}|))$.

Finally, we must determine the runtime of the decider that is invoked
after the loop is over. This is polynomial in the answer length of the
final game, which will be shown in the following item to be
$O(\poly\log(n) \cdot a^{\hf}(n)) = \poly(n)$. So in total, we get a $\poly(n)$ runtime
for the final decider, and thus a runtime of $\poly(n) + O((\log \log
\log q^{\hf}(n))^2 \cdot (|D^{\hf}| + |S^{\hf}|))$ for the decider in total.
\item \textbf{Answer length:}
  In the $i$th iteration of the loop, the answer length of the game
  increases from $a_i(n)$ to $a_{i+1}(n)$ as
  described by the following formula:
  \begin{equation}
    a_{i+1}(n) = A(a_i(n) + \ell_i(n) + O(1)) \cdot \log^{\alpha} \ell_i(n),
  \end{equation}
  where $A, \alpha$ are constants. We can bound this, somewhat
  crudely, by
  \begin{equation}
    a_{i+1}(n) \leq A' a_i(n) \cdot \log^{\alpha} n,
  \end{equation}
  for some constant $A'$, using the bounds that $\ell_i(n) \leq a_i(n)$ and $\ell_i(n) \leq \poly(n)$
  for all $i$. To get a slightly sharper bound, we can find a constant $A'$ such that
  \begin{gather}
  a_1 \leq A' a(n) \cdot \log^\alpha n, \\
  a_{i+1} \leq A' a_i(n) \log^\alpha(\log^\beta n) \quad \text{for all $i \geq 1$}.
  \end{gather}
  
  Note that the right-hand-side of the latter equation is equal to $A'
  a_i(n) \cdot \beta^{\alpha} (\log\log n)^\alpha$. For notational simplicity, we will absorb $\beta$ into $A'$.

% where $m$ is the number of iterations of compression. So as long as $m$ is $O(\log n)$, we get that $a \mapsto \poly(n) \cdot a$. And if $m = O(\log \log \log n)$, as we show in the previous Fact, then $a \mapsto O(\poly \log \log n) a$.

% \textbf{If $d$ is not constant:} If we do not use the de la Salle game, but rather the original introspection, then $d = O(\poly\log(n))$. To be concrete, let's write $d = d_0 \log^{\rho}(n)$.
After $m = g(\ell) = G \log\log\log n$ repetitions, then, the answer size is
\begin{align}
    a_m &\leq (A' (\log\log n)^\alpha)^m (A' a^{\hf}(n) \cdot \log^\alpha n) \\
    &= 2^{(\log A' + \alpha \log\log\log n) \cdot G \log\log\log n} \cdot (A' a^{\hf}(n) \cdot \log^\alpha n) \\
    &= O\left( 2^{A'' (\log\log\log n)^2} \cdot a^{\hf}(n) \log^\alpha n \right) \\
    &= O\left( 2^{\log\log n} \cdot a^{\hf}(n) \log^\alpha n \right) \\
    &= O\left( \log n \cdot \log^\alpha n \cdot a^{\hf}(n) \right) \\
    &= \poly\log n \cdot a^{\hf}(n).
\end{align}

% \anote{Below is old}

% If $m = O(\log\log\log n)$, then $d^{m} \leq d^{O(\log \log \log n)} = (\log \log n)^{O(\log d)}$. So, assuming $a = \poly\log n$, we get
% \[ (2a)^{d^m} = 2^{d^m \cdot \log (2a)} \leq 2^{(\log \log n)^{O(\log d)} \cdot \log 2a} \leq 2^{O((\log \log n)^{d'})}  \leq 2^{o(\log n)} = n^{o(1)}. \]
% Thus, we get that for sufficiently large $n$ \anote{check this!}, the answer size is sub-polynomial!
% \anote{Note that it seems important that $d$ is a constant here. If we
%   used the original introspection game rather than the one based on de
%   la Salle, this would not necessarily be the case since $d$ comes
%   from the number of rounds of parallel repetition necessary.}
\item \textbf{Completeness and soundness:} these
  follow by repeatedly applying the corresponding items from \Cref{thm:intro-gap}.
\end{enumerate}
\end{proof}

\newcommand{\lhalt}{L_{\textsc{Halt}}}
\begin{definition}
  The language $\lhalt$ consists of all strings $x$ such that $x$ is a
  description of a Turing machine $M_x$ that halts when run on an
  empty input tape. This language is complete for $\RE$.
\end{definition}

\begin{theorem}
\label{thm:mip*-main}
There is an $n$-indexed $\MIP^*[\poly(n), \poly(n), \poly(n), \poly(n)]$ protocol $(GenG^{\fullnote}, N^{\fullnote})$ to decide $\lhalt$.
\end{theorem}
\begin{proof}
This is the main result of \cite{JNVWY20}.
\end{proof}

\begin{theorem}
There is an $n$-indexed $\MIP^*[\poly\log(n), \poly\log(n), \poly(n), \poly(n)]$ protocol $(GenG^{\hf},N^{\hf})$ to decide $\lhalt$.
\label{thm:mip*-polylog}
\end{theorem}
\begin{proof}
We derive this by applying \Cref{thm:gapped-compress} to \Cref{thm:mip*-main}. In particular, the description of $GenG^\hf$ and $N^\hf$ is as follows:
\begin{itemize}
\item $GenG^\hf$ on input $z$ firstly runs $GenG^{\fullnote}(z)$ to obtain (the description of) a game family $(S^{z,\fullnote},D^{z,\fullnote})$. Then it applies the procedure GappedCompress from \Cref{thm:gapped-compress} to $(S^{z,\fullnote},D^{z,\fullnote})$ in order to obtain a game family $(S^{z,\hf},D^{z,\hf},Q^{z,\hf})$, and outputs $(S^{z,\hf},D^{z,\hf})$.
\item $N^\hf$ on input $z$ computes and outputs $N^{\fullnote}(z)$.
\end{itemize}

Note that $GenG^\hf$ runs in time $O(|D^{z,\fullnote}| + |S^{z,\fullnote}|) = \poly(|z|)$, because each of $|D^{z,\fullnote}|, |S^{z,\fullnote}|$ is upper-bounded by $\poly(|z|)$ (they were generated by $GenG^{\fullnote}$ in $\poly(|z|)$ time). $N^\hf$ runs in the time it takes to run $N^{\fullnote}$.

It remains to show that the parameter bounds are as stated: $q^\hf(n) = \poly\log(n), a^\hf(n) = \poly\log(n), t_s^\hf(n) = \poly(n), t_d^\hf(n) = \poly(n)$. This follows by noting that $q^{\fullnote}(n) = \poly(n), t_d^{\fullnote}(n) = \poly(n), |D^{\fullnote}| = \poly(n)$, in addition to $t_s^{PB}(q) = \poly\log(q), t_d^{PB}(q) = \poly(q)$ for all $q \in \mathbb{N}$, and substituting these quantities into \Cref{thm:gapped-compress}.

\end{proof}

\begin{theorem}
There is an $n$-indexed $\MIP^*[O(1), \poly\log(n), O(1), \poly(n)]$ protocol $(GenG^{\hp},N^{\hp})$ to decide $\lhalt$.
\end{theorem}
\begin{proof}
This follows in a similar fashion to the proof of \Cref{thm:mip*-polylog} by applying \Cref{thm:gapped-hyper} to \Cref{thm:mip*-polylog}.
\end{proof}

%
%
%
%\begin{theorem}
%  $\lhalt \subseteq \MIP^*[Q_0, \poly(n), \poly(n), \poly(n)]$.
%  \label{thm:gapped-re}
%\end{theorem}
%\begin{proof}
%  The main result of \cite{JNVWY20} is $\lhalt \subseteq
%  \MIP^*[\poly(n), \poly(n), \poly(n), \poly(n)]$. In fact, the
%  protocol given there is also an $n$-indexed $\MIP^*$ protocol, in the
%  sense of \Cref{def:nindexed-mipstar}. Given an input $z$, we first
%  compute the $n$-indexed sampler and decider $S^{z,\hf}, D^{z,\hf}$
%  for the corresponding protocol, as well as the index $N_z =
%  \poly(|z|)$.  We then run the procedure
%  $\mathrm{GappedHypercompress}$ to compute an $n$-indexed sampler and
%  decider $S^{z, \hp}, D^{z, \hp}$. On index $N$, the game generated
%  by this sampler and decider has value $1$ if $z \in \lhalt$, and
%  value $\leq 1/2$ if $z \not\in \lhalt$. Morevoer its question size
%  is constant, and all other parameters are polynomially bounded.
%\end{proof}

\section{Bounds from Kolmogorov complexity}
\label{sec:lower-bounds}
\paragraph{Describing the complexity of an $\MIP^*$ protocol}
Every $\MIP^*$ protocol (see Definition \ref{def:mip*-protocol} for a more complete definition of an $\MIP^*$ protocol) is specified by a pair of Turing machines $(S, D)$, called the sampler and decider, respectively. In this section, we define the following resource bounds for an $\MIP^*$ protocol, all of which are required to be upper bounded by polynomials in $n$ (the length of the input instance):
\begin{itemize}
    \item $r(n)$, the bit length of the random seed used by $S$ (more precisely, an upper bound on the number of bits of its random input which $S$ reads).
    \item $q(n)$, the bit length of a question to a single player (by padding we assume that both players always receive questions of the same bit length for all instances of size $n$). Note that, if the question distribution is free (and uniform), we can assume that $2q(n) = r(n)$, but this is not necessarily true when the question distribution is correlated.
    \item $a(n)$, the bit length of an answer from a single player (by padding again we assume that all answers have equal bit length).
    \item $t(n)$, an upper bound on both the runtime of the sampler $S$ given a random seed, and the runtime of the decider $D$ given a pair of questions and a pair of answers.
    \end{itemize}
      
\begin{remark}
Here, we define the parameters $r(n)$ and $q(n)$ slightly differently from the way that we define them in \Cref{sec:prelims} (in particular, in \Cref{def:gapped-game-family} and \Cref{def:nindexed-mipstar}). In \Cref{sec:prelims}, we identify $r(n)$ and $q(n)$ (taking $q(n)$ to be an upper bound both on the question length and the length of the random seed which $S$ reads), while in this section we will use independent parameters to track these two quantities. The reason for the difference is that identifying $q(n)$ and $r(n)$ makes for simpler notation when we are proving upper bounds, since all of the protocols we construct have similar $q(n)$ and $r(n)$; however, when we are proving lower bounds, we would like our model to be as general as possible, and so we would like to consider protocols where $r(n)$ may be much larger than $q(n)$.

We also merge $t_s(n)$ (sampler runtime) and $t_d(n)$ (decider runtime) into a single parameter $t(n)$ in this section, since it will not be important to any of our bounds to have a stricter bound on either runtime than a polynomial in $n$.
\end{remark}

\begin{lemma}
\label{lem:mip*-to-tm}
Suppose $L \in \MIP^*[r(n), q(n), a(n), t(n)]$. Then there is a constant $\delta$ such that $L \in \sf{DTIME}[h(n)] / 2^{g(n)}$, with
\[ h(n) = 2^{r(n) + 2(q(n) + a(n))} \cdot t(n), \quad g(n) = (2q(n)+\delta +1) \cdot 2^{2(q(n) + a(n))}.
\]

Alternatively, if $L \in \MIP_0^*[r(n), q(n), a(n), t(n)]$, then $L \in \sf{DTIME}[h(n)] / 2^{g(n)}$ with
\[ h(n) = 2^{r(n) + 2(q(n) + a(n))} \cdot t(n), \quad g(n) = 2^{2(q(n) + a(n))}.
\]
\end{lemma}

\begin{proof}
The essential structure of the proof is as follows. Suppose $L \in \MIP^*[r(n), q(n), a(n), t(n)]$; then, for any instance $x$, there is a two-player entangled game $G_x$ with parameters $r(|x|), q(|x|), a(|x|), t(|x|)$ such that the value of $G_x$ is 1 if $x \in L$ and $\leq \frac{1}{2}$ if $x \notin L$. We construct a deterministic Turing machine which decides whether $x \in L$ after reading some advice. The advice will simply contain a single bit for each possible $\MIP^*(2)$ game with parameters $r(|x|), q(|x|), a(|x|), t(|x|)$, and that single bit will say whether the game in question has value $=1$ or $\leq \frac{1}{2}$ (or, in the zero-gap case, $=1$ or $<1$). The deterministic Turing machine can then look up the value of the game which decides $L$ instead of playing it with two entangled provers to approximate its value.

To specify a two-player game, it suffices to fill out the following table:
\begin{center}
\begin{tabular}{ c c c }
 $q_A, q_B, a_A, a_B$ & Probability of generating $(q_A, q_B)$  & $(a_A, a_B)$ accepted answer pair, given $(q_A, q_B)$? \\
 \hline 
 \dots & \dots & \dots
\end{tabular}
\end{center}
This table has $(2^{q(n)})^2 (2^{a(n)})^2$ rows. To calculate the number of bits needed to store each row, we must specify an accuracy to which the probabilities must be represented.

In the gapped case, the gap in the game value between the YES and NO case is a constant, so it suffices to represent the probability distribution over question pairs $(q_A, q_B)$ up to an additive constant error in total variational distance. To achieve this, it suffices in turn to store each probability up to accuracy $O(2^{-2q(n)})$. Thus, we may represent the probabilities as rational numbers specified by $2q(n) + \delta$ bits for some constant $\delta$ depending on the game value gap. Now, in total, each row stores $2q(n) + \delta$ bits for the probability, as well as a bit indicating whether the question-answer pair was accepted. Therefore, it takes 
\[ g(n) := (2q(n) + \delta +1) \times 2^{2(q(n) + a(n))} \]
bits to specify a game table. It follows that there are at most 
\[ G(n) := 2^{g(n)} = 2^{(2q(n) + \delta +1) \cdot 2^{2(q(n) + a(n)) }} \]
possible games.

In the gapless case where the game value is either 1 or strictly less
than 1, we can assume without loss of generality that the question
distribution is uniform over its support, as discussed in the preliminaries. Therefore, in the gapless case, we do not need the middle column, and there are at most
\[ G(n) = 2^{2^{2(q(n) + a(n)) }} \]
possible games.

Now, suppose $L$ is a language in $\MIP^*[r(n), q(n), a(n), t(n)]$. We specify a deterministic Turing machine $M$ which decides $L$ and which takes $G(n)$ bits of advice. The advice string for input length $n$ consists of one bit for every possible game corresponding to a $\MIP^*[r(n), q(n), a(n), t(n)]$ protocol; this single bit specifies whether or not the entangled value of that game is greater than 2/3 (or whether it is equal to 1, if we're in the gapless world). The Turing machine $M$ then does the following:
\begin{enumerate}
\item	Given instance $x$, compute the game table $T$ for the $\MIP^*$ game that would decide $L$. This can be done by the following procedure.
\begin{enumerate}
    \item First, compute the table of question pairs (and approximate probabilities, in the gapped case). This can be done by looping through all possible random seeds of length $r(n)$, and for each seed $s$, computing the question pair $q_A, q_B$ (as the output of the $\MIP^*$ sampler $S$ on seed $s$) and incrementing its associated count. This takes time $O(t(n) \cdot 2^{r(n)})$.
    \item Next, for each generated question pair $q_A, q_B$, loop through all possible answer pairs $a_A, a_B$, and run the decider on $q_A, q_B, a_A, a_B$. Store this information in the appropriate row in the table.
\end{enumerate}  
Part (a) takes time $O(t(n) \cdot 2^{r(n)})$. Part (b) takes time equal to $t(n)$ multiplied by the number of question and answer pairs. The number of question and answer pairs is at most $2^{2q(n)} \cdot 2^{2a(n)}$. So in total, the runtime to generate the game table is $O( t(n) \cdot 2^{r(n) + 2(q(n) + a(n))})$.

\item Look up the game table $T$ in the advice string. Decide to output YES or NO according to whether the advice string lists 1 or 0 for $T$.
\end{enumerate}
In the gapped case, the game table $T$ which the procedure above computes specifies a game that has a value that is $\eps(\delta)$-close to the value of the $\MIP^*$ game which decides $L$, where $\eps(\delta)$ is a constant that depends on the precision constant $\delta$. Therefore, choosing $\delta$ to be small enough, and given that $\MIP^*$ games have constant gaps between YES and NO instances, we can ensure that $M$ correctly decides $L$.

In the gapless case, the game table $T$ specifies a game that, with certainty, has a value that is 1 if and only if the value of the $\MIP^*$ game that decides $L$ is 1.

Moreover, the runtime and advice length of $M$ are as claimed in the theorem statement.
\end{proof}

\begin{lemma}
\label{lem:eexp-lower-bound}
$\sf{EEXP} \not\subseteq \sf{DTIME}[2^{k n^k}] / \eps 2^{cn}$, for any constant $k \geq 1$ and any constants $c, \eps$ such that $c + \eps < 1$.
\end{lemma}

\begin{proof}
The following argument is based on the one found in Section 2 of \cite{hm95}.

We firstly define some notation. For any $n \in \mathbb{N}$, define the following:
\begin{enumerate}
\item Let $m := 2^{cn}$, for notational convenience.
\item Let $C$ be some constant to be decided later.
\item Let $\tau_n$ be the lexicographically first string in $\{0,1\}^{m}$ that is not in $K[m - 2, C \cdot 2^{2k (\frac{1}{c} \log m)^k}]$. (See Definition \ref{def:time-bounded-kolmogorov} for the definition of $K$. Note that, in terms of $n$, this evaluates to an advice bound of $2^{cn} - 2$ and a running time bound of $C \cdot 2^{2k n^k}$.) This string always exists by counting: there are $2^m$ strings in $\{0,1\}^m$, and only $2^{m-1}$ strings in $\bigcup_{i=1}^{m-2} \{0,1\}^{i}$.
\item If $\tau_n = w_1 \| w_2 \| \dots \| w_{\ell}$ for some strings $w_1, \dots, w_{\ell}$ of length $n$, we say that the strings $w_1, \dots, w_{\ell}$ are the `$n$-blocks' of $\tau_n$. Let $A_n$ be the set $\{w_i : w_i \text{ is an $n$-block of $\tau_n$}\}.$ Let $A$ (a language) be the union $\bigcup_{n \in \mathbb{N}} A_n$.
\end{enumerate}

\begin{lemma}
$A \subseteq \sf{EEXP}$.	
\end{lemma}
\begin{proof}
We describe a Turing machine that decides $A$. The machine does the following on input $x$ of length $n$:
\begin{enumerate}
\item Write down $\tau_n$. To do this:
\begin{enumerate}
\item Run $T$ (the universal simulator of \Cref{def:time-bounded-kolmogorov}) on all strings of length $\leq 2^{cn} - 2$ for $2^{2kn^k}$ steps, and lexicographically list all the strings that result.
\item Find the first string in $\{0,1\}^{2^{cn}}$ that is not in the list from (a). Set the result to be $\tau_n$.
\end{enumerate}
\item Compare $x$ with every $n$-block of $\tau_n$ to decide whether or not $x$ is in $A_n$.
\end{enumerate}
This procedure takes time $2^{2^{cn}} \cdot 2^{2 k n^k} + O(2^{cn})$, which is doubly exponential in $n$.
\end{proof}

\begin{lemma}
Assuming $A \subseteq \sf{DTIME}(2^{k n^k}) / \eps 2^{cn}$, we can generate $\tau_n$ with $(c+\eps)2^{cn}$ advice using $C \cdot 2^{2k n^k}$ time for some constant $C$.
\label{lem:advice-assumption}
\end{lemma}
\begin{proof}
Fix an $n$. Let $M$ be the $\sf{DTIME}(2^{k n^k}) / \eps 2^{cn}$ machine that decides $A$. $M$ takes advice $a$ with $|a| = \eps 2^{cn}$, and can in time $O(2^{k n^k})$ decide membership in $A$.

Let $LEX(A_n)$ be the lexicographically ordered list of all the strings in $A_n$. To reconstruct $\tau_n$, we need the following advice:
\begin{itemize}
\item $a$, the advice given to $M$,
\item $z$, a list for all $i$ of the position in $LEX(A_n)$ of the $i$th $n$-block of $\tau_n$.
\end{itemize}
The length of $a$ is $\eps 2^{cn}$. The length of $z$ is
\[ |z| = \frac{|\tau_n|}{n} \cdot \log(|A_n|) \leq \frac{2^{cn}}{n} \cdot (cn - n) \leq c 2^{cn}.\]
Therefore, the total amount of advice we need is $(c+\eps)2^{cn}$.

In order to reconstruct $\tau_n$, we will do the following:
\begin{enumerate}
\item Generate $LEX(A_n)$ by lexicographically iterating through all strings of length $n$ and running $M$ on each one.
\item For each $i \in [\frac{|\tau_n|}{n}]$, read the appropriate index in $z$, and copy out the correct $n$-block from $A_n$.
\end{enumerate}
This procedure takes time at most $2^n \cdot 2^{k n^k} + O(2^{cn})$. Using the assumption $k \geq 1$, we can write this as $O(2^{2k n^k})$.
\end{proof}

Given that $c+\eps < 1$, Lemma \ref{lem:advice-assumption} is a contradiction with the definition of $\tau_n$, since $(c+\eps)2^{cn} < 2^{cn}-2$ for sufficiently large $n$. We conclude that $\sf{EEXP} \not\subseteq \sf{DTIME}[2^{k n^k}] / \eps 2^{cn}$.
\end{proof}

\begin{theorem}
\label{thm:lower-bound}
$\sf{EEXP} \not\subseteq \MIP^*[r(n), q(n),a(n),t(n)]$ if there are constants $n_0$ and $\gamma < \frac{1}{2}$ such that $q(n) + a(n) \leq \gamma\log(n)$ for all $n > n_0$, and $t(n), r(n) = \poly(n)$.

The same statement holds with $\MIP^*[r(n), q(n),a(n),t(n)]$ replaced by $\MIP^*_0[r(n), q(n),a(n),t(n)]$.
\end{theorem}

\begin{remark}
Note that $q(n)$ and $a(n)$ are bounds on the question and answer length for a single prover; therefore, the total communication during an $\MIP^*$ protocol can be upper bounded by $2(q(n) + a(n))$. This may make the $\gamma < 1/2$ requirement seem more natural.
\end{remark}
\begin{proof}
    We present the proof for the gapped case; the proof in the gapless case is entirely analogous.
    
    Suppose that $\sf{EEXP} \subseteq \MIP^*[r(n), q(n), a(n), t(n)]$ with $r,q,a,t$ as in the theorem statement. We know from Lemma \ref{lem:mip*-to-tm} that
    
    \[\MIP^*[r(n), q(n), a(n), t(n)] \subseteq \sf{DTIME}[t(n) \cdot 2^{r(n) + 2(q(n)+a(n))}] / 2^{g(n)},\]
    where $g(n) = (2q(n) + \delta + 1) \cdot 2^{2(q(n) + a(n))}$ for some constant $\delta$. 
    
    Note that, if $t(n), r(n) = \poly(n)$, then there exists some constant $k$ such that $t(n) \cdot 2^{r(n) + 2(q(n)+a(n))} \leq 2^{k n^k}$ in the asymptotic limit. From \Cref{lem:eexp-lower-bound}, we know that $\sf{EEXP} \not\subseteq \sf{DTIME}[2^{k n^k}] / \eps 2^{c n}$ for any constant $k$ and any constants $c, \gamma$ such that $\eps + c < 1$. So if $2^{g(n)} \leq \eps 2^{c n}$ in the asymptotic limit for \emph{any} constants $\eps, c$ such that $\eps + c < 1$, we have a contradiction, since then we would have the chain of containments
    \[\sf{EEXP} \subseteq \MIP^*[r(n), q(n), a(n), t(n)] \subseteq \sf{DTIME}[2^{k n^k}] / \eps 2^{c n}.\]
    Taking logarithms of both sides, we conclude that it cannot be the case that
\[(2q(n) + \delta + 1) \cdot 2^{2(q(n) + a(n))} \leq c n + \log \eps.\]
Substituting $q(n) + a(n) \leq \gamma \log n$, we have that the left-hand-side is at most $(2\gamma \log (n) + \delta + 1)  n^{2\gamma}$.
    So when $\gamma < 1/2$, we have a contradiction, because we can choose constants $c$ and $\eps$ such that $2\gamma < c < 1$ and $\eps < 1 - c$, and in the asymptotic limit the right-hand-side will dominate the left-hand-side with these choices for $c$ and $\eps$ substituted in.
\end{proof}

\begin{remark}
Observe that \Cref{thm:lower-bound} does not imply that $\MIP^*$ with $q(n) + a(n) \leq \gamma \log(n)$ for $\gamma < \frac{1}{2}$ is \emph{decidable.} Indeed, one may construct undecidable languages in this class by a padding argument. For example, consider any language $L$ in $\MIP^*$ with $q(n) = \poly(n)$ and $a(n) = \poly(n)$; then consider $L' := \{ x \#^{2^{2^n}} : x \in L\}$. $L'$ is in $\MIP^*$ with $q(n) = O(\log\log n)$ and $a(n) = O(\log \log n)$, but $L'$ is still undecidable. However, \Cref{thm:lower-bound} shows that no language that $\sf{EEXP}$-hard under polynomial-time reductions---the binary halting problem, for instance---can be contained in the class $\MIP^*$ if $q(n) + a(n) < \frac{1}{2} \log n$. In particular, this means that $\MIP^*$ with $q(n) + a(n) < \frac{1}{2} \log n$ cannot be equal to $\RE$.
\end{remark}

\ifnames

\section{Acknowledgements}

We thank Aram Harrow, Hamoon Mousavi, Chris Umans, and Ryan Williams
for helpful conversations. TZ was supported in part by an Akamai Presidential Fellowship.

\else
\fi

\bibliographystyle{myhalpha}
\bibliography{bellqma2}

\notesendofpaper
\end{document}